\documentclass[10pt,english]{IEEEtran}
\usepackage[T1]{fontenc}
\usepackage[latin9]{inputenc}
\usepackage{array}
\usepackage{verbatim}
\usepackage{amsmath}
\usepackage{amssymb}
\usepackage{graphicx}
\usepackage{esint}

\makeatletter

\providecommand{\tabularnewline}{\\}

\makeatother

\usepackage{amsthm}\usepackage{dsfont}\usepackage{array}\usepackage{mathrsfs}\usepackage{cite}\usepackage{comment}\usepackage{mathrsfs}\usepackage{hyperref}

\makeatother

\usepackage{babel}
\usepackage{color}

\usepackage{babel}
\usepackage{sansmath}

\makeatother

\usepackage{babel}
\begin{document}
\bibliographystyle{IEEEtran}

\title{Backing off from Infinity: \\ Performance Bounds via Concentration
of Spectral Measure for Random MIMO Channels}

\author{Yuxin Chen, Andrea J. Goldsmith, and Yonina C. Eldar%
\thanks{Y. Chen is with the Department of Electrical Engineering and the Department
of Statistics, Stanford University, Stanford, CA 94305, USA (email:
yxchen@stanford.edu). 

A. J. Goldsmith is with the Department of Electrical Engineering,
Stanford University, Stanford, CA 94305, USA (email: andrea@ee.stanford.edu). 

Y. C. Eldar is with the Department of Electrical Engineering, Technion,
Israel Institute of Technology Haifa, Israel 32000 (email: yonina@ee.technion.ac.il). 

This work was supported in part by the NSF under grant CCF-0939370
and CIS-1320628, the AFOSR under MURI Grant FA9550-09-1-0643, and
BSF Transformative Science Grant 2010505. Manuscript date: \today.%
}}
\maketitle
\begin{abstract}
The performance analysis of random vector channels, particularly multiple-input-multiple-output
(MIMO) channels, has largely been established in the asymptotic regime
of large channel dimensions, due to the analytical intractability
of characterizing the exact distribution of the objective performance
metrics. This paper exposes a new non-asymptotic framework that allows
the characterization of many canonical MIMO system performance metrics
to within a narrow interval under moderate-to-large channel dimensionality,
provided that these metrics can be expressed as a separable function
of the singular values of the matrix. The effectiveness of our framework
is illustrated through two canonical examples. Specifically, we characterize
the mutual information and power offset of random MIMO channels, as
well as the minimum mean squared estimation error of MIMO channel
inputs from the channel outputs. Our results lead to simple, informative,
and reasonably accurate control of various performance metrics in
the finite-dimensional regime, as corroborated by the numerical simulations.
Our analysis framework is established via the concentration of spectral
measure phenomenon for random matrices uncovered by Guionnet and Zeitouni,
which arises in a variety of random matrix ensembles irrespective
of the precise distributions of the matrix entries.
\end{abstract}
\theoremstyle{plain}\newtheorem{lem}{\textbf{Lemma}}\newtheorem{theorem}{\textbf{Theorem}}\newtheorem{corollary}{\textbf{Corollary}}\newtheorem{prop}{\textbf{Proposition}}\newtheorem{fct}{\textbf{Fact}}\newtheorem{remark}{\textbf{Remark}}

\theoremstyle{definition}\newtheorem{definition}{\textbf{Definition}}\newtheorem{example}{\textbf{Example}}


\section{Introduction}

The past decade has witnessed an explosion of developments in multi-dimensional
vector channels \cite{ElGamalKim2012}, particularly multiple-input-multiple-output
(MIMO) channels. The exploitation of multiple (possibly correlated)
dimensions provides various benefits in wireless communication and
signal processing systems, including channel capacity gain, improved
energy efficiency, and enhanced robustness against noise and channel
variation.

Although many fundamental MIMO system performance metrics can be evaluated
via the precise spectral distributions of finite-dimensional MIMO
channels (e.g. channel capacity \cite{Tel1999,GolJinVis2003}, minimum
mean square error (MMSE) estimates of vector channel inputs from the
channel outputs \cite{KaiSayHas2000}, power offset \cite{LozanoTuilinoVerdu2005},
sampled capacity loss \cite{ChenEldarGoldsmith2013Minimax}), this
approach often results in prohibitive analytical and computational
complexity in characterizing the probability distributions and confidence
intervals of these MIMO system metrics. In order to obtain more informative
analytical insights into the MIMO system performance metrics, a large
number of works (e.g. \cite{tse1999linear,ChuahTseKahnValenzuela2002,rapajic2000information,hanly2001resource,lozano2002capacity,LozanoTuilinoVerdu2005,tulino2005impact,lozano2003multiple,muller2002random,mestre2003capacity,moustakas2003mimo,lozano2003capacity,verdu1999spectral,moustakas2007outage,dupuy2011capacity,taricco2011ergodic})
present more explicit expressions for these performance metrics with
the aid of random matrix theory. Interestingly, when the number of
input and output dimensions grow, many of the MIMO system metrics
taking the form of linear spectral statistics converge to deterministic
limits, due to various limiting laws and universality properties of
(asymptotically) large random matrices \cite{Tao2012RMT}. In fact,
the spectrum (i.e. singular-value distribution) of a random channel
matrix $\boldsymbol{H}$ tends to stabilize when the channel dimension
grows, and the limiting distribution is often universal in the sense
that it is independent of the precise distributions of the entries
of $\boldsymbol{H}$. 

\begin{comment}
The asymptotic distribution and deviation of several metrics have
also been reported \cite{HacKhoLouNajPas2008,kumar2009asymptotic,taricco2008asymptotic,kammoun2009central,kazakopoulos2011living}. 
\end{comment}

These asymptotic results are well suited for massive MIMO communication
systems. However, the limiting regime falls short in providing a full
picture of the phenomena arising in most practical systems which,
in general, have moderate dimensionality. While the asymptotic convergence
rates of many canonical MIMO system performance metrics have been
investigated as well, how large the channel dimension must be largely
depends on the realization of the growing matrix sequences. In this
paper, we propose an alternative method via concentration of measure
to evaluate many canonical MIMO system performance metrics for finite-dimension
channels, assuming that the target performance metrics can be transformed
into linear spectral statistics of the MIMO channel matrix. Moreover,
we show that the metrics fall within narrow intervals with high (or
even overwhelming) probability for moderate-to-high dimensional channels.

\subsection{Related Work}

Random matrix theory is one of the central topics in probability theory
with many connections to wireless communications and signal processing.
Several random matrix ensembles, such as Gaussian unitary ensembles,
admit exact characterization of their eigenvalue distributions \cite{mehta2004random}
under any channel dimensionality. These spectral distributions associated
with the finite-dimensional MIMO channels can then be used to compute
analytically the distributions and confidence intervals of the MIMO
system performance metrics (e.g. mutual information of MIMO fading
channels \cite{hassibi2002multiple,wang2004outage,chiani2003capacity,ratnarajah2003complex,ratnarajah2005complex}).
However, the computational intractability of integrating a large-dimensional
function over a finite-dimensional MIMO spectral distribution precludes
concise and informative capacity expressions even in moderate-sized
problems. For this reason, theoretical analysis based on precise eigenvalue
characterization is generally limited to small-dimensional vector
channels. 

In comparison, one of the central topics in modern random matrix theory
is to derive limiting distributions for the eigenvalues of random
matrix ensembles of interest, which often turns out to be simple and
informative. Several pertinent examples include the semi-circle law
for symmetric Wigner matrices \cite{wigner1958distribution}, the
circular law for i.i.d. matrix ensembles \cite{tao2010random,tao2008random},
and the Marchenko\textendash{}Pastur law \cite{marchenko1967distribution}
for rectangular random matrices. One remarkable feature of these asymptotic
laws is the \emph{universality phenomenon}, whereby the limiting spectral
distributions are often indifferent to the precise distribution of
each matrix entry. This phenomenon allows theoretical analysis to
accommodate a broad class of random matrix families beyond Gaussian
ensembles. See \cite{Tao2012RMT} for a beautiful and self-contained
exposition of these limiting results.

The simplicity and universality of these asymptotic laws have inspired
a large body of work in characterizing the asymptotic performance
limits of random vector channels. For instance, the limiting results
for i.i.d. Gaussian ensembles have been applied in early work \cite{Tel1999}
to establish the linear increase of MIMO capacity with the number
of antennas. This approach was then extended to accommodate separable
correlation fading models \cite{ChuahTseKahnValenzuela2002,tulino2005impact,muller2002random,moustakas2003mimo,mestre2003capacity,shin2003capacity,couillet2013analysis}
with general non-Gaussian distributions (particularly Rayleigh and
Ricean fading \cite{lozano2003multiple,LozanoTuilinoVerdu2005}).
These models admit analytically-friendly solutions, and have become
the cornerstone for many results for random MIMO channels. In addition,
several works have characterized the second-order asymptotics (often
in the form of central limit theorems) for mutual information \cite{moustakas2003mimo,HacKhoLouNajPas2008,taricco2008asymptotic},
information density \cite{hoydis2012random}, second-order coding
rates \cite{hoydis2013second,hoydis2013bounds}, diversity-multiplexing
tradeoff \cite{loyka2010finite}, etc., which uncover the asymptotic
convergence rates of various MIMO system performance metrics under
a large family of random matrix ensembles. The limiting tail of the
distribution (or large deviation) of the mutual information has also
been determined in the asymptotic regime \cite{kazakopoulos2011living}.
In addition to information theoretic metrics, other signal processing
and statistical metrics like MMSE \cite{tse1999linear,liang2007asymptotic,couillet2013signal},
multiuser efficiency for CDMA systems \cite{hanly2001resource}, optical
capacity \cite{karadimitrakis2013outage}, covariance matrix and principal
components \cite{johnstone2001distribution}, canonical correlations
\cite{johnstone2006high}, and likelihood ratio test statistics \cite{johnstone2006high},
have also been investigated via asymptotic random matrix theory. While
these asymptotic laws have been primarily applied to performance metrics
in the form of linear spectral statistics, more general performance
metrics can be approximated using the delicate method of ``deterministic
equivalents'' (e.g. \cite{couillet2011random,couillet2011deterministic}).

A recent trend in statistics is to move from asymptotic laws towards
non-asymptotic analysis of random matrices \cite{donoho2000high,Vershynin2012},
which aims at revealing statistical effects of a moderate-to-large
number of components, assuming a sufficient amount of independence
among them. One prominent effect in this context is the \emph{concentration
of spectral measure phenomenon} \cite{GuionnetZeitouni2000,el2009concentration},
which indicates that many separable functions of a matrix's singular
values (called \emph{linear spectral statistics} \cite{bai2004clt})
can be shown to fall within a narrow interval with high probability
even in the moderate-dimensional regime. This phenomenon has been
investigated in various fields such as high-dimensional statistics
\cite{donoho2000high}, statistical learning \cite{MassartPicard2007},
and compressed sensing \cite{Vershynin2012}. 

Inspired by the success of the measure concentration methods in the
statistics literature, our recent work \cite{ChenEldarGoldsmith2013Minimax}
develops a non-asymptotic approach to quantify the capacity of multi-band
channels under random sub-sampling strategies, which to our knowledge
is the first to exploit the concentration of spectral measure phenomenon
to analyze random MIMO channels. In general, the concentration of
measure phenomena are much less widely recognized and used in the
communication community than in the statistics and signal processing
community. Recent emergence of massive MIMO technologies \cite{Huh2012,rusek2013scaling,larsson2013massive,Ngo2013Energy},
which uses a large number of antennas to obtain both capacity gain
and improved radiated energy efficiency, provides a compelling application
of these methods to characterize system performance. Other network
/ distributed MIMO systems (e.g. \cite{huh2012network,zhang2009networked})
also require analyzing large-dimensional random vector channels. It
is our aim here to develop a general framework that promises new insights
into the performance of these emerging technologies under moderate-to-large
channel dimensionality.

\subsection{Contributions}

We develop a general non-asymptotic  framework for deriving performance
bounds of random vector channels or any general MIMO system, based
on the powerful concentration of spectral measure phenomenon of random
matrices as revealed by Guionnet and Zeitouni \cite{GuionnetZeitouni2000}.
Specifically, we introduce a general recipe that can be used to assess
various MIMO system performance metrics to within vanishingly small
confidence intervals, assuming that the objective metrics can be transformed
into linear spectral statistics associated with the MIMO channel.
Our framework and associated results can accommodate a large class
of probability distributions for random channel matrices, including
those with bounded support, a large class of sub-Gaussian measures,
and heavy-tailed distributions. To broaden the range of metrics that
we can accurately evaluate, we also develop a general concentration
of spectral measure inequality for the cases where only the \emph{exponential
mean} (instead of the mean) of the objective metrics can be computed. 

We demonstrate the effectiveness of our approach through two illustrative
examples: (1) mutual information of random vector channels under equal
power allocation; (2) MMSE in estimating signals transmitted over
random MIMO channels. These examples allow concise and, informative
characterizations even in the presence of moderate-to-high SNR and
moderate-to-large channel dimensionality. In contrast to a large body
of prior works that focus on first-order limits or asymptotic convergence
rate of the target performance metrics, we are able to derive full
characterization of these metrics in the non-asymptotic regime. Specifically,
we obtain narrow confidence intervals with precise order and reasonably
accurate pre-constants, which do not rely on careful choice of the
growing matrix sequence. Our numerical simulations also corroborate
that our theoretical predictions are reasonably accurate in the finite-dimensional
regime.

\subsection{Organization and Notation}

The rest of the paper is organized as follows. Section \ref{sec:Terminology}
introduces several families of probability measures investigated in
this paper. We present a general framework characterizing the concentration
of spectral measure phenomenon in Section \ref{sec:General-Template}.
In Section \ref{sec:Sample-Applications}, we illustrate the general
framework using a small sample of canonical examples. Finally, Section
\ref{sec:Conclusion} concludes the paper with a summary of our findings
and a discussion of several future directions. 

For convenience of presentation, we let $\mathbb{C}$ denote the set
of complex numbers. For any function $g(x)$, the Lipschitz norm of
$g(\cdot)$ is defined as
\begin{equation}
\left\Vert g\right\Vert _{\mathcal{L}}:=\sup_{x\neq y}\left|\frac{g\left(x\right)-g\left(y\right)}{x-y}\right|.\label{eq:LipschitzNorm}
\end{equation}
We let $\lambda_{i}(\boldsymbol{A})$, $\lambda_{\max}(\boldsymbol{A})$,
and $\lambda_{\min}(\boldsymbol{A})$ represent the $i$th largest
eigenvalue, the largest eigenvalue, and the smallest eigenvalue of
a Hermitian matrix $\boldsymbol{A}$, respectively, and use $\left\Vert \cdot\right\Vert $
to denote the operator norm (or spectral norm). The set $\left\{ 1,2,\cdots,n\right\} $
is denoted as $[n]$, and we write ${[n] \choose k}$ for the set
of all $k$-element subsets of $[n]$. For any set $\boldsymbol{s}$,
we use $\mathrm{card}(\boldsymbol{s})$ to represent the cardinality
of $\boldsymbol{s}$. Also, for any two index sets $\boldsymbol{s}$
and $\boldsymbol{t}$, we denote by $\boldsymbol{A}_{\boldsymbol{s},\boldsymbol{t}}$
the submatrix of $\boldsymbol{A}$ containing the rows at indices
in $\boldsymbol{s}$ and columns at indices in $\boldsymbol{t}$.
We use $x\in a+[b,c]$ to indicate that $x$ lies within the interval
$[a+b,a+c]$. Finally, the standard notation $f(n)=\mathcal{O}\left(g(n)\right)$
means that there exists a constant $c>0$ such that $f(n)\leq cg(n)$;
$f(n)=\Theta\left(g(n)\right)$ indicates that there are universal
constants $c_{1},c_{2}>0$ such that $c_{1}g(n)\leq f(n)\leq c_{2}g(n)$;
and $f(n)=\Omega\left(g(n)\right)$ indicates that there are universal
constants $c>0$ such that $f(n)\leq cg(n)$. Throughout this paper,
we use $\log\left(\cdot\right)$ to represent the natural logarithm.
Our notation is summarized in Table \ref{tab:Summary-of-Notation-Nonuniform}. 

\begin{table}
\caption{\label{tab:Summary-of-Notation-Nonuniform}Summary of Notation and
Parameters}

\centering{}%
\begin{tabular}{l>{\raggedright}p{0.4\textwidth}}
$\log(\cdot)$  & natural logarithm\tabularnewline
$\kappa$ & $\kappa=1$ for the real-valued case and $\kappa=2$ for the complex-valued
case\tabularnewline
$\nu_{ij},\nu$ & the standard deviation of $\boldsymbol{M}_{ij}$, and the maximum
standard deviation $\nu:=\max_{i,j}\nu_{ij}$\tabularnewline
$c_{\mathrm{ls}}$ & logarithmic Sobolev constant as defined in (\ref{eq:LogSobolevInequality})\tabularnewline
$\left\Vert \cdot\right\Vert _{\mathcal{L}}$ & Lipschitz constant of a function\tabularnewline
$D$ & maximum absolute value in the compact support of a bounded measure \tabularnewline
$\tau_{c}$, $\sigma_{c}$ & truncation threshold and the associated standard deviation defined
in (\ref{eq:TruncateMProb}) and (\ref{eq:TruncatedVar}), respectively\tabularnewline
$\rho$ & spectral radius of the deterministic matrix $\boldsymbol{R}$\tabularnewline
$[n]$ & $[n]:=\left\{ 1,2,\cdots,n\right\} $\tabularnewline
${[n] \choose k}$  & set of all $k$-element subsets of $[n]$\tabularnewline
\end{tabular}
\end{table}

\section{Three Families of Probability Measures\label{sec:Terminology}}

In this section, we define precisely three classes of probability
measures with different properties of the tails, which lay the foundation
of the analysis in this paper.

\begin{comment}
In this section, we define precisely several characteristics satisfied
by some classes of probability measures considered herein. In particular,
three important classes of distributions satisfy the three characteristics
we next define and discuss, and these classes of probability distributions
lay the foundation of the analysis in this paper.
\end{comment}

\begin{definition}[{\bf Bounded Distribution}]\label{defn:Bounded}We
say that a random variable $X\in\mathbb{C}$ is bounded by $D$ if
\begin{equation}
\mathbb{P}\left(\left|X\right|\leq D\right)=1.\label{eq:BoundedMeasure}
\end{equation}
 \end{definition}

The class of probability distributions that have bounded support subsumes
as special cases a broad class of distributions encountered in practice
(e.g. the Bernoulli distribution and uniform distribution). Another
class of probability measures that possess light tails (i.e. sub-Gaussian
tails) is defined as follows. 

\begin{definition}[{\bf Logarithmic Sobolev Measure}]\label{defn:LogSobolev}We
say that a random variable $X\in\mathbb{C}$ with cumulative probability
distribution (CDF) $F\left(\cdot\right)$ satisfies the logarithmic
Sobolev inequality (LSI) with uniform constant $c_{\mathrm{ls}}$
if, for \emph{any differentiable} function $h$, one has
\begin{equation}
{\displaystyle \int}h^{2}(x)\log\left(\frac{h^{2}(x)}{\int h^{2}(z)\mathrm{d}F(z)}\right)\mathrm{d}F(x)\leq2c_{\mathrm{ls}}\int\left|h'(x)\right|^{2}\mathrm{d}F(x).\label{eq:LogSobolevInequality}
\end{equation}
 \end{definition}

\begin{remark}One of the most popular techniques in demonstrating
measure concentration is the ``entropy method'' (see, e.g. \cite{boucheron2013concentration,raginsky2012concentration}),
which hinges upon establishing inequalities of the form
\begin{equation}
\mathrm{KL}\left(\mathbb{P}_{X}^{t}\|\mathbb{P}_{X}\right)\leq O\left(t^{2}\right)\label{eq:EntropyInequality}
\end{equation}
for some measure $\mathbb{P}_{X}$ and its tilted measure $\mathbb{P}_{X}^{t}$,
where $\mathrm{KL}\left(\cdot\|\cdot\right)$ denotes the Kullback\textendash{}Leibler
divergence. See \cite[Chapter 3]{raginsky2012concentration} for detailed
definitions and derivation. It turns out that for a probability measure
satisfying the LSI, one can pick the function $h(\cdot)$ in (\ref{eq:LogSobolevInequality})
in some appropriate fashion to yield the entropy-type inequality (\ref{eq:EntropyInequality}).
In fact, the LSI has been well recognized as one of the most fundamental
criteria to demonstrate exponentially sharp concentration for various
metrics. \end{remark}

\begin{remark}While the measure satisfying the LSI necessarily exhibits
sub-Gaussian tails (e.g. \cite{boucheron2013concentration,raginsky2012concentration}),
many concentration results under log-Sobolev measures cannot be extended
to general sub-Gaussian distributions (e.g. for bounded measures the
concentration results have only been shown for convex functions instead
of general Lipschitz functions). \end{remark}

A number of measures satisfying the LSI have been discussed in the
expository paper \cite{ledoux1999concentration}. One of the most
important examples is the standard Gaussian distribution, which satisfies
the LSI with logarithmic Sobolev constant $c_{\mathrm{LS}}=1$. In
many situations, the probability measures obeying the LSI exhibit
very sharp concentration of spectral measure phenomenon (typically
sharper than general bounded measures). 

While prior works on measure concentration focus primarily on measures
with bounded support or measures with sub-Gaussian tails, it is also
of interest to accommodate a more general class of distributions (e.g.
heavy-tailed distributions). For this purpose, we introduce sub-exponential
distributions and heavy-tailed distributions as follows.

\begin{definition}[{\bf Sub-Exponential Distribution}]\label{defn:Subexponential}A
random variable $X$ is said to be sub-exponentially distributed with
parameter $\lambda$ if it satisfies
\begin{equation}
\mathbb{P}\left(\left|X\right|>x\right)\leq c_{0}e^{-\lambda x},\quad\forall x>0\label{eq:SubExponential}
\end{equation}
for some absolute constant $c_{0}>0$.\end{definition}

\begin{definition}[{\bf Heavy-tailed Distribution}]\label{defn:HeavyTail}A
random variable $X$ is said to be heavy-tailed distributed if 
\begin{equation}
\lim_{x\rightarrow\infty}e^{-\lambda x}\mathbb{P}\left(\left|X\right|>x\right)=\infty,\quad\forall\lambda>0.
\end{equation}
\end{definition}

That said, sub-exponential (resp. heavy-tailed) distributions are
those probability measures whose tails are lighter (resp. heavier)
than exponential distributions. Commonly encountered heavy-tailed
distributions include log-normal and power-law distributions. 

To facilitate analysis, for both sub-exponential and heavy-tailed
distributions, we define the following two quantities $\tau_{c}$
and $\sigma_{c}$ with respect to $X$ for some sequence $c(n)$ such
that

\begin{equation}
\mathbb{P}\left(\left|X\right|>\tau_{c}\right)\leq\frac{1}{mn^{c(n)+1}},\label{eq:TruncateMProb}
\end{equation}
and
\begin{equation}
\sigma_{c}=\sqrt{\mathbb{E}\left[\left|X\right|^{2}{\bf 1}{}_{\left\{ \left|X\right|<\tau_{c}\right\} }\right]}.\label{eq:TruncatedVar}
\end{equation}
In short, $\tau_{c}$ represents a truncation threshold such that
the truncated $X$ coincides with the true $X$ with high probability,
and $\sigma_{c}$ denotes the standard deviation of the truncated
$X$. The idea is to study $X{\bf 1}{}_{\left\{ \left|X\right|<\tau_{c}\right\} }$
instead of $X$ in the analysis, since the truncated value is bounded
and exhibits light tails. For instance, if $X$ is exponentially distributed
such that $\mathbb{P}\left(\left|X\right|>x\right)=e^{-x}$ and $c(n)=\log n$,
then this yields that
\[
\tau_{c}=\log^{2}n+\log\left(mn\right).
\]
Typically, one would like to pick $c(n)$ such that $\tau_{c}$ becomes
a small function of $m$ and $n$ obeying $\sigma_{c}^{2}\approx\mathbb{E}[\left|X\right|^{2}]$.

\section{Concentration of Spectral Measure in Large Random Vector Channels\label{sec:General-Template}}

We present a general mathematical framework that facilitates the analysis
of random vector channels. The proposed approach is established upon
the concentration of spectral measure phenomenon derived by Guionnet
and Zeitouni \cite{GuionnetZeitouni2000}, which is a consequence
of Talagrand's concentration inequalities \cite{talagrand1995concentration}.
While the adaptation of such general concentration results to our
settings requires moderate mathematical effort, it leads to a very
effective framework to assess the fluctuation of various MIMO system
performance metrics. We will provide a few canonical examples in Section
\ref{sec:Sample-Applications} to illustrate the power of this framework.

Consider a random matrix $\boldsymbol{M}=[\boldsymbol{M}_{ij}]_{1\leq i\leq n,1\leq j\leq m}$,
where $\boldsymbol{M}_{ij}$'s are assumed to be independently distributed.
We use $\boldsymbol{M}_{ij}^{\mathrm{R}}$ and $\boldsymbol{M}_{ij}^{\mathrm{I}}$
to represent respectively the real and imaginary parts of $\boldsymbol{M}_{ij}$,
which are also generated independently. Note, however, that $\boldsymbol{M}_{ij}$'s
are \emph{not} necessarily drawn from identical distributions. Set
the numerical value 
\begin{equation}
\kappa:=\begin{cases}
1,\quad & \text{in the real-valued case},\\
2,\quad & \text{in the complex-valued case}.
\end{cases}\label{eq:KappaRealComplex}
\end{equation}
We further assume that all entries have matching two moments such
that for all $i$ and $j$:
\begin{align}
 & \quad\quad\mathbb{E}\left[\boldsymbol{M}_{ij}\right]=0,\quad\quad\quad{\bf Var}\left(\boldsymbol{M}_{ij}^{\mathrm{R}}\right)=\nu_{ij}^{2},\label{eq:ZeroMeanUnitVar}\\
 & {\bf Var}\left(\boldsymbol{M}_{ij}^{\mathrm{I}}\right)=\begin{cases}
\nu_{ij}^{2},\quad & \text{in the complex-valued case},\\
0, & \text{in the real-valued case},
\end{cases}
\end{align}
where $\nu_{ij}^{2}$ ($\nu_{ij}\geq0$) are uniformly bounded by
\begin{equation}
\max_{i,j}\left|\nu_{ij}\right|\leq\nu.\label{eq:UniformVar}
\end{equation}

In this paper, we focus on the class of MIMO system performance metrics
that can be transformed into additively separable functions of the
eigenvalues of the random channel matrix (called \emph{linear spectral
statistics} \cite{bai2004clt}), i.e. the type of metrics taking the
form 
\begin{equation}
\sum_{i=1}^{n}f\left(\lambda_{i}\left(\frac{1}{n}\boldsymbol{M}\boldsymbol{R}\boldsymbol{R}^{*}\boldsymbol{M}^{*}\right)\right)\label{eq:LinearSpectralStats}
\end{equation}
for some function $f\left(\cdot\right)$ and some deterministic matrix
$\boldsymbol{R}\in\mathbb{C}^{m\times m}$, where $\lambda_{i}\left(\boldsymbol{X}\right)$
denotes the $i$th eigenvalue of a matrix $\boldsymbol{X}$. As can
been seen, many vector channel performance metrics (e.g. MIMO mutual
information, MMSE, sampled channel capacity loss) can all be transformed
into certain forms of linear spectral statistics. We will characterize
in Proposition \ref{thm:GeneralTemplate} and Theorem \ref{thm:Deviation-Ef}
the measure concentration of (\ref{eq:LinearSpectralStats}), which
quantifies the fluctuation of (\ref{eq:LinearSpectralStats}) in finite-dimensional
random vector channels.

\subsection{Concentration Inequalities for the Spectral Measure \label{sub:Concentration-of-Spectral}}

Many linear spectral statistics of random matrices sharply concentrate
within a narrow interval indifferent to the precise entry distributions.
Somewhat surprisingly, such spectral measure concentration is shared
by a very large class of random matrix ensembles. We provide formal
quantitative illustration of such concentration of spectral measure
phenomena below, which follows from \cite[Corollary 1.8]{GuionnetZeitouni2000}.
This behavior, while not widely used in communication and signal processing,
is a general result in probability theory derived from the Talagrand
concentration inequality (see \cite{talagrand1995concentration} for
details). 

Before proceeding to the concentration results, we define
\begin{equation}
f_{0}\left(\boldsymbol{M}\right):=\frac{1}{n}\sum_{i=1}^{\min\left(m,n\right)}f\left(\lambda_{i}\left(\frac{1}{n}\boldsymbol{M}\boldsymbol{R}\boldsymbol{R}^{*}\boldsymbol{M}^{*}\right)\right)
\end{equation}
for the sake of notational simplicity.

\begin{prop}\label{thm:GeneralTemplate}Consider a random matrix
$\boldsymbol{M}=\left[\boldsymbol{M}_{ij}\right]_{1\leq i\leq n,1\leq j\leq m}$
satisfying (\ref{eq:ZeroMeanUnitVar}) and (\ref{eq:UniformVar}).
Suppose that $\boldsymbol{R}\in\mathbb{C}^{m\times m}$ is any given
matrix and $\left\Vert \boldsymbol{R}\right\Vert \leq\rho.$ Consider
a function $f(x)$ such that $g(x):=f\left(x^{2}\right)$ ($x\in\mathbb{R}^{+}$)
is a real-valued Lipschitz function with Lipschitz constant $\left\Vert g\right\Vert _{\mathcal{L}}$
as defined in (\ref{eq:LipschitzNorm}). Let $\kappa$ and $\nu$
be defined in (\ref{eq:KappaRealComplex}) and (\ref{eq:UniformVar}),
respectively.

(a) \textbf{(Bounded Measure)} If $\boldsymbol{M}_{ij}$'s are bounded
by $D$ and $g(\cdot)$ is convex, then for any $\beta>8\sqrt{\pi}$,
\begin{equation}
\left|f_{0}\left(\boldsymbol{M}\right)-\mathbb{E}\left[f_{0}\left(\boldsymbol{M}\right)\right]\right|\leq\frac{\beta D\rho\nu\left\Vert g\right\Vert _{\mathcal{L}}}{n}\label{eq:BoundedDeviation-corollary}
\end{equation}
with probability exceeding $1-4\exp\left(-\frac{\beta^{2}}{8\kappa}\right)$.

(b) \textbf{(Logarithmic Sobolev Measure)} If the measure of $\boldsymbol{M}_{ij}$
satisfies the LSI with a uniformly bounded constant $c_{\mathrm{ls}}$,
then for any $\beta>0,$ 
\begin{equation}
\left|f_{0}\left(\boldsymbol{M}\right)-\mathbb{E}\left[f_{0}\left(\boldsymbol{M}\right)\right]\right|\leq\frac{\beta\sqrt{c_{\mathrm{ls}}}\rho\nu\left\Vert g\right\Vert _{\mathcal{L}}}{n}\label{eq:LogSobolevDeviation-Corollary}
\end{equation}
with probability exceeding $1-2\exp\left(-\frac{\beta^{2}}{\kappa}\right)$.

(c) \textbf{(Sub-Exponential and Heavy-tailed Measure)} Suppose that
$\boldsymbol{M}_{ij}$'s are independently drawn from either sub-exponential
distributions or heavy-tailed distributions and that their distributions
are symmetric about 0, with $\tau_{c}$ and $\sigma_{c}$ defined
respectively in (\ref{eq:TruncateMProb}) and (\ref{eq:TruncatedVar})
with respect to $\boldsymbol{M}_{ij}$ for some sequence $c(n)$.
If $g(\cdot)$ is convex, then
\begin{equation}
\left|f_{0}\left(\boldsymbol{M}\right)-\mathbb{E}\left[f_{0}\left(\tilde{\boldsymbol{M}}\right)\right]\right|\leq\frac{2\kappa\sqrt{c(n)\log n}\tau_{c}\sigma_{c}\rho\left\Vert g\right\Vert _{\mathcal{L}}}{n}\label{eq:HeavyTailDeviation-corollary}
\end{equation}
with probability exceeding $1-5n^{-c(n)}$, where $\tilde{\boldsymbol{M}}$
is defined such that $\tilde{\boldsymbol{M}}_{ij}:=\boldsymbol{M}_{ij}{\bf 1}{}_{\left\{ \left|\boldsymbol{M}_{ij}\right|<\tau_{c}\right\} }$.\end{prop}

\begin{proof}See Appendix \ref{sec:Proof-of-Theorem-General-Template}.\end{proof}

\begin{remark}By setting $\beta=\sqrt{\log n}$ in Proposition \ref{thm:GeneralTemplate}(a)
and (b), one can see that under either measures of bounded support
or measures obeying the LSI, 
\begin{equation}
\left|f_{0}\left(\boldsymbol{M}\right)-\mathbb{E}\left[f_{0}\left(\tilde{\boldsymbol{M}}\right)\right]\right|=\mathcal{O}\left(\frac{\sqrt{\log n}}{n}\right)\label{eq:SubGaussian}
\end{equation}
with high probability. In contrast, under heavy-tailed measures, 
\begin{equation}
\left|f_{0}\left(\boldsymbol{M}\right)-\mathbb{E}\left[f_{0}\left(\tilde{\boldsymbol{M}}\right)\right]\right|=\mathcal{O}\left(\frac{\sqrt{c(n)\tau_{c}^{2}\sigma_{c}^{2}\log n}}{n}\right)\label{eq:HeavyTail}
\end{equation}
with high probability, where $\sqrt{c(n)}\tau_{c}\sigma_{c}$ is typically
a growing function in $n$. As a result, the concentration under sub-Gaussian
distributions is sharper than that under heavy-tailed measures by
a factor of $\sqrt{c(n)}\tau_{c}\sigma_{c}$.\end{remark}

\begin{remark}The bounds derived in Proposition \ref{thm:GeneralTemplate}
scale linearly with $\nu$, which is the maximum standard deviation
of the entries of $\boldsymbol{M}$. This allows us to assess the
concentration phenomena for matrices with entries that have non-uniform
variance. \end{remark}

Proposition \ref{thm:GeneralTemplate}(a) and \ref{thm:GeneralTemplate}(b)
assert that for both measures of bounded support and a large class
of sub-Gaussian distributions, many separable functions of the spectra
of random matrices exhibit sharp concentration, assuming a sufficient
amount of independence between the entries. More remarkably, the tails
behave at worst like a Gaussian random variable with well-controlled
variance. Note that for a bounded measure, we require the objective
metric of the form (\ref{eq:LinearSpectralStats}) to satisfy certain
convexity conditions in order to guarantee concentration. In contrast,
the fluctuation of general Lipschitz functions can be well controlled
for logarithmic Sobolev measures. This agrees with the prevailing
wisdom that the standard Gaussian measure (which satisfies the LSI)
often exhibits sharper concentration than general bounded distributions
(e.g. Bernoulli measures). 

Proposition \ref{thm:GeneralTemplate}(c) demonstrates that spectral
measure concentration arises even when the tail distributions of $\boldsymbol{M}_{ij}$
are much heavier than standard Gaussian random variables, although
it might not be as sharp as for sub-Gaussian measures. This remarkable
feature comes at a price, namely, the deviation of the objective metrics
is much less controlled than for sub-Gaussian distributions. However,
this degree of concentration might still suffice for most practical
purposes. Note that the concentration result for heavy-tailed distributions
is stated in terms of the truncated version $\tilde{\boldsymbol{M}}$.
The nice feature of the truncated $\tilde{\boldsymbol{M}}$ is that
its entries are all bounded (and hence sub-Gaussian), which can often
be quantified or estimated in a more convenient fashion. Finally,
we remark that the concentration depends on the choice of the sequence
$c(n)$, which in turn affects the size of $\tau_{c}$ and $\sigma_{c}$.
We will illustrate the resulting size of confidence intervals in Section
\ref{sub:Confidence-Interval} via several examples.

\subsection{Approximation of Expected Empirical Distribution\label{sub:Convergence-Rate-Expected}}

Although Proposition \ref{thm:GeneralTemplate} ensures sharp measure
concentration of various linear spectral statistics, a more precise
characterization requires evaluating the mean value of the target
metric (\ref{eq:LinearSpectralStats}) (i.e. $\mathbb{E}\left[f_{0}\left(\boldsymbol{M}\right)\right]$).
While limiting laws often admit simple asymptotic characterization
of this mean value for a general class of metrics, whether the convergence
rate can be quantified often needs to be studied on a case-by-case
basis. In fact, this has become an extensively researched topic in
mathematics (e.g. \cite{bai1993convergence,bai2003convergence,gotze2010rate,gotze2005rate}).
In this subsection, we develop an approximation result that allows
the expected value of a broader class of metrics to be well approximated
by concise and informative expressions. 

Recall that
\begin{equation}
f_{0}\left(\boldsymbol{M}\right):=\frac{1}{n}\sum_{i=1}^{\min\left(m,n\right)}f\left(\lambda_{i}\left(\frac{1}{n}\boldsymbol{M}\boldsymbol{R}\boldsymbol{R}^{*}\boldsymbol{M}^{*}\right)\right)
\end{equation}
We consider a large class of situations where the \emph{exponential
mean} of the target metric (\ref{eq:LinearSpectralStats})
\begin{align}
\mathcal{E}\left(f\right): & =\mathbb{E}\left[\exp\left(nf_{0}\left(\boldsymbol{M}\right)\right)\right]\label{eq:DefnExponentialMean}\\
 & =\mathbb{E}\left[\exp\left(\sum_{i=1}^{\min\{m,n\}}f\left(\lambda_{i}\left(\frac{1}{n}\boldsymbol{M}\boldsymbol{R}\boldsymbol{R}^{*}\boldsymbol{M}^{*}\right)\right)\right)\right]\nonumber 
\end{align}
(instead of $\mathbb{E}\left[f_{0}\left(\boldsymbol{M}\right)\right]$)
can be approximated in a reasonably accurate manner. This is particularly
relevant when $f(\cdot)$ is a logarithmic function. For example,
this applies to log-determinant functions
\[
f_{0}\left(\boldsymbol{M}\right)=\frac{1}{n}\log\mathrm{\det}\left(\boldsymbol{I}+\frac{1}{n}\boldsymbol{M}\boldsymbol{R}\boldsymbol{R}^{*}\boldsymbol{M}^{*}\right),
\]
which are of significant interest in various applications such as
wireless communications \cite{Gold2005}, multivariate hypothesis
testing \cite{Fujikoshi2010}, etc. While $\mathbb{E}\left[\mathrm{\det}\left(\boldsymbol{I}+\frac{1}{n}\boldsymbol{M}\boldsymbol{M}^{*}\right)\right]$
often admits a simple distribution-free expression, $\mathbb{E}\left[\log\mathrm{\det}\left(\boldsymbol{I}+\frac{1}{n}\boldsymbol{M}\boldsymbol{M}^{*}\right)\right]$
is highly dependent on precise distributions of the entries of the
matrix. 

One might already notice that, by Jensen's inequality, $\log\mathcal{E}\left(f\right)$
is larger than the mean objective metric $\mathbb{E}\left[f_{0}\left(\boldsymbol{M}\right)\right]$.
Nevertheless, in many situations, these two quantities differ by only
a vanishingly small gap, which is formally demonstrated in the following
lemma. 

\begin{lem}\label{thm:ComputeMeanSubExponential}

(a) (\textbf{Sub-Exponential Tail}) Suppose that 
\[
\mathbb{P}\left(\left|f_{0}\left(\boldsymbol{M}\right)-\mathbb{E}\left[f_{0}\left(\boldsymbol{M}\right)\right]\right|>\frac{y}{n}\right)\leq c_{1}\exp\left(-c_{2}y\right)
\]
 for some values $c_{1}>0$ and $c_{2}>1$, then 
\begin{equation}
\frac{1}{n}\log\mathcal{E}\left(f\right)-\frac{\log\left(1+\frac{c_{1}}{c_{2}-1}\right)}{n}\leq\mathbb{E}\left[f_{0}\left(\boldsymbol{M}\right)\right]\leq\frac{1}{n}\log\mathcal{E}\left(f\right)\label{eq:ApproxEf_subexponential}
\end{equation}
with $\mathcal{E}\left(f\right)$ defined in (\ref{eq:DefnExponentialMean}).

(b) (\textbf{Sub-Gaussian Tail}) Suppose that
\[
\mathbb{P}\left(\left|f_{0}\left(\boldsymbol{M}\right)-\mathbb{E}\left[f_{0}\left(\boldsymbol{M}\right)\right]\right|>\frac{y}{n}\right)\leq c_{1}\exp\left(-c_{2}y^{2}\right)
\]
 for some values $c_{1}>0$ and $c_{2}>0$, then
\begin{align}
\frac{1}{n}\log\mathcal{E}\left(f\right) & \geq\mathbb{E}\left[f_{0}\left(\boldsymbol{M}\right)\right]\nonumber \\
 & \geq\frac{1}{n}\log\mathcal{E}\left(f\right)-\frac{\small\log\small\left(1+\sqrt{\frac{\pi c_{1}^{2}}{c_{2}}}\exp\left(\frac{1}{4c_{2}}\right)\right)}{n}\label{eq:ApproxEf_subgaussian}
\end{align}
with $\mathcal{E}\left(f\right)$ defined in (\ref{eq:DefnExponentialMean}).\end{lem}

\begin{proof}See Appendix \ref{sec:Proof-of-Theorem-Compute-Mean-SubExponential}.\end{proof}

\begin{remark}For measures with sub-exponential tails, if $c_{2}<1$,
the concentration is not decaying sufficiently fast and is unable
to ensure that $\mathcal{E}\left(f\right)=\mathbb{E}\left[e^{nf_{0}\left(\boldsymbol{M}\right)}\right]$
exists. \end{remark}

In short, Lemma \ref{thm:ComputeMeanSubExponential} asserts that
if $f_{0}\left(\boldsymbol{M}\right)$ possesses a sub-exponential
or a sub-Gaussian tail, then $\mathbb{E}[f_{0}\left(\boldsymbol{M}\right)]$
can be approximated by $\frac{1}{n}\log\mathcal{E}\left(f\right)$
in a reasonably tight manner, namely, within a gap no worse than $\mathcal{O}\left(\frac{1}{n}\right)$.
Since Proposition \ref{thm:GeneralTemplate} implies sub-Gaussian
tails for various measures, we immediately arrive at the following
concentration results that concern a large class of sub-Gaussian and
heavy-tailed measures.

\begin{theorem}\label{thm:Deviation-Ef}Let $c_{\rho,f,D}$, $c_{\rho,f,c_{\mathrm{ls}}}$,
and $c_{\rho,f,\tau_{c},\sigma_{c}}$ be numerical values defined
in Table \ref{tab:Summary-of-Preconstants-Thm-Deviation}, and set
\begin{equation}
\mu_{\rho,g,A}:=\nu\rho\left\Vert g\right\Vert _{\mathcal{L}}A.\label{eq:Kappa_beta_rho_g}
\end{equation}

(1) (\textbf{Bounded Measure}) Under the assumptions of Proposition
\ref{thm:GeneralTemplate}(a), we have
\begin{equation}
\begin{cases}
f_{0}\left(\boldsymbol{M}\right) & \leq\frac{1}{n}\log\mathcal{E}\left(f\right)+\frac{\beta\mu_{\rho,g,D}}{n},\\
f_{0}\left(\boldsymbol{M}\right) & \geq\frac{1}{n}\log\mathcal{E}\left(f\right)-\frac{\beta\mu_{\rho,g,D}}{n}-\frac{c_{\rho,f,D}}{n},
\end{cases}\label{eq:BoundedDeviation-Ef}
\end{equation}
with probability exceeding $1-4\exp\left(-\frac{\beta^{2}}{8\kappa}\right)$. 

(2) (\textbf{Logarithmic Sobolev Measure}) Under the assumptions of
Proposition \ref{thm:GeneralTemplate}(b), we have
\begin{align}
\begin{cases}
f_{0}\left(\boldsymbol{M}\right) & \leq\frac{1}{n}\log\mathcal{E}\left(f\right)+\frac{\beta\mu_{\rho,g,\sqrt{c_{\mathrm{ls}}}}}{n},\\
f_{0}\left(\boldsymbol{M}\right) & \geq\frac{1}{n}\log\mathcal{E}\left(f\right)-\frac{\beta\mu_{\rho,g,\sqrt{c_{\mathrm{ls}}}}}{n}-\frac{c_{\rho,f,c_{\mathrm{ls}}}}{n},
\end{cases}\label{eq:LogSobolevDeviation-Ef}
\end{align}
with probability at least $1-2\exp\left(-\frac{\beta^{2}}{\kappa}\right)$. 

(3) (\textbf{Heavy-tailed Distribution}) Under the assumptions of
Proposition \ref{thm:GeneralTemplate}(c), we have
\begin{equation}
\begin{cases}
f_{0}\left(\boldsymbol{M}\right) & \leq\frac{1}{n}\log\mathcal{E}_{\tilde{\boldsymbol{M}}}\left(f\right)+\frac{\mu_{\rho,g,\zeta}}{n},\\
f_{0}\left(\boldsymbol{M}\right) & \geq\frac{1}{n}\log\mathcal{E}_{\tilde{\boldsymbol{M}}}\left(f\right)-\frac{\mu_{\rho,g,\zeta}}{n}-\frac{c_{\rho,f,\tau_{c},\sigma_{c}}}{n},
\end{cases}\label{eq:HeavyTailDeviation-Ef}
\end{equation}
with probability exceeding $1-5n^{-c(n)}$, where
\[
\begin{cases}
\zeta & :=\text{ }\tau_{c}\sigma_{c}\sqrt{8\kappa c(n)\log n};\\
\mathcal{E}_{\tilde{\boldsymbol{M}}}\left(f\right) & :=\text{ }\mathbb{E}\left[\exp\left(nf_{0}\left(\tilde{\boldsymbol{M}}\right)\right)\right].
\end{cases}
\]
\end{theorem}

\begin{proof}See Appendix \ref{sec:Proof-of-Theorem-Deviation-Ef}.\end{proof}

\begin{table}
\caption{\label{tab:Summary-of-Preconstants-Thm-Deviation}Summary of parameters
of Theorem \ref{thm:Deviation-Ef}.}

\vspace{5pt}
\centering{}\footnotesize%
\begin{tabular}{>{\centering}p{0.9\linewidth}}
\hline 
\raggedright{}\emph{$c_{\rho,f,D}:=\log\left(1+\sqrt{8\kappa\pi}D\rho\nu\left\Vert g\right\Vert _{\mathcal{L}}e^{\frac{8\pi}{\kappa}+2\kappa D^{2}\rho^{2}\nu^{2}\left\Vert g\right\Vert _{\mathcal{L}}^{2}}\right)$}\tabularnewline
\raggedright{}\emph{$c_{\rho,f,c_{\mathrm{ls}}}:=\log\left(1+\sqrt{4\kappa\pi c_{\mathrm{ls}}}\rho\nu\left\Vert g\right\Vert _{\mathcal{L}}e^{\frac{\kappa}{4}c_{\mathrm{ls}}\rho^{2}\nu^{2}\left\Vert g\right\Vert _{\mathcal{L}}^{2}}\right)$}\tabularnewline
\raggedright{}\emph{$c_{\rho,f,\tau_{c},\sigma_{c}}:=\log\left(1+\sqrt{8\kappa\pi}\tau_{c}\sigma_{c}\rho\left\Vert g\right\Vert _{\mathcal{L}}e^{\frac{8\pi}{\kappa}+2\kappa\tau_{c}^{2}\sigma_{c}^{2}\rho^{2}\left\Vert g\right\Vert _{\mathcal{L}}^{2}}\right)$}\tabularnewline
\hline 
\tabularnewline
\end{tabular}
\end{table}

While Lemma \ref{thm:ComputeMeanSubExponential} focuses on sub-exponential
and sub-Gaussian measures, we are able to extend the concentration
phenomenon around $\frac{1}{n}\log\mathcal{E}\left(f\right)$ to a
much larger class of distributions including heavy-tailed measures
through certain truncation arguments.

To get a more informative understanding of Theorem \ref{thm:Deviation-Ef},
consider the case where $\left\Vert g\right\Vert _{\mathcal{L}}$,
$\tau_{c}$, $D$, $c_{\mathrm{ls}}$, $\rho$ are all constants.
One can see that with probability exceeding $1-\epsilon$ for any
small constant $\epsilon>0$,
\[
f_{0}\left(\boldsymbol{M}\right)=\frac{1}{n}\log\mathcal{E}\left(f\right)+\mathcal{O}\left(\frac{1}{n}\right)
\]
holds for measures of bounded support and measures satisfying the
LSI. In comparison, for symmetric power-law measures satisfying $\mathbb{P}\left(\left|\boldsymbol{M}_{ij}\right|\geq x\right)\leq x^{-\lambda}$
for some $\lambda>0$, the function $\tau_{c}$ is typically a function
of the form $n^{\delta}$ for some constant $\delta>0$. In this case,
with probability at least $1-\epsilon$, one has 
\[
f_{0}\left(\boldsymbol{M}\right)=\frac{1}{n}\log\mathcal{E}_{\tilde{\boldsymbol{M}}}\left(f\right)+\mathcal{O}\left(\frac{1}{n^{1-\delta}}\right)
\]
for heavy-tailed measures, where the uncertainty cannot be controlled
as well as for bounded measures or measures obeying the LSI. 

In the scenario where $\mathcal{E}\left(f\right)$ can be computed,
Theorem \ref{thm:Deviation-Ef} presents a full characterization of
the confidence interval of the objective metrics taking the form of
linear spectral statistics. In fact, in many applications, $\mathcal{E}\left(f\right)$
(rather than $\mathbb{E}\left[f_{0}\right]$) can be precisely computed
for a general class of probability distributions beyond the Gaussian
measure, which allows for accurate characterization of the concentration
of the objective metrics via Theorem \ref{thm:Deviation-Ef}, as illustrated
in Section \ref{sec:Sample-Applications}.

\subsection{Confidence Interval\label{sub:Confidence-Interval}}

In this subsection, we demonstrate that the sharp spectral measure
concentration phenomenon allows us to estimate the mean of the target
metric in terms of narrow confi{}dence intervals.

Specifically, suppose that we have obtained the value of the objective
metric $f_{0}\left(\boldsymbol{M}\right)=\frac{1}{n}\sum_{i}f\left(\lambda_{i}\left(\frac{1}{n}\boldsymbol{M}\boldsymbol{R}\boldsymbol{R}^{*}\boldsymbol{M}^{*}\right)\right)$
for a given realization $\boldsymbol{M}$. The goal is find an interval
(called $(1-\alpha_{0})$ confidence interval) 
\[
\big[l\left(\boldsymbol{M}\right),u\left(\boldsymbol{M}\right)\big]
\]
such that
\begin{equation}
\mathbb{P}\left(\mathbb{E}\left[f_{0}\right]\in\big[l\left(\boldsymbol{M}\right),u\left(\boldsymbol{M}\right)\big]\right)\geq1-\alpha_{0}
\end{equation}
for some constant $\alpha_{0}\in\left(0,1\right)$. 

Consider the assumptions of Proposition \ref{thm:GeneralTemplate},
i.e. $\boldsymbol{M}_{ij}$'s are independently generated satisfying
$\mathbb{E}[\boldsymbol{M}_{ij}]=0$ and $\mathbb{E}[|\boldsymbol{M}_{ij}|^{2}]=\nu_{ij}^{2}\leq\nu^{2}$.
An immediate consequence of Proposition \ref{thm:GeneralTemplate}
is stated as follows.
\begin{itemize}
\item \textbf{(Bounded Measure)} If $\boldsymbol{M}_{ij}$'s are bounded
by $D$ and $g(\cdot)$ is convex, then
\begin{equation}
\left[f_{0}\left(\boldsymbol{M}\right)\pm\frac{\sqrt{8\kappa}D\rho\nu\left\Vert g\right\Vert _{\mathcal{L}}\sqrt{\log\frac{4}{\alpha_{0}}}}{n}\right]
\end{equation}
is an $(1-\alpha_{0})$ confidence interval for $\mathbb{E}\left[f_{0}\left(\boldsymbol{M}\right)\right]$.
\end{itemize}
\vspace{0.05in}
\begin{itemize}
\item \textbf{(Logarithmic Sobolev Measure)} If the measure of $\boldsymbol{M}_{ij}$
satisfies the LSI with a uniform constant $c_{\mathrm{ls}}$, then
\begin{equation}
\left[f_{0}\left(\boldsymbol{M}\right)\pm\frac{\sqrt{\kappa c_{\mathrm{ls}}}\rho\nu\left\Vert g\right\Vert _{\mathcal{L}}\sqrt{\log\frac{2}{\alpha_{0}}}}{n}\right]
\end{equation}
is an $(1-\alpha_{0})$ confidence interval for $\mathbb{E}\left[f_{0}\left(\boldsymbol{M}\right)\right]$.
\end{itemize}
\vspace{0.05in}
\begin{itemize}
\item \textbf{(Sub-Exponential Measure)} If the measure of $\boldsymbol{M}_{ij}$
is symmetric about 0 and $\mathbb{P}\left(\left|\boldsymbol{M}_{ij}\right|>x\right)\leq e^{-\lambda x}$
for some constant $\lambda>0$, then $c(n)=\frac{\log\left(\frac{5}{\alpha_{0}}\right)}{\log n}$
and $\tau_{c}=\frac{1}{\lambda}\log\left(\frac{5mn}{\alpha_{0}}\right)$,
indicating that
\begin{equation}
\small\left[f_{0}\left(\boldsymbol{M}\right)\pm\frac{\sqrt{4\kappa\log\left(\frac{5}{\alpha_{0}}\right)}\rho\nu\left\Vert g\right\Vert _{\mathcal{L}}}{\lambda}\frac{\log\left(\frac{5mn}{\alpha_{0}}\right)}{n}\right]
\end{equation}
is an $(1-\alpha_{0})$ confidence interval for $\mathbb{E}[f_{0}(\tilde{\boldsymbol{M}})]$.
\end{itemize}
\vspace{0.05in}
\begin{itemize}
\item \textbf{(Power-Law Measure)} If the measure of $\boldsymbol{M}_{ij}$
is symmetric about 0 and $\mathbb{P}\left(\left|\boldsymbol{M}_{ij}\right|>x\right)\leq x^{-\lambda}$
for some constant $\lambda>0$, then $c(n)=\frac{\log\left(\frac{5}{\alpha_{0}}\right)}{\log n}$
and $\tau_{c}=\left(\frac{5mn}{\alpha_{0}}\right)^{\frac{1}{\lambda}}$,
indicating that
\[
\small\left[f_{0}\left(\boldsymbol{M}\right)\pm\frac{\left(\sqrt{4\kappa\log\left(\frac{5}{\alpha_{0}}\right)}\left(\frac{5}{\alpha_{0}}\right)^{\frac{1}{\lambda}}\rho\nu\left\Vert g\right\Vert _{\mathcal{L}}\right)\left(mn\right)^{\frac{1}{\lambda}}}{n}\right]
\]
is an $(1-\alpha_{0})$ confidence interval for $\mathbb{E}[f_{0}(\tilde{\boldsymbol{M}})]$.
\end{itemize}
One can see from the above examples that the spans of the confidence
intervals under power-law distributions are much less controlled than
that under sub-Gaussian measure. Depending on the power-law decay
exponent $\lambda$, the typical deviation can be as large as $\mathcal{O}\left(n^{-1+\frac{2}{\lambda}}\right)$
as compared to $\mathcal{O}\left(\frac{1}{n}\right)$ under various
sub-Gaussian measures.

When $D$, $c_{\mathrm{ls}}$, $\rho$, $\left\Vert g\right\Vert _{\mathcal{L}}$,
and $\alpha_{0}$ are all constants, the widths of the above confidence
intervals decay with $n$, which is negligible for many metrics of
interest.

\subsection{A General Template for Applying Proposition \ref{thm:GeneralTemplate}
and Theorem \ref{thm:Deviation-Ef}\label{sub:Template}}

For pedagogical reasons, we provide here a general recipe regarding
how to apply Proposition \ref{thm:GeneralTemplate} and Theorem \ref{thm:Deviation-Ef}
to evaluate the fluctuation of system performance metrics in random
MIMO channels with channel matrix $\boldsymbol{H}$. 
\begin{enumerate}
\item Transform the performance metric into a linear spectral statistic,
i.e. write the metric in the form $f_{0}\left(\boldsymbol{M}\right)=\frac{1}{n}\sum_{i=1}^{n}f\left(\lambda_{i}\left(\frac{1}{n}\boldsymbol{H}\boldsymbol{R}\boldsymbol{R}^{*}\boldsymbol{H}^{*}\right)\right)$
for some function $f\left(\cdot\right)$ and some deterministic matrix
$\boldsymbol{R}$. 
\item For measures satisfying the LSI, it suffices to calculate the Lipschitz
constant of $g(x):=f(x^{2})$. For both measures of bounded support
and heavy-tailed distributions, since the function $g(x):=f(x^{2})$
is non-convex in general, one typically needs to convexify $g(x)$
first. In particular, one might want to identify two reasonably tight
approximation $g_{1}(x)$ and $g_{2}(x)$ of $g(x)$ such that: 1)
$g_{1}(x)$ and $g_{2}(x)$ are both convex (or concave); 2) $g_{1}(x)\leq g(x)\leq g_{2}(x)$.
\item Apply Proposition \ref{thm:GeneralTemplate} and / or Theorem \ref{thm:Deviation-Ef}
on either $g(x)$ (for measures satisfying the LSI) or on $g_{1}(x)$
and $g_{2}(x)$ (for bounded measures or heavy-tailed measures) to
obtain concentration bounds.
\end{enumerate}
This recipe will be used to establish the canonical examples provided
in Section \ref{sec:Sample-Applications}.

\section{Some Canonical Examples\label{sec:Sample-Applications}}

In this section, we apply our general analysis framework developed
in Section \ref{sec:General-Template} to a few canonical examples
that arise in wireless communications and signal processing. Rather
than making each example as general as possible, we present only simple
settings that admit concise expressions from measure concentration.
We emphasize these simple illustrative examples in order to demonstrate
the effectiveness of our general treatment.

\subsection{Mutual Information and Power Offset of Random MIMO Channels\label{sec:Applications:-MIMO-Capacity}}

Consider the following MIMO channel
\begin{equation}
\boldsymbol{y}=\boldsymbol{H}\boldsymbol{x}+\boldsymbol{z},\label{eq:MIMOchannel}
\end{equation}
where $\boldsymbol{H}\in\mathbb{R}^{n_{\mathrm{r}}\times n_{\mathrm{t}}}$
denotes the channel matrix, $\boldsymbol{x}\in\mathbb{R}^{n_{\mathrm{t}}}$
represents the transmit signal, and $\boldsymbol{y}\in\mathbb{R}^{n_{\mathrm{r}}}$
is the received signal. We denote by $\boldsymbol{z}\sim\mathcal{N}\left(\boldsymbol{0},\sigma^{2}\boldsymbol{I}_{n}\right)$
the additive Gaussian noise. Note that (\ref{eq:MIMOchannel}) allows
modeling of a large class of random vector channels (e.g. MIMO-OFDM
channels\cite{chiani2011outage}, CDMA systems \cite{verdu1999spectral},
undersampled channels \cite{ChenGolEld2010}) beyond multiple antenna
channels. For instance, in unfaded direct-sequence CDMA systems, the
columns of $\boldsymbol{H}$ can represent random spreading sequences;
see \cite[Section 3.1.1]{TulinoVerdu2004} for details. 

The total power is assumed to be $P$, independent of $n_{\mathrm{t}}$
and $n_{\mathrm{r}}$, and the signal-to-noise ratio (SNR) is denoted
by 
\begin{equation}
\mbox{\ensuremath{\mathsf{SNR}}}:=\frac{P}{\sigma^{2}}.\label{eq:SNR}
\end{equation}
In addition, we denote the degrees of freedom as
\begin{equation}
n:=\min\left(n_{\text{t}},n_{\text{r}}\right),
\end{equation}
and use $\alpha>0$ to represent the ratio
\begin{equation}
\alpha:=\frac{n_{\mathrm{t}}}{n_{\mathrm{r}}}.\label{eq:DefnAlpha}
\end{equation}
We suppose throughout that $\alpha$ is a \emph{universal constant}
that does not scale with $n_{\text{t}}$ and $n_{\text{r}}$. 

Consider the simple channel model where $\left\{ \boldsymbol{H}_{ij}:1\leq i\leq n_{\mathrm{r}},1\leq j\leq n_{\mathrm{t}}\right\} $
are independently distributed. Suppose that channel state information
(CSI) is available to both the transmitter and the receiver. When
equal power allocation is adopted at all transmit antennas, it is
well known that the mutual information $C\left(\boldsymbol{H},\mbox{\ensuremath{\mathsf{SNR}}}\right)$
of the MIMO channel (\ref{eq:MIMOchannel}) under equal power allocation
is \cite{Tel1999}
\begin{equation}
C\left(\boldsymbol{H},\mbox{\ensuremath{\mathsf{SNR}}}\right)=\log\det\left(\boldsymbol{I}+\mbox{\ensuremath{\mathsf{SNR}}}\cdot\frac{1}{n_{\mathrm{t}}}\boldsymbol{H}\boldsymbol{H}^{*}\right),
\end{equation}
which depends only on the eigenvalue distribution of $\boldsymbol{H}\boldsymbol{H}^{*}$.
In the presence of asymptotically high SNR and channel dimensions,
it is well known that if $\boldsymbol{H}_{ij}$'s are independent
with zero mean and unit variance, then, almost surely,
\begin{align}
 & \lim_{\mathrm{SNR}\rightarrow\infty}\lim_{n_{\mathrm{r}}\rightarrow\infty}\left(\frac{C\left(\boldsymbol{H},\mbox{\ensuremath{\mathsf{SNR}}}\right)}{n_{\mathrm{r}}}-\min\left\{ \alpha,1\right\} \cdot\log\mbox{\ensuremath{\mathsf{SNR}}}\right)\nonumber \\
 & \quad=\begin{cases}
-1+\left(\alpha-1\right)\log\left(\frac{\alpha}{\alpha-1}\right),\quad & \text{if }\alpha\geq1,\\
-\alpha\log\left(\alpha e\right)+\left(1-\alpha\right)\log\left(\frac{1}{1-\alpha}\right), & \text{if }\alpha<1,
\end{cases}
\end{align}
which is independent of the precise entry distribution of $\boldsymbol{H}$
(e.g. \cite{lozano2002capacity}). The method of deterministic equivalents
has also been used to obtain good approximations under finite channel
dimensions \cite[Chapter 6]{couillet2011random}. In contrast, our
framework characterizes the concentration of mutual information in
the presence of finite SNR and finite $n_{\mathrm{r}}$ with reasonably
tight confidence intervals. Interestingly, the MIMO mutual information
is well-controlled within a narrow interval, as formally stated in
the following theorem. 

\begin{theorem}\label{thm:CapacityIIDChannel}Assume perfect CSI
at both the transmitter and the receiver, and equal power allocation
at all transmit antennas. Suppose that $\boldsymbol{H}\in\mathbb{R}^{n_{\mathrm{r}}\times n_{\mathrm{t}}}$
where $\boldsymbol{H}_{ij}$'s are independent random variables satisfying
$\mathbb{E}\left[\boldsymbol{H}_{ij}\right]=0$ and $\mathbb{E}[|\boldsymbol{H}_{ij}|^{2}]=1$.
Set $n:=\min\left\{ n_{\mathrm{r}},n_{\mathrm{t}}\right\} $.

(a) If $\boldsymbol{H}_{ij}$'s are bounded by $D$, then for any
$\beta>8\sqrt{\pi}$, 
\begin{equation}
\small\frac{C\left(\boldsymbol{H},\mbox{\ensuremath{\mathsf{SNR}}}\right)}{n_{\mathrm{r}}}\in\begin{cases}
\log\mbox{\ensuremath{\mathsf{SNR}}}+\left[\frac{1}{n_{\mathrm{r}}}\log\mathcal{R}\left(\frac{2}{e\small\mathsf{SNR}},n_{\mathrm{r}},n_{\mathrm{t}}\right)+\frac{\beta r_{\mathrm{bd}}^{\mathrm{lb},+}}{n_{\mathrm{r}}},\right.\\
\left.\quad\text{ }\text{ }\frac{1}{n_{\mathrm{r}}}\log\mathcal{R}\left(\frac{e}{2\small\mathsf{SNR}},n_{\mathrm{r}},n_{\mathrm{t}}\right)+\frac{\beta r_{\mathrm{bd}}^{\mathrm{ub},+}}{n_{\mathrm{r}}}\right],\text{ }\text{ if }\alpha\geq1\\
\quad\\
\alpha\log\frac{\mbox{\ensuremath{\mathsf{SNR}}}}{\alpha}+\left[\frac{1}{n_{\mathrm{r}}}\log\mathcal{R}\left(\frac{2\alpha}{e\small\mathsf{SNR}},n_{\mathrm{t}},n_{\mathrm{r}}\right)+\frac{\beta r_{\mathrm{bd}}^{\mathrm{lb},-}}{n_{\mathrm{r}}},\right.\\
\quad\text{ }\text{ }\left.\frac{1}{n_{\mathrm{r}}}\log\mathcal{R}\left(\frac{e\alpha}{2\small\mathsf{SNR}},n_{\mathrm{t}},n_{\mathrm{r}}\right)+\frac{\beta r_{\mathrm{bd}}^{\mathrm{ub},-}}{n_{\mathrm{r}}}\right],\text{ }\text{ if }\alpha<1
\end{cases}\label{eq:BoundsFepsilonBounded-Capacity}
\end{equation}
with probability exceeding $1-8\exp\left(-\frac{\beta^{2}}{8}\right)$. 

(b) If $\boldsymbol{H}_{ij}$'s satisfy the LSI with respect to a
uniform constant $c_{\mathrm{ls}}$, then for any $\beta>0$,
\begin{align}
\small\frac{C\left(\boldsymbol{H},\mbox{\ensuremath{\mathsf{SNR}}}\right)}{n_{\mathrm{r}}}\in\begin{cases}
\log\mbox{\ensuremath{\mathsf{SNR}}}+\frac{1}{n_{\mathrm{r}}}\log\mathcal{R}\left(\frac{1}{\small\mathsf{SNR}},n_{\mathrm{r}},n_{\mathrm{t}}\right)\\
\quad\quad\quad+\frac{\beta}{n_{\mathrm{r}}}\left[r_{\mathrm{ls}}^{\mathrm{lb},+},r_{\mathrm{ls}}^{\mathrm{ub},+}\right],\quad & \text{if }\alpha\geq1\\
\quad\\
\alpha\log\frac{\mbox{\ensuremath{\mathsf{SNR}}}}{\alpha}+\frac{1}{n_{\mathrm{r}}}\log\mathcal{R}\left(\frac{\alpha}{\small\mathsf{SNR}},n_{\mathrm{t}},n_{\mathrm{r}}\right)\\
\quad\quad\quad+\frac{\beta}{n_{\mathrm{r}}}\left[r_{\mathrm{ls}}^{\mathrm{lb},-},r_{\mathrm{ls}}^{\mathrm{ub},-}\right],\quad & \text{if }\alpha<1
\end{cases}\label{eq:BoundsFepsilonSubgaussian-Capacity}
\end{align}
with probability exceeding $1-4\exp\left(-\beta^{2}\right)$. 

(c) Suppose that $\boldsymbol{H}_{ij}$'s are independently drawn
from either sub-exponential distributions or heavy-tailed distributions
and that the distributions are symmetric about 0. Let $\tau_{c}(n)$
be defined as in (\ref{eq:TruncateMProb}) with respect to $\boldsymbol{H}_{ij}$'s
for some sequence $c(n)$. Then,
\begin{equation}
\small\frac{C\left(\boldsymbol{H},\mbox{\ensuremath{\mathsf{SNR}}}\right)}{n_{\mathrm{r}}}\in\begin{cases}
\log\mbox{\ensuremath{\mathsf{SNR}}}+\left[\frac{1}{n_{\mathrm{r}}}\log\mathcal{R}\left(\frac{2}{e\sigma_{\mathrm{c}}^{2}\small\mathsf{SNR}},n_{\mathrm{r}},n_{\mathrm{t}}\right)+\frac{r_{\mathrm{ht}}^{\mathrm{lb},+}}{n_{\mathrm{r}}},\right.\\
\left.\quad\quad\quad\quad\quad\frac{1}{n_{\mathrm{r}}}\log\mathcal{R}\left(\frac{e}{2\sigma_{c}^{2}\small\mathsf{SNR}},n_{\mathrm{r}},n_{\mathrm{t}}\right)+\frac{r_{\mathrm{ht}}^{\mathrm{ub},+}}{n_{\mathrm{r}}}\right]\\
\quad\quad\quad\quad+2\log\sigma_{c},\quad\quad\quad\quad\quad\quad\quad\text{if }\alpha\geq1\\
\\
\alpha\log\frac{\mbox{\ensuremath{\mathsf{SNR}}}}{\alpha}+\left[\frac{1}{n_{\mathrm{r}}}\log\mathcal{R}\left(\frac{2\alpha}{e\sigma_{\mathrm{c}}^{2}\small\mathsf{SNR}},n_{\mathrm{r}},n_{\mathrm{t}}\right)+\frac{r_{\mathrm{ht}}^{\mathrm{lb},-}}{n_{\mathrm{r}}},\right.\\
\quad\quad\quad\quad\quad\left.\frac{1}{n_{\mathrm{r}}}\log\mathcal{R}\left(\frac{e\alpha}{2\sigma_{c}^{2}\small\mathsf{SNR}},n_{\mathrm{t}},n_{\mathrm{r}}\right)+\frac{r_{\mathrm{ht}}^{\mathrm{ub},-}}{n_{\mathrm{r}}},\right]\\
\quad\quad\quad\quad+2\log\sigma_{c},\quad\quad\quad\quad\quad\quad\quad\text{if }\alpha<1
\end{cases}\label{eq:HeavyTailDeviation-Ef-Capacity}
\end{equation}
with probability exceeding $1-\frac{10}{n^{c(n)}}$.

Here,
\begin{equation}
\mathcal{R}\left(\epsilon,n,m\right):=\sum_{i=0}^{n}{n \choose i}\frac{\epsilon^{n-i}m^{-i}m!}{\left(m-i\right)!},\label{eq:ReExpression}
\end{equation}
and the residual terms are provided in Table \ref{tab:Summary-of-Preconstants-Cap}.\end{theorem}

\begin{proof}See Appendix \ref{sec:Proof-of-Theorem-CapacityIIDChannels}.
\end{proof}

\begin{remark}In the regime where $\beta=\Theta\left(1\right)$,
one can easily see that the magnitudes of all the residual terms $r_{\mathrm{bd}}^{\mathrm{ub},+}$,
$r_{\mathrm{bd}}^{\mathrm{lb},+}$, $r_{\mathrm{bd}}^{\mathrm{ub},-}$,
$r_{\mathrm{bd}}^{\mathrm{lb},-}$, $r_{\mathrm{ls}}^{\mathrm{ub},+}$,
$r_{\mathrm{ls}}^{\mathrm{lb},+}$, $r_{\mathrm{ls}}^{\mathrm{ub},-}$,
$r_{\mathrm{ls}}^{\mathrm{lb},-}$ do not scale with $n$. \end{remark}

\begin{table*}
\caption{\label{tab:Summary-of-Preconstants-Cap}Summary of parameters of Theorem
\ref{thm:CapacityIIDChannel} and Corollary \ref{corollary:CapacityIIDChannel-highSNR}.}

\vspace{5pt}
\centering{}\footnotesize%
\begin{tabular}{>{\raggedright}p{0.09\linewidth}>{\raggedright}p{0.55\linewidth}>{\raggedright}p{0.25\linewidth}}
\hline 
Theorem \ref{thm:CapacityIIDChannel} & $r_{\mathrm{bd}}^{\mathrm{lb},+}:=-\frac{D\sqrt{\small\mathsf{SNR}}}{\sqrt{\alpha}}-\frac{\log\left\{ 1+\sqrt{\frac{8\pi D^{2}\small\mathsf{SNR}}{\alpha}}e^{8\pi+\frac{2D^{2}\small\mathsf{SNR}}{\alpha}}\right\} }{\beta}$ & $r_{\mathrm{bd}}^{\mathrm{ub},+}:=\frac{D\sqrt{\small\mathsf{SNR}}}{\sqrt{\frac{e\alpha}{2}}}$\tabularnewline
 & $r_{\mathrm{bd}}^{\mathrm{lb},-}:=-\small D\sqrt{\small\mathsf{SNR}}-\frac{\log\left\{ 1+\sqrt{8\pi D^{2}\small\mathsf{SNR}}\text{ }e^{8\pi+2D^{2}\small\mathsf{SNR}}\right\} }{\beta}$ & $r_{\mathrm{bd}}^{\mathrm{ub},-}:=\frac{D\sqrt{\small\mathsf{SNR}}}{\sqrt{\frac{e}{2}}}$\tabularnewline
 & $r_{\mathrm{ls}}^{\mathrm{lb},+}:=-\frac{\sqrt{c_{\mathrm{ls}}\small\mathsf{SNR}}}{\sqrt{\alpha}}-\frac{\log\left(1+\sqrt{\frac{4\pi c_{\mathrm{ls}}\small\mathsf{SNR}}{\alpha}}e^{\frac{c_{\mathrm{ls}}\small\mathsf{SNR}}{4\alpha}}\right)}{\beta}$ & $r_{\mathrm{ls}}^{\mathrm{ub},+}:=\frac{\sqrt{c_{\mathrm{ls}}\small\mathsf{SNR}}}{\sqrt{\alpha}}$\tabularnewline
 & $r_{\mathrm{ls}}^{\mathrm{lb},-}:=-\small\sqrt{c_{\mathrm{ls}}}\sqrt{\small\mathsf{SNR}}-\frac{\log\left(1+\sqrt{4\pi c_{\mathrm{ls}}\small\mathsf{SNR}}\text{ }e^{\frac{c_{\mathrm{ls}}\small\mathsf{SNR}}{4}}\right)}{\beta}$ & $r_{\mathrm{ls}}^{\mathrm{ub},-}:=\small\sqrt{c_{\mathrm{ls}}}\sqrt{\small\mathsf{SNR}}$\tabularnewline
 & $r_{\mathrm{ht}}^{\mathrm{lb},+}:=-\frac{\tau_{c}\sigma_{c}\sqrt{8c(n)\log n}\sqrt{\small\mathsf{SNR}}}{\sqrt{\alpha}}-\footnotesize\log\left(1+\sqrt{\frac{8\pi\tau_{c}^{2}\sigma_{c}^{2}\small\mathsf{SNR}}{\alpha}}e^{8\pi+\frac{2\tau_{c}^{2}\sigma_{c}^{2}\small\mathsf{SNR}}{\alpha}}\right)$ & $r_{\mathrm{ht}}^{\mathrm{ub},+}=\frac{\tau_{c}\sigma_{c}\sqrt{8c(n)\log n}\sqrt{\small\mathsf{SNR}}}{\sqrt{\frac{e}{2}\alpha}}$\tabularnewline
 & $r_{\mathrm{ht}}^{\mathrm{lb},-}:=\footnotesize-\tau_{c}\sigma_{c}\sqrt{8c(n)\log n}\sqrt{\small\mathsf{SNR}}-\log\left(1+\sqrt{8\pi\tau_{c}^{2}\sigma_{c}^{2}\small\mathsf{SNR}}\text{ }e^{4\pi+2\tau_{c}^{2}\sigma_{c}^{2}\small\mathsf{SNR}}\right)$ & $r_{\mathrm{ht}}^{\mathrm{ub},-}:=\frac{\tau_{c}\sigma_{c}\sqrt{8c(n)\log n}\sqrt{\small\mathsf{SNR}}}{\sqrt{\frac{e}{2}}}$\tabularnewline
 &  & \tabularnewline
Corollary \ref{corollary:CapacityIIDChannel-highSNR} & $\gamma_{\mathrm{bd}}^{\mathrm{ub},+}:=\frac{1.5\log\left(en_{\mathrm{r}}\right)}{n_{\mathrm{r}}}+4\sqrt{\frac{2\alpha}{e\small\mathsf{SNR}}}\log\sqrt{\frac{e\small\mathsf{SNR}}{2\alpha}}$ & $\gamma_{\mathrm{bd}}^{\mathrm{lb},+}:=\frac{\frac{1}{2}\log n_{\mathrm{r}}-\log\frac{n_{\mathrm{t}}+1}{2\pi}}{n_{\mathrm{r}}}$\tabularnewline
 & $\gamma_{\mathrm{bd}}^{\mathrm{ub},-}:=\frac{1.5\alpha\log\left(en_{\mathrm{t}}\right)}{n_{\mathrm{r}}}+4\alpha^{2}\sqrt{\frac{2}{e\small\mathsf{SNR}}}\log\sqrt{\frac{e\small\mathsf{SNR}}{2}}$ & $\gamma_{\mathrm{bd}}^{\mathrm{lb},-}:=\frac{\alpha\left(\frac{1}{2}\log n_{\mathrm{t}}-\log\frac{n_{\mathrm{r}}+1}{2\pi}\right)}{n_{\mathrm{r}}}$\tabularnewline
 & $\gamma_{\mathrm{ls}}^{\mathrm{ub},+}:=\frac{1.5\log\left(en_{\mathrm{r}}\right)}{n_{\mathrm{r}}}+4\sqrt{\frac{\alpha}{\small\mathsf{SNR}}}\log\sqrt{\frac{\small\mathsf{SNR}}{\alpha}}$ & $\gamma_{\mathrm{ls}}^{\mathrm{lb},+}:=\frac{\frac{1}{2}\log n_{\mathrm{r}}-\log\frac{n_{\mathrm{t}}+1}{2\pi}}{n_{\mathrm{r}}}$\tabularnewline
 & $\gamma_{\mathrm{ls}}^{\mathrm{ub},-}:=\frac{1.5\alpha\log\left(en_{\mathrm{t}}\right)}{n_{\mathrm{r}}}+4\frac{\alpha^{2}}{\sqrt{\small\mathsf{SNR}}}\log\sqrt{\small\mathsf{SNR}}$ & $\gamma_{\mathrm{ls}}^{\mathrm{lb},-}:=\frac{\alpha\left(\frac{1}{2}\log n_{\mathrm{t}}-\log\frac{n_{\mathrm{r}}+1}{2\pi}\right)}{n_{\mathrm{r}}}$\tabularnewline
 & $\gamma_{\mathrm{ht}}^{\mathrm{ub},+}:=\frac{1.5\log\left(en_{\mathrm{r}}\right)}{n_{\mathrm{r}}}+4\sqrt{\frac{2\alpha}{e\sigma_{c}^{2}\small\mathsf{SNR}}}\log\sqrt{\frac{e\sigma_{c}^{2}\small\mathsf{SNR}}{2\alpha}}$ & $\gamma_{\mathrm{ht}}^{\mathrm{lb},+}:=\frac{\frac{1}{2}\log n_{\mathrm{r}}-\log\frac{n_{\mathrm{t}}+1}{2\pi}}{n_{\mathrm{r}}}$\tabularnewline
 & $\gamma_{\mathrm{ht}}^{\mathrm{ub},-}:=\frac{1.5\alpha\log\left(en_{\mathrm{t}}\right)}{n_{\mathrm{r}}}+4\alpha^{2}\sqrt{\frac{2}{e\sigma_{c}^{2}\small\mathsf{SNR}}}\log\sqrt{\frac{e\sigma_{c}^{2}\small\mathsf{SNR}}{2}}$ & $\gamma_{\mathrm{ht}}^{\mathrm{lb},-}:=\frac{\alpha\left(\frac{1}{2}\log n_{\mathrm{t}}-\log\frac{n_{\mathrm{r}}+1}{2\pi}\right)}{n_{\mathrm{r}}}$\tabularnewline
\hline 
\end{tabular}
\end{table*}

The above theorem relies on the expression $\mathcal{R}\left(\epsilon,m,n\right)$.
In fact, this function is exactly equal to $\mathbb{E}\left[\det\left(\epsilon\boldsymbol{I}+\frac{1}{m}\boldsymbol{M}\boldsymbol{M}^{*}\right)\right]$
in a \emph{distribution-free} manner, as revealed by the following
lemma.

\begin{lem}\label{lemma:ExpDetIplustMM}Consider any random matrix
$\boldsymbol{M}\in\mathbb{C}^{n\times m}$ such that $\boldsymbol{M}_{ij}$'s
are independently generated satisfying 
\begin{equation}
\mathbb{E}\left[\boldsymbol{M}_{ij}\right]=0,\quad\text{and}\quad\mathbb{E}\left[\left|\boldsymbol{M}_{ij}\right|^{2}\right]=1.
\end{equation}
Then one has
\begin{equation}
\mathbb{E}\left[\det\left(\epsilon\boldsymbol{I}+\frac{1}{m}\boldsymbol{M}\boldsymbol{M}^{*}\right)\right]=\sum_{i=0}^{n}{n \choose i}\frac{\epsilon^{n-i}m^{-i}m!}{\left(m-i\right)!}\label{eq:ExpDetIplusMM}
\end{equation}
\end{lem}

\begin{proof}This lemma improves upon known results under Gaussian
random ensembles by generalizing them to a very general class of random
ensembles. See Appendix \ref{sec:Proof-of-Lemma-ExpDetIplustMM} for
the detailed proof. \end{proof}

To get a more quantitative assessment of the concentration intervals,
we plot the 95\% confidence interval of the MIMO mutual information
for a few cases in Fig \ref{fig:GaussianCap} when the channel matrix
$\boldsymbol{H}$ is an i.i.d. Gaussian random matrix. The expected
value of the capacity is adopted from the mean of 3000 Monte Carlo
trials. Our theoretical predictions of the deviation bounds are compared
against the simulation results consisting of Monte Carlo trials. One
can see from the plots that our theoretical predictions are fairly
accurate even for small channel dimensions, which corroborates the
power of concentration of measure techniques. 

\begin{figure*}
\centering

\emph{}%
\begin{tabular}{cc}
\begin{tabular}{cc}
\includegraphics[width=0.45\textwidth]{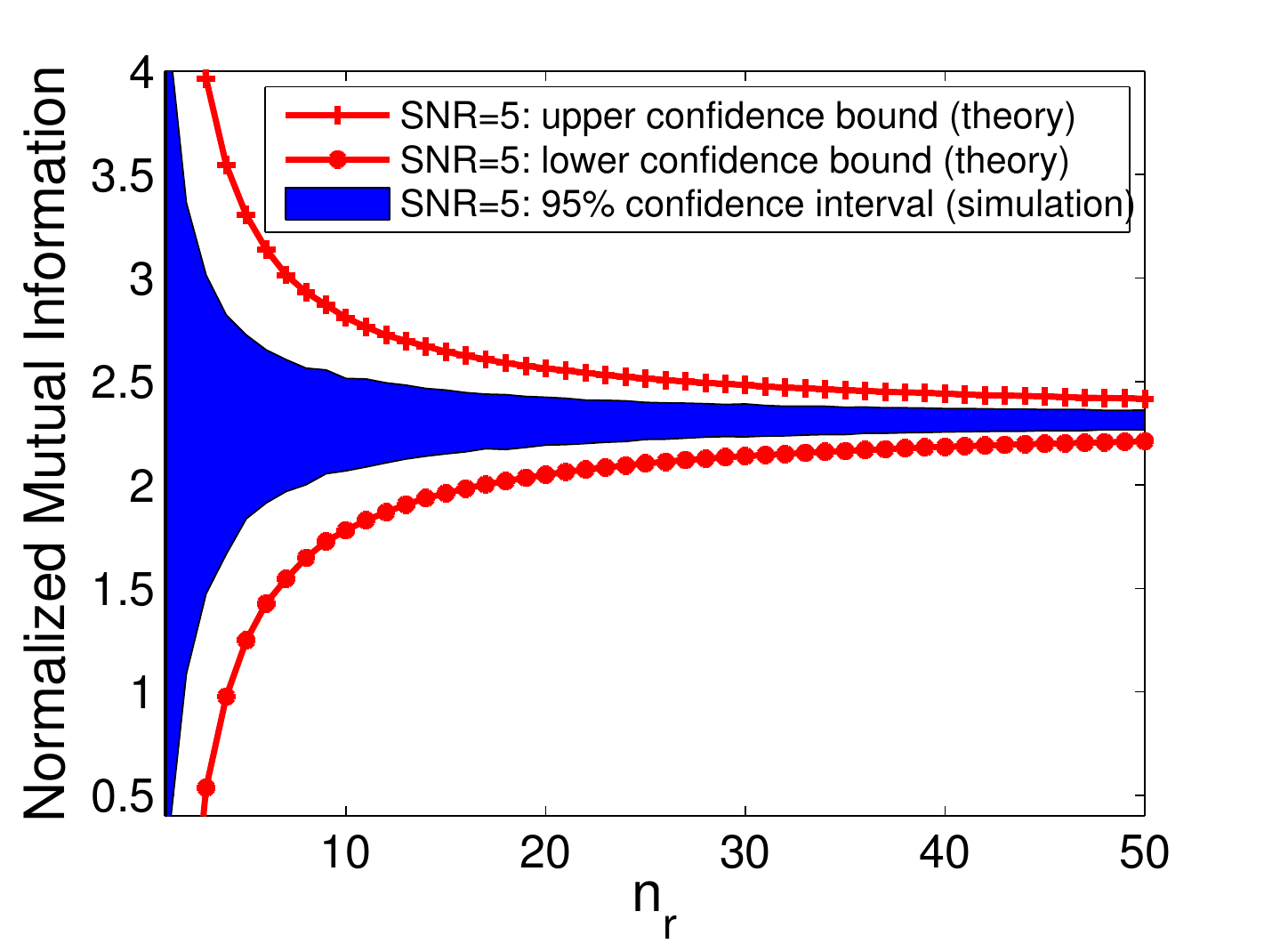} & \includegraphics[width=0.45\textwidth]{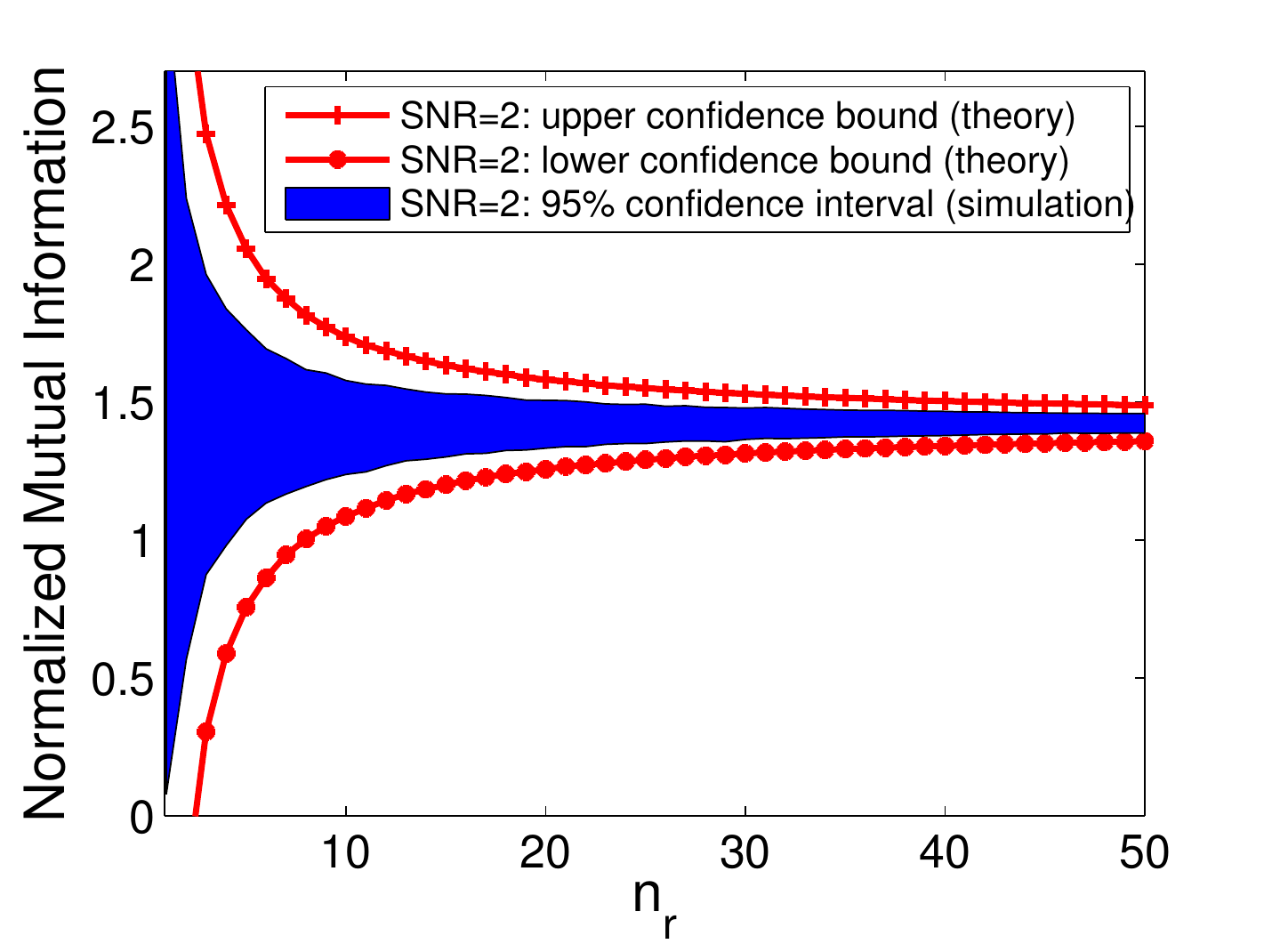}\tabularnewline
(a) & (b)\tabularnewline
\includegraphics[width=0.45\textwidth]{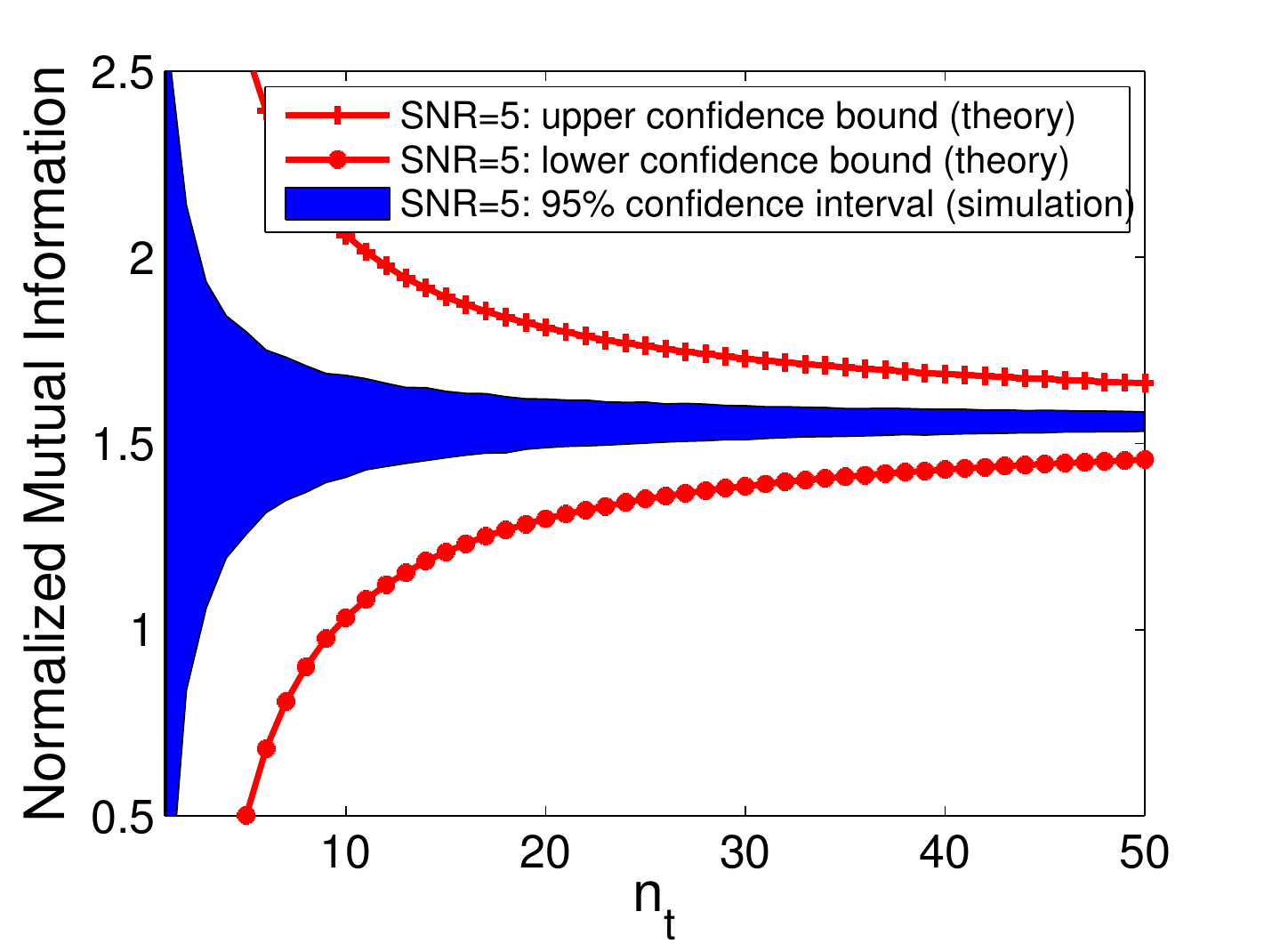} & \includegraphics[width=0.45\textwidth]{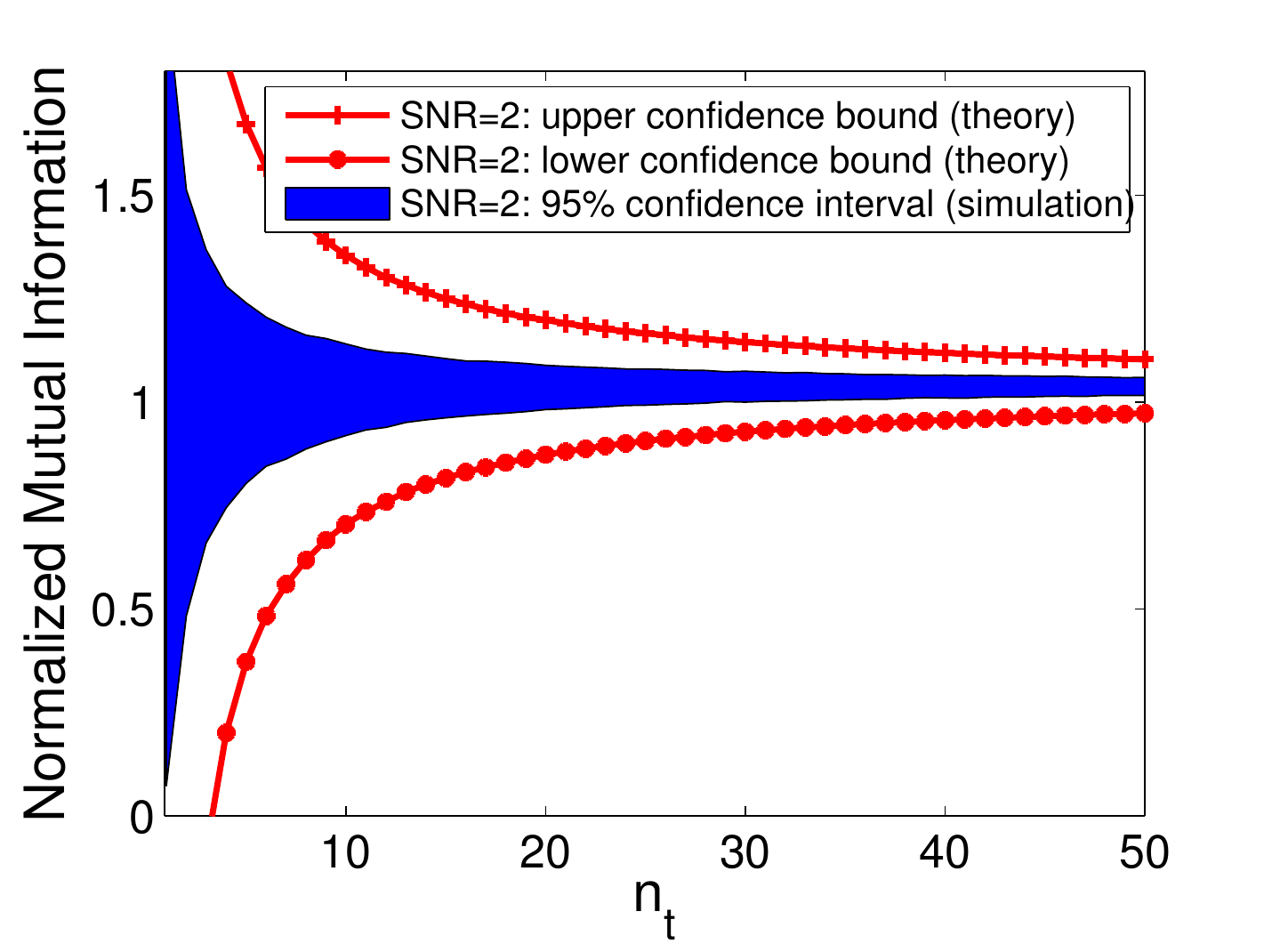}\tabularnewline
(c) & (d)\tabularnewline
\end{tabular} & \tabularnewline
\end{tabular}\caption{\label{fig:GaussianCap}$95\%$-confidence interval for MIMO mutual
information $\frac{C\left(\boldsymbol{H},\small\mathsf{SNR}\right)}{n_{r}}$
when $\boldsymbol{H}$ is drawn from i.i.d. standard Gaussian ensemble.
From Table \ref{tab:Summary-of-Preconstants-Cap}, the sizes of the
upper and lower confidence spans are given by $\frac{\beta\sqrt{\small\mathsf{SNR}}}{\sqrt{\alpha}n_{\mathrm{r}}}$
($\alpha>1$) and $\frac{\beta\sqrt{\small\mathsf{SNR}}}{n_{\mathrm{r}}}$
($\alpha<1$), respectively, where $\beta=\sqrt{\log\frac{8}{\left(1-5\%\right)}}$.
The expected value of $\frac{C\left(\boldsymbol{H},\small\mathsf{SNR}\right)}{n_{r}}$
is adopted from the mean of 3000 Monte Carlo trials. The theoretical
confidence bounds are compared against the simulation results (with
3000 Monte Carlo trials). The results are shown for (a) $\alpha=2$,
$\mbox{\ensuremath{\mathsf{SNR}}}=5$; (b) $\alpha=2$, $\mbox{\ensuremath{\mathsf{SNR}}}=2$;
(c) $\alpha=1/2$, $\mbox{\ensuremath{\mathsf{SNR}}}=5$; and (d)
$\alpha=1/2$, $\mbox{\ensuremath{\mathsf{SNR}}}=2$.}

\end{figure*}

At moderate-to-high SNR, the function $\mathcal{R}(\epsilon,n,m)$
admits a simple approximation. This observation leads to a concise
and informative characterization of the mutual information of the
MIMO channel, as stated in the following corollary.

\begin{corollary}\label{corollary:CapacityIIDChannel-highSNR}Suppose
that $\mbox{\ensuremath{\mathsf{SNR}}}>\frac{2}{e}\max\left\{ \alpha,1\right\} \cdot\max\left\{ e^{2}\alpha^{3},4\alpha\right\} $,
and set $n:=\min\{n_{\mathrm{t}},n_{\mathrm{r}}\}$.

(a) Consider the assumptions of Theorem \ref{thm:CapacityIIDChannel}(a).
For any $\beta>8\sqrt{\pi}$, 
\begin{equation}
\small\frac{C\left(\boldsymbol{H},\mbox{\ensuremath{\mathsf{SNR}}}\right)}{n_{\mathrm{r}}}\in\begin{cases}
\log\frac{\mbox{\small\ensuremath{\mathsf{SNR}}}}{e}+\left(\alpha-1\right)\log\left(\frac{\alpha}{\alpha-1}\right)+\\
\quad\left[\gamma_{\mathrm{bd}}^{\mathrm{lb},+}+\frac{\beta r_{\mathrm{bd}}^{\mathrm{lb},+}}{n_{\mathrm{r}}},\gamma_{\mathrm{bd}}^{\mathrm{ub},+}+\frac{\beta r_{\mathrm{bd}}^{\mathrm{ub},+}}{n_{\mathrm{r}}}\right], & \text{if }\alpha\geq1,\\
\\
\alpha\log\frac{\small\mathsf{SNR}}{\alpha e}+\left(1-\alpha\right)\log\left(\frac{1}{1-\alpha}\right)+\\
\quad\left[\gamma_{\mathrm{bd}}^{\mathrm{lb},-}+\frac{\beta r_{\mathrm{bd}}^{\mathrm{lb},-}}{n_{\mathrm{r}}},\gamma_{\mathrm{bd}}^{\mathrm{ub},-}+\frac{\beta r_{\mathrm{bd}}^{\mathrm{ub},-}}{n_{\mathrm{r}}}\right], & \text{if }\alpha<1,
\end{cases}\label{eq:BoundsFepsilonBounded-Capacity-High}
\end{equation}
with probability exceeding $1-8\exp\left(-\frac{\beta^{2}}{8}\right)$.

(b) Consider the assumptions of Theorem \ref{thm:CapacityIIDChannel}(b).
For any $\beta>0$, 
\begin{align}
\small\frac{C\left(\boldsymbol{H},\mbox{\ensuremath{\mathsf{SNR}}}\right)}{n_{\mathrm{r}}}\in\begin{cases}
\log\frac{\small\mathsf{SNR}}{e}+\left(\alpha-1\right)\log\left(\frac{\alpha}{\alpha-1}\right)+\\
\quad\left[\gamma_{\mathrm{ls}}^{\mathrm{lb},+}+\frac{\beta r_{\mathrm{ls}}^{\mathrm{lb},+}}{n_{\mathrm{r}}},\gamma_{\mathrm{ls}}^{\mathrm{ub},+}+\frac{\beta r_{\mathrm{ls}}^{\mathrm{ub},+}}{n_{\mathrm{r}}}\right],\quad & \text{if }\alpha\geq1,\\
\\
\alpha\log\frac{\small\mathsf{SNR}}{\alpha e}+\left(1-\alpha\right)\log\left(\frac{1}{1-\alpha}\right)+\\
\quad\left[\gamma_{\mathrm{ls}}^{\mathrm{lb},-}+\frac{\beta r_{\mathrm{ls}}^{\mathrm{lb},-}}{n_{\mathrm{r}}},\gamma_{\mathrm{ls}}^{\mathrm{ub},-}+\frac{\beta r_{\mathrm{ls}}^{\mathrm{ub},-}}{n_{\mathrm{r}}}\right], & \text{if }\alpha<1,
\end{cases}\label{eq:BoundsFepsilonSubgaussian-Capacity-High}
\end{align}
with probability exceeding $1-4\exp\left(-\beta^{2}\right)$.

(c) Consider the assumptions of Theorem \ref{thm:CapacityIIDChannel}(c).
Then, one has
\begin{equation}
\small\frac{C\left(\boldsymbol{H},\mbox{\ensuremath{\mathsf{SNR}}}\right)}{n_{\mathrm{r}}}\in\begin{cases}
\log\frac{\small\mathsf{SNR}}{e}+\left(\alpha-1\right)\log\left(\frac{\alpha}{\alpha-1}\right)+\log\sigma_{\mathrm{c}}+\\
\quad\left[\gamma_{\mathrm{ht}}^{\mathrm{lb},+}+\frac{r_{\mathrm{ht}}^{\mathrm{lb},+}}{n_{\mathrm{r}}},\gamma_{\mathrm{ht}}^{\mathrm{ub},+}+\frac{r_{\mathrm{ht}}^{\mathrm{ub},+}}{n_{\mathrm{r}}}\right],\quad\quad\text{if }\alpha\geq1,\\
\\
\alpha\log\frac{\small\mathsf{SNR}}{\alpha e}+\left(1-\alpha\right)\log\left(\frac{1}{1-\alpha}\right)+\log\sigma_{\mathrm{c}}+\\
\quad\left[\gamma_{\mathrm{ht}}^{\mathrm{lb},-}+\frac{r_{\mathrm{ht}}^{\mathrm{lb},-}}{n_{\mathrm{r}}},\gamma_{\mathrm{ht}}^{\mathrm{ub},-}+\frac{r_{\mathrm{ht}}^{\mathrm{ub},-}}{n_{\mathrm{r}}}\right],\quad\quad\text{if }\alpha<1,
\end{cases}\label{eq:HeavyTailDeviation-Ef-Capacity-High}
\end{equation}
with probability at least $1-10n^{-c(n)}$.

Here, the residual terms are formally provided in Table \ref{tab:Summary-of-Preconstants-Cap}.\end{corollary}

\begin{proof} See Appendix \ref{sec:Proof-of-Lemma-Bound-LogEAA}.
\end{proof}

\begin{remark}One can see that the magnitudes of all these extra
residual terms $\gamma_{\mathrm{bd}}^{\mathrm{ub},-}$, $\gamma_{\mathrm{bd}}^{\mathrm{lb},-}$,
$\gamma_{\mathrm{bd}}^{\mathrm{ub},+}$, $\gamma_{\mathrm{bd}}^{\mathrm{lb},+}$,
$\gamma_{\mathrm{ls}}^{\mathrm{ub},-}$, $\gamma_{\mathrm{ls}}^{\mathrm{lb},-}$,
$\gamma_{\mathrm{ls}}^{\mathrm{ub},+}$, and $\gamma_{\mathrm{ls}}^{\mathrm{lb},+}$
are no worse than the order
\begin{equation}
\mathcal{O}\left(\frac{\log n}{n}+\small\frac{\log\small\mathsf{SNR}}{\sqrt{\small\mathsf{SNR}}}\right),\label{eq:ResidualDev}
\end{equation}
which vanish as SNR and channel dimensions grow. \end{remark}

In fact, Corollary \ref{corollary:CapacityIIDChannel-highSNR} is
established based on the simple approximation of $\mathcal{R}\left(\epsilon,n,m\right):=\mathbb{E}\left[\det\left(\epsilon\boldsymbol{I}+\frac{1}{m}\boldsymbol{M}\boldsymbol{M}^{*}\right)\right]$.
At moderate-to-high SNR, this function can be approximated reasonably
well through much simpler expressions, as stated in the following
lemma. 

\begin{lem}\label{Lemma:Bound-LogEAA}Consider a random matrix $\boldsymbol{M}=\left(\boldsymbol{M}_{ij}\right)_{1\leq i\leq n,1\leq j\leq m}$
such that $\alpha:=m/n\geq1$. Suppose that $\boldsymbol{M}_{ij}$'s
are independent satisfying $\mathbb{E}\left[\boldsymbol{M}_{ij}\right]=0$
and $\mathbb{E}\left[\left|\boldsymbol{M}_{ij}\right|^{2}\right]=1$.
If $\frac{4}{n}<\epsilon<\min\left\{ \frac{1}{e^{2}\alpha^{3}},\frac{1}{4\alpha}\right\} $,
then one has
\begin{align}
 & \frac{1}{n}\log\mathbb{E}\left[\det\left(\epsilon\boldsymbol{I}+\frac{1}{m}\boldsymbol{M}\boldsymbol{M}^{*}\right)\right]\nonumber \\
 & \quad=\left(\alpha-1\right)\log\left(\frac{\alpha}{\alpha-1}\right)-1+r_{\mathrm{E}},\label{eq:BoundsLemmaElogDetAA-1}
\end{align}
where the residual term $r_{\mathrm{E}}$ satisfies
\[
r_{\mathrm{E}}\in\small\left[\frac{\frac{1}{2}\log n-\log\frac{1}{2\pi}\left(m+1\right)}{n},\frac{1.5\log\left(en\right)}{n}+2\sqrt{\alpha\epsilon}\log\frac{1}{\alpha\epsilon}\right].
\]
\end{lem}

\begin{proof}Appendix \ref{sec:Proof-of-Lemma-Bound-LogEAA}.\end{proof}

Combining Theorem \ref{thm:CapacityIIDChannel} and Lemma \ref{Lemma:Bound-LogEAA}
immediately establishes Corollary \ref{corollary:CapacityIIDChannel-highSNR}.

\subsubsection{Implications of Theorem \ref{thm:CapacityIIDChannel} and Corollary
\ref{corollary:CapacityIIDChannel-highSNR}}

Some implications of Corollary \ref{corollary:CapacityIIDChannel-highSNR}
are listed as follows.
\begin{enumerate}
\item When we set $\beta=\Theta\left(\sqrt{\frac{\log n}{n}}\right)$, Corollary
\ref{corollary:CapacityIIDChannel-highSNR} implies that in the presence
of moderate-to-high SNR and moderate-to-large channel dimensionality
$n=\min\left\{ n_{\mathrm{t}},n_{\mathrm{r}}\right\} $, the information
rate per receive antenna behaves like
\begin{align}
\footnotesize\frac{C\left(\boldsymbol{H},\mbox{\ensuremath{\mathsf{SNR}}}\right)}{n_{\mathrm{r}}} & =\begin{cases}
\log\frac{\small\mathsf{SNR}}{e}+\left(\alpha-1\right)\log\left(\frac{\alpha}{\alpha-1}\right)+\\
\text{ }\mathcal{O}\left(\frac{\small\mathsf{SNR}+\sqrt{\log n}}{n}+\frac{\log\small\mathsf{SNR}}{\sqrt{\small\mathsf{SNR}}}\right),\text{ }\text{if }\alpha\geq1\\
\\
\alpha\log\frac{\small\mathsf{SNR}}{\alpha e}+\left(1-\alpha\right)\log\left(\frac{1}{1-\alpha}\right)+\\
\text{ }\mathcal{O}\left(\frac{\small\mathsf{SNR}+\sqrt{\log n}}{n}+\frac{\log\small\mathsf{SNR}}{\sqrt{\small\mathsf{SNR}}}\right),\text{ }\text{if }\alpha<1
\end{cases}\label{eq:CapExpression_Order}
\end{align}
with high probability. That said, for a fairly large regime, the mutual
information falls within a narrow range. In fact, the mutual information
depends almost only on $\mbox{\ensuremath{\mathsf{SNR}}}$ and the
ratio $\alpha=\frac{n_{\mathrm{t}}}{n_{\mathrm{r}}}$, and is independent
of the precise number of antennas $n_{\mathrm{r}}$ and $n_{\mathrm{t}}$,
except for some vanishingly small residual terms. The first-order
term of the expression (\ref{eq:CapExpression_Order}) coincides with
existing \emph{asymptotic} results (e.g. \cite{lozano2002capacity}).
This in turns validates our concentration of measure approach.
\item Theorem \ref{thm:CapacityIIDChannel} indicates that the size of the
typical confidence interval for $\frac{C\left(\boldsymbol{H},\small\mathsf{SNR}\right)}{n_{\mathrm{r}}}$
decays with $n$ at a rate not exceeding $\mathcal{O}\left(\frac{1}{n}\right)$
(by picking $\beta=\Theta\left(1\right)$) under measures of bounded
support or measures satisfying the LSI. Note that it has been established
(e.g. \cite[Theorem 2 and Proposition 2]{HacKhoLouNajPas2008}) that
the asymptotic standard deviation of $\frac{C\left(\boldsymbol{H},\small\mathsf{SNR}\right)}{n_{\mathrm{r}}}$
scales as $\Theta\left(\frac{1}{n}\right)$. This reveals that the
concentration of measure approach we adopt is able to obtain the confidence
interval with optimal size. Recent works by Li, Mckay and Chen \cite{li2013distribution,chen2012coulumb}
derive polynomial expansion for each of the moments (e.g. mean, standard
deviation, and skewness) of MIMO mutual information, which can also
be used to approximate the distributions. In contrast, our results
are able to characterize any concentration interval in a simpler and
more informative manner. 
\item Define the\emph{ power offset }for any channel dimension as follows
\begin{equation}
\mathcal{L}\left(\boldsymbol{H},\mbox{\ensuremath{\mathsf{SNR}}}\right):=\log\mbox{\ensuremath{\mathsf{SNR}}}-\frac{C\left(\boldsymbol{H},\small\mathsf{SNR}\right)}{\min\left\{ n_{\mathrm{r}},n_{\mathrm{t}}\right\} }.\label{eq:DefnPowerOffset}
\end{equation}
One can see that this converges to the notion of \emph{high-SNR power
offset} investigated by Lozano et. al. \cite{LozanoTuilinoVerdu2005,tulino2005impact})
in the limiting regime (i.e. when $\mathrm{SNR}\rightarrow\infty$).
Our results reveal the fluctuation of $\mathcal{L}\left(\boldsymbol{H},\mbox{\ensuremath{\mathsf{SNR}}}\right)$
such that
\begin{equation}
\small\mathcal{L}\left(\boldsymbol{H},\mbox{\ensuremath{\mathsf{SNR}}}\right)=\begin{cases}
\left(\alpha-1\right)\log\left(\frac{\alpha-1}{\alpha}\right)+1+\\
\text{ }\text{ }\mathcal{O}\left(\frac{\small\mathsf{SNR}+\sqrt{\log n}}{n}+\frac{\log\small\mathsf{SNR}}{\sqrt{\small\mathsf{SNR}}}\right),\text{ if }\alpha\geq1\\
\\
\frac{1-\alpha}{\alpha}\log\left(1-\alpha\right)+\log\left(e\alpha\right)+\\
\text{ }\text{ }\mathcal{O}\left(\frac{\small\mathsf{SNR}+\sqrt{\log n}}{n}+\frac{\log\small\mathsf{SNR}}{\sqrt{\small\mathsf{SNR}}}\right),\text{ if }\alpha<1
\end{cases}\label{eq:PowerOffset}
\end{equation}
with high probability. The first-order term of $\mathcal{L}\left(\boldsymbol{H},\mbox{\ensuremath{\mathsf{SNR}}}\right)$
agrees with the known results on asymptotic power offset (e.g. \cite[Proposition 2]{LozanoTuilinoVerdu2005}\cite[Equation (84)]{tulino2005impact})
when $n\rightarrow\infty$ and $\mbox{\ensuremath{\mathsf{SNR}}}\rightarrow\infty$.
Our results distinguish from existing results in the sense that we
can accurately accommodate a much larger regime beyond the one with
asymptotically large SNR and channel dimensions.
\item The same mutual information and power offset values are shared among
a fairly general class of distributions even in the non-asymptotic
regime. Moreover, heavy-tailed distributions lead to well-controlled
information rate and power-offset as well, although the spectral measure
concentration is often less sharp than for sub-Gaussian distributions
(which ensure exponentially sharp concentration).
\item Finally, we note that Corollary \ref{corollary:CapacityIIDChannel-highSNR}
not only characterizes the order of the residual term, but also provides
reasonably accurate characterization of a narrow confidence interval
such that the mutual information and the power offset lies within
it, with high probability. All pre-constants are explicitly provided,
resulting in a full assessment of the mutual information and the power
offset. Our results do not rely on careful choice of growing matrix
sequences, which are typically required in those works based on asymptotic
laws.
\end{enumerate}

\subsubsection{Connection to the General Framework}

We now illustrate how the proof of Theorem \ref{thm:CapacityIIDChannel}
follows from the general framework presented in Section \ref{sec:General-Template}.
\begin{enumerate}
\item Note that the mutual information can be alternatively expressed as
\begin{align*}
 & \frac{1}{n_{\mathrm{r}}}C\left(\boldsymbol{H},\mbox{\ensuremath{\mathsf{SNR}}}\right)\\
 & \quad=\frac{1}{n_{\mathrm{r}}}\sum_{i=1}^{\min\left\{ n_{\mathrm{r}},n_{\mathrm{t}}\right\} }\log\left(1+\mbox{\ensuremath{\mathsf{SNR}}}\cdot\lambda_{i}\left(\frac{1}{n_{\mathrm{t}}}\boldsymbol{H}\boldsymbol{H}^{*}\right)\right)\\
 & \quad:=\frac{1}{n_{\mathrm{r}}}\sum_{i=1}^{\min\left\{ n_{\mathrm{r}},n_{\mathrm{t}}\right\} }\left\{ \log\mbox{\ensuremath{\mathsf{SNR}}}+f\left(\lambda_{i}\left(\frac{1}{n_{\mathrm{t}}}\boldsymbol{H}\boldsymbol{H}^{*}\right)\right)\right\} ,
\end{align*}
where $f\left(x\right):=\log\left(\frac{1}{\small\mathsf{SNR}}+x\right)$
on the domain $x\geq0$. As a result, the function $g(x):=f(x^{2})=\log\left(\frac{1}{\small\mathsf{SNR}}+x^{2}\right)$
has Lipschitz norm $\left\Vert g\right\Vert _{\mathcal{L}}\leq\small\mathsf{SNR}^{1/2}$.
\item For measures satisfying the LSI, Theorem \ref{thm:CapacityIIDChannel}(b)
immediately follows from Theorem \ref{thm:Deviation-Ef}(b). For measures
of bounded support or heavy-tailed measures, one can introduce the
following auxiliary function with respect to some small $\epsilon$:
\[
g_{\epsilon}(x):=\begin{cases}
\log\left(\epsilon+x^{2}\right), & \text{if }x\geq\sqrt{\epsilon},\\
\frac{1}{\sqrt{\epsilon}}\left(x-\sqrt{\epsilon}\right)+\log\left(2\epsilon\right),\quad & \text{if }0\leq x<\sqrt{\epsilon},
\end{cases}
\]
which is obtained by convexifying $g(x)$. If we set $f_{\epsilon}(x):=g_{\epsilon}(\sqrt{x})$,
then one can easily check that
\[
f_{\small\mathsf{SNR}^{-1}}\left(x\right)\leq f\left(x\right)\leq f_{\frac{e}{2}\small\mathsf{SNR}^{-1}}\left(x\right).
\]
Applying Theorem \ref{thm:Deviation-Ef}(a) allows us to derive the
concentration on 
\[
\frac{1}{n_{\mathrm{r}}}\sum_{i=1}^{\min\left\{ n_{\mathrm{r}},n_{\mathrm{t}}\right\} }f_{\small\mathsf{SNR}^{-1}}\left(\lambda_{i}\left(\frac{1}{n_{\mathrm{t}}}\boldsymbol{H}\boldsymbol{H}^{*}\right)\right)
\]
and
\[
\frac{1}{n_{\mathrm{r}}}\sum_{i=1}^{\min\left\{ n_{\mathrm{r}},n_{\mathrm{t}}\right\} }f_{\frac{e}{2}\small\mathsf{SNR}^{-1}}\left(\lambda_{i}\left(\frac{1}{n_{\mathrm{t}}}\boldsymbol{H}\boldsymbol{H}^{*}\right)\right),
\]
which in turn provide tight lower and upper bounds for $\frac{1}{n_{\mathrm{r}}}\sum_{i=1}^{\min\left\{ n_{\mathrm{r}},n_{\mathrm{t}}\right\} }f\left(\lambda_{i}\left(\frac{1}{n_{\mathrm{t}}}\boldsymbol{H}\boldsymbol{H}^{*}\right)\right)$.
\item Finally, the mean value $\mathbb{E}\left[\mathrm{det}\left(\boldsymbol{I}+\mbox{\ensuremath{\mathsf{SNR}}}\cdot\frac{1}{n_{\mathrm{t}}}\boldsymbol{H}\boldsymbol{H}^{*}\right)\right]$
admits a closed-form expression independent of precise entry distributions,
as derived in Lemma \ref{lemma:ExpDetIplustMM}.
\end{enumerate}

\subsection{MMSE Estimation in Random Vector Channels\label{sec:Applications:-MMSE}}

Next, we show that our analysis framework can be applied in Bayesian
inference problems that involve estimation of\emph{ }uncorrelated
signals components. Consider a simple linear vector channel model
\[
\boldsymbol{y}=\boldsymbol{H}\boldsymbol{x}+\boldsymbol{z},
\]
where $\boldsymbol{H}\in\mathbb{C}^{n\times p}$ is a known matrix,
$\boldsymbol{x}\in\mathbb{C}^{p}$ denotes an input signal with zero
mean and covariance matrix $P\boldsymbol{I}_{p}$, and $\boldsymbol{z}$
represents a zero-mean noise uncorrelated with $\boldsymbol{x}$,
which has covariance $\sigma^{2}\boldsymbol{I}_{n}$. Here, we denote
by $p$ and $n$ the input dimension and the sample size, respectively,
which agrees with the conventional notation adopted in the statistics
community. In this subsection, we assume that
\begin{equation}
\alpha:=\frac{p}{n}\leq1,
\end{equation}
i.e. the sample size exceeds the dimension of the input signal. The
goal is to estimate the channel input $\boldsymbol{x}$ from the channel
output $\boldsymbol{y}$ with minimum $\ell_{2}$ distortion.

The MMSE estimate of $\boldsymbol{x}$ given $\boldsymbol{y}$ can
be written as \cite{KaiSayHas2000}
\begin{align}
\hat{\boldsymbol{x}} & =\mathbb{E}\left[\boldsymbol{x}\mid\boldsymbol{y}\right]=\mathbb{E}\left[\boldsymbol{x}\boldsymbol{y}^{*}\right]\left(\mathbb{E}\left[\boldsymbol{y}\boldsymbol{y}^{*}\right]\right)^{-1}\boldsymbol{y}\nonumber \\
 & =P\boldsymbol{H}^{*}\left(\sigma^{2}\boldsymbol{I}_{n}+P\boldsymbol{H}\boldsymbol{H}^{*}\right)^{-1}\boldsymbol{y},\label{eq:MMSEestimator}
\end{align}
and the resulting MMSE is given by
\begin{align}
\small\mbox{\ensuremath{\mathsf{MMSE}}}\left(\boldsymbol{H},\mbox{\ensuremath{\mathsf{SNR}}}\right) & =\small\mathrm{tr}\left(P\boldsymbol{I}_{p}-P^{2}\boldsymbol{H}^{*}\left(\sigma^{2}\boldsymbol{I}_{n}+P\boldsymbol{H}\boldsymbol{H}^{*}\right)^{-1}\boldsymbol{H}\right).\label{eq:MMSE_defn}
\end{align}
The normalized MMSE (NMMSE) can then be written as
\begin{align}
 & \mbox{\ensuremath{\mathsf{NMMSE}}}\left(\boldsymbol{H},\mbox{\ensuremath{\mathsf{SNR}}}\right):=\frac{\mbox{\ensuremath{\mathsf{MMSE}}}\left(\boldsymbol{H},\small\mathsf{SNR}\right)}{\mathbb{E}\left[\left\Vert \boldsymbol{x}\right\Vert ^{2}\right]}\nonumber \\
 & \quad=\mathrm{tr}\left(\boldsymbol{I}_{p}-\boldsymbol{H}^{*}\left(\frac{1}{\small\mathsf{SNR}}\boldsymbol{I}_{n}+\boldsymbol{H}\boldsymbol{H}^{*}\right)^{-1}\boldsymbol{H}\right),
\end{align}
where
\[
\mbox{\ensuremath{\mathsf{SNR}}}:=\frac{P}{\sigma^{2}}.
\]

Using the concentration of measure technique, we can evaluate $\mbox{\ensuremath{\mathsf{NMMSE}}}\left(\boldsymbol{H}\right)$
to within a narrow interval with high probability, as stated in the
following theorem. 

\begin{theorem}\label{thm:MMSE-iid} Suppose that $\boldsymbol{H}=\boldsymbol{A}\boldsymbol{M}$,
where $\boldsymbol{A}\in\mathbb{C}^{m\times n}$ is a deterministic
matrix for some integer $m\geq p$, and $\boldsymbol{M}\in\mathbb{C}^{n\times p}$
is a random matrix such that $\boldsymbol{M}_{ij}$'s are independent
random variables satisfying $\mathbb{E}\left[\boldsymbol{M}_{ij}\right]=0$
and $\mathbb{E}[\left|\boldsymbol{M}_{ij}\right|^{2}]=\frac{1}{p}$. 

(a) If $\sqrt{p}\boldsymbol{M}_{ij}$'s are bounded by $D$, then
for any $\beta>8\sqrt{\pi}$,
\begin{align}
\small\frac{\mbox{\ensuremath{\mathsf{NMMSE}}}\left(\boldsymbol{H},\small\mathsf{SNR}\right)}{p}\in & \small\left[\mathcal{M}\left(\frac{8}{9\small\mathsf{SNR}},\boldsymbol{H}\right)+\frac{\beta\tau_{\mathrm{bd}}^{\mathrm{lb}}}{p},\right.\nonumber \\
 & \small\left.\quad\mathcal{M}\left(\frac{9}{8\small\mathsf{SNR}},\boldsymbol{H}\right)+\frac{\beta\tau_{\mathrm{bd}}^{\mathrm{ub}}}{p}\right]
\end{align}
with probability exceeding $1-8\exp\left(-\frac{\beta^{2}}{16}\right)$.

(b) If $\sqrt{p}\boldsymbol{M}_{ij}$'s satisfy the LSI with respect
to a uniform constant $c_{\mathrm{ls}}$, then for any $\beta>0$,
\begin{equation}
\small\frac{\mbox{\ensuremath{\mathsf{NMMSE}}}\left(\boldsymbol{H},\small\mathsf{SNR}\right)}{p}\in\mathcal{M}\left(\frac{1}{\small\mathsf{SNR}},\boldsymbol{H}\right)+\frac{\beta}{p}\left[\tau_{\mathrm{ls}}^{\mathrm{lb}},\tau_{\mathrm{ls}}^{\mathrm{ub}}\right]
\end{equation}
with probability exceeding $1-4\exp\left(-\frac{\beta^{2}}{2}\right)$.

(c) Suppose that $\sqrt{p}\boldsymbol{M}_{ij}$'s are independently
drawn from either sub-exponential distributions or heavy-tailed distributions
and that the distributions are symmetric about 0. Let $\tau_{c}$
be defined as in (\ref{eq:TruncateMProb}) with respect to $\sqrt{p}\boldsymbol{M}_{ij}$'s
for some sequence $c(n)$. Then
\begin{align}
\small\frac{\mbox{\ensuremath{\mathsf{NMMSE}}}\left(\boldsymbol{H},\small\mathsf{SNR}\right)}{p}\in\small & \left[\small\mathcal{M}\left(\frac{8}{9\small\mathsf{SNR}},\boldsymbol{A}\tilde{\boldsymbol{M}}\right)+\frac{\tau_{\mathrm{ht}}^{\mathrm{lb}}}{p},\right.\nonumber \\
 & \small\left.\quad\small\mathcal{M}\left(\frac{9}{8\small\mathsf{SNR}},\boldsymbol{A}\tilde{\boldsymbol{M}}\right)+\frac{\tau_{\mathrm{ht}}^{\mathrm{ub}}}{p}\right]
\end{align}
with probability exceeding $1-\frac{10}{p^{c(n)}}$, where $\tilde{\boldsymbol{M}}$
is defined such that $\tilde{\boldsymbol{M}}_{ij}:=\boldsymbol{M}_{ij}{\bf 1}{}_{\left\{ \sqrt{p}\left|\boldsymbol{M}_{ij}\right|<\tau_{c}\right\} }$. 

Here, the function $\mathcal{M}\left(\epsilon,\boldsymbol{H}\right)$
is defined as
\begin{equation}
\mathcal{M}\left(\delta,\boldsymbol{H}\right):=\frac{1}{p}\mathbb{E}\left[\mathrm{tr}\left(\left(\delta\boldsymbol{I}+\boldsymbol{H}^{*}\boldsymbol{H}\right)^{-1}\right)\right],\label{eq:MeExpression}
\end{equation}
and the residual terms are formally defined in Table \ref{tab:Summary-of-Preconstants-MMSE}.\end{theorem}

\begin{proof}See Appendix \ref{sec:Proof-of-Theorem-MMSE-iid}.\end{proof}

Theorem \ref{thm:MMSE-iid} ensures that the MMSE per signal component
falls within a small interval with high probability. This concentration
phenomenon again arises for a very large class of probability measures
including bounded distributions, sub-Gaussian distributions, and heavy-tailed
distributions. Note that $\mathcal{M}\left(\frac{9}{8\small\mathsf{SNR}},\boldsymbol{H}\right)$,
$\mathcal{M}\left(\frac{8}{9\small\mathsf{SNR}},\boldsymbol{H}\right)$,
and $\mathcal{M}\left(\frac{1}{\small\mathsf{SNR}},\boldsymbol{H}\right)$
are often very close to each other at moderate-to-high SNR (e.g. see
the analysis of Corollary \ref{cor:MMSE-iid-Gaussian}). The spans
of the residual intervals for both bounded and logarithmic Sobolev
measures do not exceed the order
\begin{equation}
\mathcal{O}\left(\frac{\beta\left\Vert \boldsymbol{A}\right\Vert \mbox{\ensuremath{\mathsf{SNR}}}^{1.5}}{p}\right),
\end{equation}
which is often negligible compared with the MMSE value per signal
component.

In general, we are not aware of a closed-form expression for $\mathcal{M}\left(\delta,\boldsymbol{H}\right)$
under various random matrix ensembles. If $\boldsymbol{H}$ is drawn
from a Gaussian ensemble, however, we are able to derive a simple
expression and bound this value to within a narrow confidence interval
with high probability, as follows. 

\begin{corollary}\label{cor:MMSE-iid-Gaussian} Suppose that $\boldsymbol{H}_{ij}\sim\mathcal{N}\left(0,\frac{1}{p}\right)$
are independent Gaussian random variables, and assume that $\alpha<1$.
Then
\begin{equation}
\frac{\mbox{\ensuremath{\mathsf{NMMSE}}}\left(\boldsymbol{H},\mbox{\ensuremath{\mathsf{SNR}}}\right)}{p}\in\frac{\alpha}{1-\alpha}+\left[\tau_{\mathrm{g}}^{\mathrm{lb}}+\frac{\beta\tau_{\mathrm{ls}}^{\mathrm{lb}}}{p},\tau_{\mathrm{g}}^{\mathrm{ub}}+\frac{\beta\tau_{\mathrm{ls}}^{\mathrm{ub}}}{p}\right]
\end{equation}
with probability exceeding $1-4\exp\left(-\frac{\beta^{2}}{2}\right)$.
Here, $\tau_{\mathrm{g}}^{\mathrm{lb}}$ and $\tau_{\mathrm{g}}^{\mathrm{ub}}$
are presented in Theorem \ref{thm:MMSE-iid}, while $\tau_{\mathrm{g}}^{\mathrm{lb}}$
and $\tau_{\mathrm{g}}^{\mathrm{ub}}$ are formally defined in Table
\ref{tab:Summary-of-Preconstants-MMSE}.\end{corollary}

\begin{proof}See Appendix \ref{sec:Proof-of-Theorem-MMSE-iid}.\end{proof}

\begin{table*}
\caption{\label{tab:Summary-of-Preconstants-MMSE}Summary of parameters of
Theorem \ref{thm:MMSE-iid} and Corollary \ref{cor:MMSE-iid-Gaussian}.}

\vspace{5pt}
\centering{}\footnotesize%
\begin{tabular}{>{\raggedright}p{0.11\linewidth}>{\raggedright}p{0.3\linewidth}>{\raggedright}p{0.3\linewidth}}
\hline 
Theorem \ref{thm:MMSE-iid} & $\tau_{\mathrm{bd}}^{\mathrm{lb}}:=-\frac{2\sqrt{2}D\left\Vert \boldsymbol{A}\right\Vert \small\mathsf{SNR}^{1.5}}{3\sqrt{3}}$ & $\tau_{\mathrm{bd}}^{\mathrm{ub}}:=\frac{3\sqrt{3}D\left\Vert \boldsymbol{A}\right\Vert \small\mathsf{SNR}^{1.5}}{8}$\tabularnewline
 & $\tau_{\mathrm{ls}}^{\mathrm{lb}}:=-\frac{3\sqrt{3}\sqrt{c_{\mathrm{ls}}}\left\Vert \boldsymbol{A}\right\Vert \small\mathsf{SNR}^{1.5}}{8}$ & $\tau_{\mathrm{ls}}^{\mathrm{ub}}:=-\frac{3\sqrt{3}\sqrt{c_{\mathrm{ls}}}\left\Vert \boldsymbol{A}\right\Vert \small\mathsf{SNR}^{1.5}}{8}$\tabularnewline
 & $\tau_{\mathrm{ht}}^{\mathrm{lb}}:=-\frac{8\sqrt{2}\tau_{c}\sigma_{c}\left\Vert \boldsymbol{A}\right\Vert \sqrt{c(n)\log n}\small\mathsf{SNR}^{1.5}}{3\sqrt{3}}$ & $\tau_{\mathrm{ht}}^{\mathrm{ub}}:=\frac{3\sqrt{3}\tau_{c}\sigma_{c}\left\Vert \boldsymbol{A}\right\Vert \sqrt{c(n)\log n}\small\mathsf{SNR}^{1.5}}{2}$\tabularnewline
 &  & \tabularnewline
Corollary \ref{cor:MMSE-iid-Gaussian}  & $\tau_{\mathrm{g}}^{\mathrm{ub}}:=\frac{\alpha}{n\left(1-\alpha-\frac{1}{n}\right)^{2}}$ & $\tau_{\mathrm{g}}^{\mathrm{lb}}:=-\frac{\alpha^{2}}{\small\mathsf{SNR}\left(1-\alpha-\frac{3}{n}\right)^{3}}$\tabularnewline
\hline 
\end{tabular}
\end{table*}

Corollary \ref{cor:MMSE-iid-Gaussian} implies that the MMSE per signal
component behaves like
\[
\small\frac{\mbox{\ensuremath{\mathsf{NMMSE}}}\left(\boldsymbol{H},\small\mathsf{SNR}\right)}{p}=\frac{\alpha}{1-\alpha}+\small\mathcal{O}\left(\frac{\small\mathsf{SNR}^{1.5}}{n}+\frac{1}{n}+\frac{1}{\small\mathsf{SNR}}\right)
\]
with high probability (by setting $\beta=\Theta\left(1\right)$).
Except for the residual term that vanishes in the presence of high
SNR and large signal dimensionality $n$, the first order term depends
only on the SNR and the oversampling ratio $\frac{1}{\alpha}$. Therefore,
the normalized MMSE converges to an almost deterministic value even
in the non-asymptotic regime. This illustrates the effectiveness of
the proposed analysis approach.

\subsubsection{Connection to the General Framework}

We now demonstrate how the proof of Theorem \ref{thm:MMSE-iid} follows
from the general template presented in Section \ref{sec:General-Template}.
\begin{enumerate}
\item Observe that the MMSE can be alternatively expressed as
\begin{align*}
\small\frac{\mbox{\ensuremath{\mathsf{NMMSE}}}\left(\boldsymbol{H},\small\mathsf{SNR}\right)}{p} & =\frac{1}{p}\sum_{i=1}^{\min\left\{ n,p\right\} }\frac{1}{\small\mathsf{SNR}^{-1}+\lambda_{i}\left(\boldsymbol{H}\boldsymbol{H}^{*}\right)}\\
 & :=\frac{1}{p}\sum_{i=1}^{\min\left\{ n,p\right\} }f\left(\lambda_{i}\left(\boldsymbol{H}\boldsymbol{H}^{*}\right)\right),
\end{align*}
where $f\left(x\right):=\frac{1}{\small\mathsf{SNR}^{-1}+x}$ $(x\geq0)$.
Consequently, the function $g(x):=f(x^{2})=\frac{1}{\small\mathsf{SNR}^{-1}+x^{2}}$
has Lipschitz norm $\left\Vert g\right\Vert _{\mathcal{L}}\leq\frac{3\sqrt{3}}{8}\small\mbox{\ensuremath{\mathsf{SNR}}}^{3/2}$.
\item For measures satisfying the LSI, Theorem \ref{thm:MMSE-iid}(b) can
be immediately obtained by Proposition \ref{thm:GeneralTemplate}(b).
For measures of bounded support or heavy-tailed measures, one can
introduce the following auxiliary function with respect to some small
$\epsilon$: 
\[
\tilde{g}_{\epsilon}(x):=\begin{cases}
\frac{1}{\epsilon^{2}+x^{2}}, & \text{if }x>\frac{1}{\sqrt{3}}\epsilon,\\
-\frac{3\sqrt{3}}{8}\left(x-\frac{\epsilon}{\sqrt{3}}\right)+\frac{3}{4\epsilon^{2}},\quad & \text{if }x\leq\frac{1}{\sqrt{3}}\epsilon,
\end{cases}
\]
which is obtained by convexifying $g(x)$. By setting $\tilde{f}_{\epsilon}(x):=\tilde{g}_{\epsilon}(\sqrt{x})$,
one has
\[
\tilde{f}_{\frac{3}{2\sqrt{2}}\small\mathsf{SNR}^{-1/2}}\left(x\right)\leq f\left(x\right)\leq\tilde{f}_{\small\mathsf{SNR}^{-1/2}}\left(x\right).
\]
Applying Proposition \ref{thm:GeneralTemplate}(a) gives the concentration
bounds on 
\[
\frac{1}{p}\sum_{i=1}^{\min\left\{ p,n\right\} }\tilde{f}_{\frac{3}{2\sqrt{2}}\small\mathsf{SNR}^{-1/2}}\left(\lambda_{i}\left(\boldsymbol{H}\boldsymbol{H}^{*}\right)\right)
\]
and
\[
\frac{1}{p}\sum_{i=1}^{\min\left\{ p,n\right\} }\tilde{f}_{\small\mathsf{SNR}^{-1/2}}\left(\lambda_{i}\left(\boldsymbol{H}\boldsymbol{H}^{*}\right)\right),
\]
which in turn provide tight sandwich bounds for $\frac{1}{p}\sum_{i=1}^{\min\left\{ p,n\right\} }f\left(\lambda_{i}\left(\boldsymbol{H}\boldsymbol{H}^{*}\right)\right)$.
\end{enumerate}

\section{Conclusions and Future Directions\label{sec:Conclusion}}

We have presented a general analysis framework that allows many canonical
MIMO system performance metrics taking the form of linear spectral
statistics to be assessed within a narrow concentration interval.
Moreover, we can guarantee that the metric value falls within the
derived interval with high probability, even under moderate channel
dimensionality. We demonstrate the effectiveness of the proposed framework
through two canonical metrics in wireless communications and signal
processing: mutual information and power offset of MIMO channels,
and MMSE for Gaussian processes. For each of these examples, our approach
allows a more refined and accurate characterization than those derived
in previous work in the asymptotic dimensionality limit. While our
examples here are presented for i.i.d. channel matrices, they can
all be immediately extended to accommodate dependent random matrices
using the same techniques. In other words, our analysis can be applied
to characterize various performance metrics of wireless systems under
correlated fading, as studied in \cite{ChuahTseKahnValenzuela2002}. 

Our work suggests a few future research directions for different matrix
ensembles. Specifically, if the channel matrix $\boldsymbol{H}$ cannot
be expressed as a linear correlation model $\boldsymbol{M}\boldsymbol{A}$
for some known matrix $\boldsymbol{A}$ and a random i.i.d. matrix
$\boldsymbol{M}$, then it will be of great interest to see whether
there is a systematic approach to analyze the concentration of spectral
measure phenomenon. For instance, if $\boldsymbol{H}$ consists of
independent sub-Gaussian columns but the elements within each column
are correlated (i.e. cannot be expressed as a weighted sum of i.i.d.
random variables), then it remains to be seen whether generalized
concentration of measure techniques (e.g. \cite{el2009concentration})
can be applied to find confidence intervals for the system performance
metrics. The spiked random matrix model \cite{passemier2014asymptotic},
which also admits simple asymptotic characterization, is another popular
candidate in modeling various system metrics. Another more challenging
case is when the rows of $\boldsymbol{H}$ are randomly drawn from
a known set of structured candidates. An example is a random DFT ensemble
where each row is independently drawn from a DFT matrix, which is
frequently encountered in the compressed sensing literature (e.g.
\cite{CandRomTao06}). Understanding the spectral concentration of
such random matrices could be of great interest in investigating OFDM
systems under random subsampling.

\section*{Acknowledgments}

The authors would like to thank the anonymous reviewers for extensive
and constructive comments that greatly improved the paper. 

\appendices

\section{Proof of Proposition \ref{thm:GeneralTemplate}\label{sec:Proof-of-Theorem-General-Template}}

Part (a) and (b) of Proposition \ref{thm:GeneralTemplate} are immediate
consequences of \cite[Corollary 1.8]{GuionnetZeitouni2000} via simple
algebraic manipulation. Here, we only provide the proof for heavy-tailed
distributions.

Define a new matrix $\tilde{\boldsymbol{M}}$ as a truncated version
of $\boldsymbol{M}$ such that
\[
\tilde{\boldsymbol{M}}_{ij}=\begin{cases}
\boldsymbol{M}_{ij},\quad & \text{if }\left|\boldsymbol{M}_{ij}\right|<\tau_{c},\\
0, & \text{otherwise}.
\end{cases}
\]
Note that the entries of $\tilde{\boldsymbol{M}}$ have zero mean
and variance not exceeding $\sigma_{c}^{2}$, and are uniformly bounded
in magnitude by $\tau_{c}$. Consequently, Theorem \ref{thm:GeneralTemplate}(a)
asserts that for any $\beta>8\sqrt{\pi}$, we have
\begin{equation}
\left|f_{0}\left(\tilde{\boldsymbol{M}}\right)-\mathbb{E}\left[f_{0}\left(\tilde{\boldsymbol{M}}\right)\right]\right|\leq\frac{\beta\tau_{c}\sigma_{c}\rho\left\Vert g\right\Vert _{\mathcal{L}}}{n}\label{eq:BoundedDeviation-truncated}
\end{equation}
with probability exceeding $1-4\exp\left(-\frac{\beta^{2}}{16}\right)$.

In addition, a simple union bound taken collectively with (\ref{eq:TruncateMProb})
yields 
\begin{align}
\mathbb{P}\left(\boldsymbol{M}\neq\tilde{\boldsymbol{M}}\right) & \leq\sum_{1\leq i\leq n,1\leq j\leq m}\mathbb{P}\left(\frac{1}{\nu_{ij}}\left|\boldsymbol{M}_{ij}\right|>\tau_{c}\right)\nonumber \\
 & \leq mn\mathbb{P}\left(\frac{1}{\nu_{ij}}\left|\boldsymbol{M}_{ij}\right|>\tau_{c}\right)=\frac{1}{n^{c(n)}}.\label{eq:HeavyTailMequalMtilde}
\end{align}
By setting $\beta=4\sqrt{c(n)\log n}$, one has
\[
4\exp\left(-\frac{\beta^{2}}{16}\right)=\frac{4}{n^{c(n)}}.
\]
This combined with (\ref{eq:BoundedDeviation-truncated}) and the
union bound implies that
\begin{equation}
\left|f_{0}\left(\boldsymbol{M}\right)-\mathbb{E}\left[f_{0}\left(\tilde{\boldsymbol{M}}\right)\right]\right|\leq\frac{4\sqrt{c(n)\log n}\tau_{c}\sigma_{c}\rho\left\Vert g\right\Vert _{\mathcal{L}}}{n}\label{eq:BoundedDeviation-truncated-1}
\end{equation}
holds with probability at least $1-5n^{-\mathrm{c}(n)}$, as claimed.

\section{Proof of Lemma \ref{thm:ComputeMeanSubExponential}\label{sec:Proof-of-Theorem-Compute-Mean-SubExponential}}

For notational simplicity, define
\begin{align*}
Y: & =n\left\{ f_{0}\left(\boldsymbol{M}\right)-\mathbb{E}\left[f_{0}\left(\boldsymbol{M}\right)\right]\right\} \\
 & =\sum_{i=1}^{\min\{m,n\}}\left\{ f\left(\lambda_{i}\left(\frac{1}{n}\boldsymbol{M}\boldsymbol{R}\boldsymbol{R}^{*}\boldsymbol{M}^{*}\right)\right)\right.\\
 & \quad\quad\quad\quad\quad-\left.\mathbb{E}\left[f\left(\lambda_{i}\left(\frac{1}{n}\boldsymbol{M}\boldsymbol{R}\boldsymbol{R}^{*}\boldsymbol{M}^{*}\right)\right)\right]\right\} 
\end{align*}
and
\[
Z:=nf_{0}\left(\boldsymbol{M}\right)=\sum_{i=1}^{\min\{m,n\}}f\left(\lambda_{i}\left(\frac{1}{n}\boldsymbol{M}\boldsymbol{R}\boldsymbol{R}^{*}\boldsymbol{M}^{*}\right)\right),
\]
and denote by $\mu_{\left|Y\right|}(y)$ the probability measure of
$\left|Y\right|$. 

(1) Suppose that $\mathbb{P}\left(\left|Y\right|>y\right)\leq c_{1}\exp\left(-c_{2}y\right)$
holds for some universal constants $c_{1}>0$ and $c_{2}>1$. Applying
integration by parts yields the following inequality
\begin{align}
\mathbb{E}\left[e^{Y}\right] & \leq\mathbb{E}\left[e^{\left|Y\right|}\right]={\displaystyle \int}_{0}^{\infty}e^{\left|Y\right|}\mathrm{d}\mu_{\left|Y\right|}(y)\nonumber \\
 & =-\left.e^{y}\mathbb{P}\left(\left|Y\right|>y\right)\right|_{0}^{\infty}+\int_{0}^{\infty}e^{y}\mathbb{P}\left(\left|Y\right|>y\right)\mathrm{d}y\nonumber \\
 & \leq1+c_{1}\int_{0}^{\infty}\exp\left(-\left(c_{2}-1\right)y\right)\mathrm{d}y\nonumber \\
 & =1+\frac{c_{1}}{c_{2}-1}.
\end{align}
This gives rise to the following bound
\[
\log\mathbb{E}\left[e^{Z}\right]-\log e^{\mathbb{E}\left[Z\right]}=\log\mathbb{E}\left[e^{Y}\right]\leq\log\left(1+\frac{c_{1}}{c_{2}-1}\right).
\]
Therefore, we can conclude
\begin{equation}
\frac{1}{n}\log\mathbb{E}\left[e^{Z}\right]-\frac{\log\left(1+\frac{c_{1}}{c_{2}-1}\right)}{n}\leq\frac{1}{n}\mathbb{E}\left[Z\right]\leq\frac{1}{n}\log\mathbb{E}\left[e^{Z}\right],
\end{equation}
where the last inequality follows from Jensen's inequality
\[
\frac{1}{n}\mathbb{E}\left[Z\right]=\frac{1}{n}\log e^{\mathbb{E}\left[Z\right]}\leq\frac{1}{n}\log\mathbb{E}e^{\left[Z\right]}.
\]

(2) Similarly, if $\mathbb{P}\left(\left|Y\right|>y\right)\leq c_{1}\exp\left(-c_{2}y^{2}\right)$
holds for some universal constants $c_{1}>0$ and $c_{2}>0$, then
one can bound
\begin{align}
\mathbb{E}\left[e^{Y}\right] & \leq-\left.e^{y}\mathbb{P}\left(\left|Y\right|>y\right)\right|_{0}^{\infty}+\int_{0}^{\infty}e^{y}\mathbb{P}\left(\left|Y\right|>y\right)\mathrm{d}y\nonumber \\
 & \leq1+c_{1}\int_{0}^{\infty}\exp\left(-c_{2}y^{2}+y\right)\mathrm{d}y\nonumber \\
 & \leq1+c_{1}\sqrt{\frac{\pi}{c_{2}}}\exp\left(\frac{1}{4c_{2}}\right)\nonumber \\
 & \quad\quad\quad\cdot{\displaystyle \int}_{-\infty}^{\infty}\sqrt{\frac{c_{2}}{\pi}}\exp\left(-c_{2}y^{2}+y-\frac{1}{4c_{2}}\right)\mathrm{d}y\nonumber \\
 & \leq1+c_{1}\sqrt{\frac{\pi}{c_{2}}}\exp\left(\frac{1}{4c_{2}}\right),
\end{align}
where the last inequality follows since $\sqrt{\frac{c_{2}}{\pi}}\exp\left(-c_{2}y^{2}+y-\frac{1}{4c_{2}}\right)$
is the pdf of $\mathcal{N}\left(\frac{1}{2c_{2}},\frac{2}{c_{2}}\right)$
and hence integrates to 1.

Applying the same argument as for the measure with sub-exponential
tails then leads to
\begin{align*}
\frac{1}{n}\log\mathbb{E}\left[e^{Z}\right] & \geq\frac{1}{n}\mathbb{E}\left[Z\right]\\
 & \geq\frac{1}{n}\log\mathbb{E}\left[e^{Z}\right]-\frac{\log\left(1+\sqrt{\frac{c_{1}^{2}\pi}{c_{2}}}\exp\left(\frac{1}{4c_{2}}\right)\right)}{n}.
\end{align*}

\section{Proof of Theorem \ref{thm:Deviation-Ef}\label{sec:Proof-of-Theorem-Deviation-Ef}}

For notational convenience, define
\begin{align*}
Z & :=\sum_{i=1}^{\min\{m,n\}}f\left(\lambda_{i}\left(\frac{1}{n}\boldsymbol{M}\boldsymbol{R}\boldsymbol{R}^{*}\boldsymbol{M}^{*}\right)\right),\\
Y & :=\text{ }Z-\mathbb{E}\left[Z\right].
\end{align*}

(1) Under the assumptions of Proposition \ref{thm:GeneralTemplate}(a),
the concentration result (\ref{eq:BoundedDeviation-corollary}) implies
that for any $y>8\sqrt{\pi}\rho D\left\Vert g\right\Vert _{\mathcal{L}}$,
one has
\[
\mathbb{P}\left(\left|Y\right|>y\right)\leq4\exp\left(-\frac{y^{2}}{8\kappa D^{2}\rho^{2}\nu^{2}\left\Vert g\right\Vert _{\mathcal{L}}^{2}}\right).
\]
For $y\leq8\sqrt{\pi}D\rho\nu\left\Vert g\right\Vert _{\mathcal{L}}$,
we can employ the trivial bound $\mathbb{P}\left(\left|Y\right|>y\right)\leq1$.
These two bounds taken collectively lead to a universal bound such
that for any $y\geq0$, one has
\[
\mathbb{P}\left(\left|Y\right|>y\right)\leq\exp\left(\frac{8\pi}{\kappa}\right)\exp\left(-\frac{y^{2}}{8\kappa D^{2}\rho^{2}\nu^{2}\left\Vert g\right\Vert _{\mathcal{L}}^{2}}\right).
\]
Lemma \ref{thm:ComputeMeanSubExponential} then suggests
\begin{align*}
 & \frac{1}{n}\log\mathcal{E}\left(f\right)\geq\frac{1}{n}\mathbb{E}\left[Z\right]\\
 & \text{ }\geq\frac{\log\mathcal{E}\left(f\right)}{n}-\frac{\log\left(1+\sqrt{8\kappa\pi}D\rho\nu\left\Vert g\right\Vert _{\mathcal{L}}e^{\frac{8\pi}{\kappa}+2\kappa D^{2}\rho^{2}\nu^{2}\left\Vert g\right\Vert _{\mathcal{L}}^{2}}\right)}{n}.
\end{align*}
This together with (\ref{eq:BoundedDeviation-corollary}) establishes
(\ref{eq:BoundedDeviation-Ef}).

(2) Under the assumptions of Proposition \ref{thm:GeneralTemplate}(b),
the concentration inequality (\ref{eq:BoundedDeviation-corollary})
asserts that for any $y>0$, we have
\[
\mathbb{P}\left(\left|Y\right|>y\right)\leq2\exp\left(-\frac{y^{2}}{\kappa c_{\mathrm{ls}}\rho^{2}\nu^{2}\left\Vert g\right\Vert _{\mathcal{L}}^{2}}\right).
\]
Applying Lemma \ref{thm:ComputeMeanSubExponential} then yields 
\begin{align*}
 & \frac{1}{n}\log\mathcal{E}\left(f\right)\geq\frac{1}{n}\mathbb{E}\left[Z\right]\\
 & \text{ }\geq\frac{\log\mathcal{E}\left(f\right)}{n}-\frac{\log\left(1+\sqrt{4\kappa\pi c_{\mathrm{ls}}\rho^{2}\nu^{2}\left\Vert g\right\Vert _{\mathcal{L}}^{2}}e^{\frac{\kappa}{4}c_{\mathrm{ls}}\rho^{2}\nu^{2}\left\Vert g\right\Vert _{\mathcal{L}}^{2}}\right)}{n}.
\end{align*}
This combined with (\ref{eq:LogSobolevDeviation-Corollary}) establishes
(\ref{eq:LogSobolevDeviation-Ef}).

(3) Under the assumptions of Proposition \ref{thm:GeneralTemplate}(c),
define
\begin{align*}
\tilde{Z} & :=\sum_{i=1}^{n}f\left(\lambda_{i}\left(\frac{1}{n}\tilde{\boldsymbol{M}}\boldsymbol{R}\boldsymbol{R}^{*}\tilde{\boldsymbol{M}}^{*}\right)\right),\\
\tilde{Y} & :=\tilde{Z}-\mathbb{E}\left[\tilde{Z}\right].
\end{align*}
Combining the concentration result (\ref{eq:BoundedDeviation-truncated})
with (\ref{eq:BoundedDeviation-Ef}) gives rise to
\begin{equation}
\begin{cases}
\frac{1}{n}\sum_{i=1}^{n}f\left(\lambda_{i}\left(\frac{1}{n}\tilde{\boldsymbol{M}}\boldsymbol{R}\boldsymbol{R}^{*}\tilde{\boldsymbol{M}}^{*}\right)\right)\\
\quad\quad\leq\frac{1}{n}\log\mathcal{E}_{\tilde{\boldsymbol{M}}}\left(f\right)+\frac{\beta\tau_{c}\sigma_{c}\rho\left\Vert g\right\Vert _{\mathcal{L}}}{n},\\
\\
\frac{1}{n}\sum_{i=1}^{n}f\left(\lambda_{i}\left(\frac{1}{n}\tilde{\boldsymbol{M}}\boldsymbol{R}\boldsymbol{R}^{*}\tilde{\boldsymbol{M}}^{*}\right)\right)\\
\quad\quad\geq\frac{1}{n}\log\mathcal{E}_{\tilde{\boldsymbol{M}}}\left(f\right)-\frac{\beta\tau_{c}\sigma_{c}\rho\left\Vert g\right\Vert _{\mathcal{L}}}{n}-\frac{c_{\rho,f,\tau_{c},\sigma_{c}}}{n},
\end{cases}\label{eq:HeavyTailDeviation-Ef-1}
\end{equation}
with probability exceeding $1-4\exp\left(-\frac{\beta^{2}}{8\kappa}\right)$.
We have shown in (\ref{eq:HeavyTailMequalMtilde}) that 
\[
\mathbb{P}\left(\boldsymbol{M}\neq\tilde{\boldsymbol{M}}\right)\leq\frac{1}{n^{c(n)}}
\]
and $4\exp\left(-\frac{\beta^{2}}{8\kappa}\right)=\frac{4}{n^{c(n)}}$
when $\beta=\sqrt{8\kappa c(n)\log n}$, concluding the proof via
the union bound.

\section{Proof of Lemma \ref{lemma:ExpDetIplustMM}\label{sec:Proof-of-Lemma-ExpDetIplustMM}}

For any $n\times n$ matrix $\boldsymbol{A}$ with eigenvalues $\lambda_{1},\cdots,\lambda_{n}$,
define the characteristic polynomial of $\boldsymbol{A}$ as
\begin{align}
p_{\boldsymbol{A}}(t) & =\det\left(t\boldsymbol{I}-\boldsymbol{A}\right)=t^{n}-S_{1}(\lambda_{1},\cdots,\lambda_{n})t^{n-1}+\cdots\nonumber \\
 & \quad\quad\quad\quad\quad\quad\quad\quad+(-1)^{n}S_{n}(\lambda_{1},\cdots,\lambda_{n}),
\end{align}
where $S_{l}(\lambda_{1},\cdots,\lambda_{n})$ is the $l$th elementary
symmetric polynomial defined as follows
\begin{equation}
S_{l}(\lambda_{1},\cdots,\lambda_{n}):=\sum_{1\leq i_{1}<\cdots<i_{l}\leq n}\prod_{j=1}^{l}\lambda_{i_{j}}.
\end{equation}
Let $E_{l}(\boldsymbol{A})$ represent the sum of determinants of
all $l$-by-$l$ principal minors of $\boldsymbol{A}$. It has been
shown in \cite[Theorem 1.2.12]{HornJohnson} that 
\begin{equation}
S_{l}(\lambda_{1},\cdots,\lambda_{n})=E_{l}(\boldsymbol{A}),\quad\quad1\leq l\leq n,
\end{equation}
which follows that
\begin{equation}
\det\left(t\boldsymbol{I}+\boldsymbol{A}\right)=t^{n}+E_{1}(\boldsymbol{A})t^{n-1}+\cdots+E_{n}(\boldsymbol{A}).\label{eq:DetIplusA_PrincipalMinor}
\end{equation}
On the other hand, for any principal minor $\left(\boldsymbol{M}\boldsymbol{M}^{*}\right)_{\boldsymbol{s},\boldsymbol{s}}$
of $\boldsymbol{M}\boldsymbol{M}^{*}$ coming from the rows and columns
at indices from $\boldsymbol{s}$ (where $\boldsymbol{s}\subseteq[n]$
is an index set), then one can show that
\begin{align}
\mathbb{E}\left[\det\left(\left(\boldsymbol{M}\boldsymbol{M}^{*}\right)_{\boldsymbol{s},\boldsymbol{s}}\right)\right] & =\mathbb{E}\left[\det\left(\boldsymbol{M}_{\boldsymbol{s},[m]}\boldsymbol{M}_{\boldsymbol{s},[m]}^{*}\right)\right]\nonumber \\
 & =\sum_{\boldsymbol{t}:\mathrm{card}\left(\boldsymbol{t}\right)=\mathrm{card}\left(\boldsymbol{s}\right)}\mathbb{E}\left[\det\left(\boldsymbol{M}_{\boldsymbol{s},\boldsymbol{t}}\boldsymbol{M}_{\boldsymbol{s},\boldsymbol{t}}^{*}\right)\right]\nonumber \\
 & ={m \choose \mathrm{card}(\boldsymbol{s})}\cdot\mathbb{E}\left[\det\left(\boldsymbol{M}_{\boldsymbol{s},\boldsymbol{s}}\boldsymbol{M}_{\boldsymbol{s},\boldsymbol{s}}^{*}\right)\right].\label{eq:EdetMM}
\end{align}

Consider now any i.i.d. random matrix $\boldsymbol{G}\in\mathbb{C}^{l\times l}$
such that $\mathbb{E}\left[\boldsymbol{G}_{ij}\right]=0$ and $\mathbb{E}[\left|\boldsymbol{G}_{ij}\right|^{2}]=1$.
If we denote by $\prod_{l}$ the permutation group of $l$ elements,
then the Leibniz formula for the determinant gives
\[
\det\left(\boldsymbol{G}\right)=\sum_{\sigma\in\prod_{l}}\text{sgn}(\sigma)\prod_{i=1}^{l}\boldsymbol{G}_{i,\sigma(i)}.
\]
Since $\boldsymbol{G}_{ij}$ are jointly independent, we have
\begin{align}
\mathbb{E}\left[\text{det}\left(\boldsymbol{G}\boldsymbol{G}^{*}\right)\right] & =\mathbb{E}\left[\left|\det\left(\boldsymbol{G}\right)\right|^{2}\right]\nonumber \\
 & =\sum_{\sigma\in\prod_{l}}\mathbb{E}\left[\prod_{i=1}^{l}\left|\boldsymbol{G}_{i,\sigma(i)}\right|^{2}\right]\nonumber \\
 & =\sum_{\sigma\in\prod_{l}}\prod_{i=1}^{l}\mathbb{E}\left[\left|\boldsymbol{G}_{i,\sigma(i)}\right|^{2}\right]=l!,\label{eq:EdetGG}
\end{align}
which is distribution-free. 

So far, we are already able to derive the closed-form expression for
$\mathbb{E}\left[\det\left(\epsilon\boldsymbol{I}+\frac{1}{m}\boldsymbol{M}\boldsymbol{M}^{*}\right)\right]$
by combining (\ref{eq:DetIplusA_PrincipalMinor}), (\ref{eq:EdetMM})
and (\ref{eq:EdetGG}) via straightforward algebraic manipulation. 

Alternatively, (\ref{eq:DetIplusA_PrincipalMinor}), (\ref{eq:EdetMM})
and (\ref{eq:EdetGG}) taken collectively indicate that $\mathbb{E}\left[\det\left(\epsilon\boldsymbol{I}+\frac{1}{m}\boldsymbol{M}\boldsymbol{M}^{*}\right)\right]$
is independent of precise distribution of the entries $\boldsymbol{M}_{ij}$'s.
As a result, we only need to evaluate $\mathbb{E}\left[\det\left(\epsilon\boldsymbol{I}+\frac{1}{m}\boldsymbol{M}\boldsymbol{M}^{*}\right)\right]$
under i.i.d. Gaussian random matrices. To this end, one can simply
cite the closed-form expression derived for Gaussian random matrices,
which has been reported in \cite[Theorem 2.13]{TulinoVerdu2004}.
This concludes the proof.

\section{Proof of Theorem \ref{thm:CapacityIIDChannel}\label{sec:Proof-of-Theorem-CapacityIIDChannels}}

When equal power allocation is adopted, the MIMO mutual information
per receive antenna is given by
\begin{align}
 & \frac{C\left(\boldsymbol{H},\small\mathsf{SNR}\right)}{n_{\mathrm{r}}}=\frac{1}{n_{\mathrm{r}}}\log\det\left(\boldsymbol{I}_{n_{\mathrm{r}}}+\frac{\small\mathsf{SNR}}{n_{\text{t}}}\boldsymbol{H}\boldsymbol{H}^{*}\right)\nonumber \\
= & \frac{1}{n_{\mathrm{r}}}\sum_{i=1}^{\min\{n_{\mathrm{r}},n_{\mathrm{t}}\}}\log\left(\frac{1}{\small\mathsf{SNR}}+\frac{1}{n_{\text{t}}}\lambda_{i}\left(\boldsymbol{H}\boldsymbol{H}^{*}\right)\right)\nonumber \\
 & \quad\quad\quad\quad\quad\quad\quad\quad+\min\{1,\alpha\}\log\mbox{\ensuremath{\mathsf{SNR}}},\nonumber \\
= & \begin{cases}
\frac{1}{n_{\mathrm{r}}}\log\det\left(\frac{1}{\small\mathsf{SNR}}\boldsymbol{I}_{n_{\mathrm{r}}}+\frac{1}{n_{\text{t}}}\boldsymbol{H}\boldsymbol{H}^{*}\right)+\log\mbox{\ensuremath{\mathsf{SNR}}},\quad & \text{if }\alpha\geq1\\
\frac{\alpha}{n_{\mathrm{t}}}\log\det\left(\frac{\alpha}{\small\mathsf{SNR}}\boldsymbol{I}_{n_{\mathrm{t}}}+\frac{1}{n_{\text{r}}}\boldsymbol{H}^{*}\boldsymbol{H}\right)+\alpha\log\frac{\small\mathsf{SNR}}{\alpha}, & \text{if }\alpha<1
\end{cases}\label{eq:CapExpand}
\end{align}
where $\alpha=n_{\mathrm{t}}/n_{\mathrm{r}}$ is assumed to be an
absolute constant. 

The first term of the capacity expression (\ref{eq:CapExpand}) in
both cases exhibits sharp measure concentration, as stated in the
following lemma. This in turn completes the proof of Theorem \ref{thm:CapacityIIDChannel}.

\begin{lem}\label{lemma-ConcentrationLogDetPlusConstant}Suppose
that $\alpha:=\frac{m}{n}\geq1$ is an absolute constant. Consider
a real-valued random matrix $\boldsymbol{M}=[\zeta_{ij}]_{1\leq i\leq n,1\leq j\leq m}$,
where $\zeta_{ij}$'s are jointly independent with zero mean and unit
variance.

(a) If $\zeta_{ij}$'s are bounded by $D$, then for any $\beta>8\sqrt{\pi}$,
one has
\begin{equation}
\begin{cases}
\frac{1}{n}\log\det\left(\epsilon\boldsymbol{I}+\frac{1}{m}\boldsymbol{M}\boldsymbol{M}^{*}\right)\leq\frac{1}{n}\log\mathcal{R}\left(\frac{e}{2}\epsilon,n,m\right)+\frac{D}{\sqrt{\frac{e}{2}\epsilon\alpha}}\frac{\beta}{n},\\
\\
\frac{1}{n}\log\det\left(\epsilon\boldsymbol{I}+\frac{1}{m}\boldsymbol{M}\boldsymbol{M}^{*}\right)\geq\frac{1}{n}\log\mathcal{R}\left(\frac{2}{e}\epsilon,n,m\right)-\frac{D}{\sqrt{\epsilon\alpha}}\frac{\beta}{n}\\
\quad\quad\quad\quad\quad\quad\quad\quad\quad\quad\quad\quad\quad-\frac{\log\left\{ 1+\sqrt{\frac{8\pi D^{2}}{\epsilon\alpha}}e^{8\pi+\frac{2D^{2}}{\epsilon\alpha}}\right\} }{n},
\end{cases}\label{eq:BoundsFepsilonBounded-Lemma}
\end{equation}
with probability exceeding $1-4\exp\left(-\frac{\beta^{2}}{8}\right)$.

(b) If $\zeta_{ij}$ satisfies the LSI with respect to a uniform constant
$c_{\mathrm{ls}}$, then for any $\beta>0$, one has
\begin{align}
\begin{cases}
\frac{1}{n}\log\det\left(\epsilon\boldsymbol{I}+\frac{1}{m}\boldsymbol{M}\boldsymbol{M}^{*}\right)\leq\frac{1}{n}\log\mathcal{R}\left(\epsilon,n,m\right)+\frac{\sqrt{c_{\mathrm{ls}}}}{\sqrt{\epsilon\alpha}}\frac{\beta}{n},\\
\\
\frac{1}{n}\log\det\left(\epsilon\boldsymbol{I}+\frac{1}{m}\boldsymbol{M}\boldsymbol{M}^{*}\right)\geq\frac{1}{n}\log\mathcal{R}\left(\epsilon,n,m\right)-\frac{\sqrt{c_{\mathrm{ls}}}}{\sqrt{\epsilon\alpha}}\frac{\beta}{n}\\
\quad\quad\quad\quad\quad\quad\quad\quad\quad\quad\quad\quad\quad-\frac{\log\left(1+\sqrt{\frac{4\pi c_{\mathrm{ls}}}{\epsilon\alpha}}e^{\frac{c_{\mathrm{ls}}}{4\epsilon\alpha}}\right)}{n},
\end{cases}\label{eq:BoundsFepsilonSubgaussian-Lemma}
\end{align}
with probability at least $1-2\exp\left(-\beta^{2}\right)$.

(c) If $\zeta_{ij}$'s are heavy-tailed distributed, then one has
\begin{equation}
\begin{cases}
\frac{1}{n}\log\det\left(\epsilon\boldsymbol{I}+\frac{1}{m}\boldsymbol{M}\boldsymbol{M}^{*}\right)\leq\frac{1}{n}\log\mathcal{R}\left(\frac{e\epsilon}{2\sigma_{c}^{2}},n,m\right)\\
\quad+2\log\sigma_{c}+\frac{\tau_{c}\sigma_{c}\sqrt{8c(n)\log n}}{\sqrt{\epsilon\alpha}n},\,\\
\\
\frac{1}{n}\log\det\left(\epsilon\boldsymbol{I}+\frac{1}{m}\boldsymbol{M}\boldsymbol{M}^{*}\right)\geq\frac{1}{n}\log\mathcal{R}\left(\frac{2\epsilon}{e\sigma_{c}^{2}},n,m\right)\\
\quad+2\log\sigma_{c}-\frac{\tau_{c}\sigma_{c}\sqrt{8c(n)\log n}}{\sqrt{\epsilon\alpha}n}-\frac{\log\left(1+\sqrt{\frac{8\pi\tau_{c}^{2}\sigma_{c}^{2}}{\epsilon\alpha}}e^{8\pi+\frac{2\tau_{c}^{2}\sigma_{c}^{2}}{\epsilon\alpha}}\right)}{n},
\end{cases}\label{eq:HeavyTailDeviation-Ef-Lemma}
\end{equation}
with probability exceeding $1-\frac{5}{n^{c(n)}}$.

Here, the function $\mathcal{R}\left(\epsilon\right)$ is defined
as
\[
\mathcal{R}\left(\epsilon,n,m\right):=\sum_{i=0}^{n}{n \choose i}\frac{\epsilon^{n-i}m^{-i}m!}{\left(m-i\right)!}.
\]
\end{lem}\begin{proof}[{\bf Proof of Lemma \ref{lemma-ConcentrationLogDetPlusConstant}}]Observe
that the derivative of $g(x):=\log\left(\epsilon+x^{2}\right)$ is
\[
g'(x)=\frac{2x}{\left(\epsilon+x^{2}\right)},\quad x\geq0,
\]
which is bounded within the interval $\left[0,\frac{1}{\sqrt{\epsilon}}\right]$
when $x\geq0$. Therefore, the Lipschitz norm of $g(x)$ satisfies
\[
\left\Vert g\right\Vert _{\mathcal{L}}\leq\frac{1}{\sqrt{\epsilon}}.
\]
The three class of probability measures are discussed separately as
follows.

(a) Consider first the measure uniformly bounded by $D$. In order
to apply Theorem \ref{thm:Deviation-Ef}, we need to first convexify
the objective metric. Define a function $g_{\epsilon}(x)$ such that
\[
g_{\epsilon}(x):=\begin{cases}
\log\left(\epsilon+x^{2}\right),\quad & \text{if }x\geq\sqrt{\epsilon},\\
\frac{1}{\sqrt{\epsilon}}\left(x-\sqrt{\epsilon}\right)+\log\left(2\epsilon\right),\quad & \text{if }0\leq x<\sqrt{\epsilon}.
\end{cases}
\]
Apparently, $g_{\epsilon}(x)$ is a concave function at $x\geq0$,
and its Lipschitz constant can be bounded above by
\[
\left\Vert g_{\epsilon}\right\Vert _{\mathcal{L}}\leq\frac{1}{\sqrt{\epsilon}}.
\]

By setting $f_{\epsilon}\left(x\right):=g_{\epsilon}\left(\sqrt{x}\right)$,
one can easily verify that
\[
\log\left(\frac{2}{e}\epsilon+x\right)\leq f_{\epsilon}(x)\leq\log\left(\epsilon+x\right),
\]
and hence
\begin{align}
\small\frac{1}{n}\log\det\left(\frac{2}{e}\epsilon\boldsymbol{I}+\frac{1}{m}\boldsymbol{M}\boldsymbol{M}^{*}\right) & \leq\small\frac{1}{n}\sum_{i=1}^{n}\left[f_{\epsilon}\left(\lambda_{i}\left(\frac{1}{m}\boldsymbol{M}\boldsymbol{M}^{*}\right)\right)\right]\nonumber \\
 & \leq\small\frac{1}{n}\log\det\left(\epsilon\boldsymbol{I}+\frac{1}{m}\boldsymbol{M}\boldsymbol{M}^{*}\right).\label{eq:BoundsFepsilon}
\end{align}
One desired feature of $f_{\epsilon}(x)$ is that $g_{\epsilon}(x)=f_{\epsilon}\left(x^{2}\right)$
is concave with finite Lipschitz norm bounded above by $\frac{1}{\sqrt{\epsilon}}$,
and is sandwiched between two functions whose \emph{exponential means}
are computable. 

By assumption, each entry of $\boldsymbol{M}$ is bounded by $D$,
and $\rho=\sqrt{\frac{n}{m}}=\frac{1}{\sqrt{\alpha}}$. Theorem \ref{thm:Deviation-Ef}
taken collectively with (\ref{eq:BoundsFepsilon}) suggests that
\begin{align}
 & \frac{1}{n}\sum_{i=1}^{n}f_{\epsilon}\left(\lambda_{i}\left(\frac{1}{m}\boldsymbol{M}\boldsymbol{M}^{*}\right)\right)\nonumber \\
 & \quad\leq\frac{1}{n}\log\mathbb{E}\left[\prod_{i=1}^{n}\exp\left\{ f_{\epsilon}\left(\lambda_{i}\left(\frac{1}{m}\boldsymbol{M}\boldsymbol{M}^{*}\right)\right)\right\} \right]+\frac{D}{\sqrt{\epsilon\alpha}}\frac{\beta}{n}\nonumber \\
 & \quad\leq\frac{1}{n}\log\mathbb{E}\left[\det\left(\epsilon\boldsymbol{I}+\frac{1}{m}\boldsymbol{M}\boldsymbol{M}^{*}\right)\right]+\frac{D}{\sqrt{\epsilon\alpha}}\frac{\beta}{n}\label{eq:BoundsFepsilonUB}
\end{align}
and
\begin{align}
 & \frac{1}{n}\sum_{i=1}^{n}f_{\epsilon}\left(\lambda_{i}\left(\frac{1}{m}\boldsymbol{M}\boldsymbol{M}^{*}\right)\right)\nonumber \\
 & \quad\geq\text{ }\frac{1}{n}\log\mathbb{E}\left[\prod_{i=1}^{n}\exp\left\{ f_{\epsilon}\left(\lambda_{i}\left(\frac{1}{m}\boldsymbol{M}\boldsymbol{M}^{*}\right)\right)\right\} \right]-\frac{D}{\sqrt{\epsilon\alpha}}\frac{\beta}{n}\nonumber \\
 & \quad\quad\quad\quad-\frac{\log\left\{ 1+\sqrt{\frac{8\pi D^{2}}{\epsilon\alpha}}e^{8\pi+\frac{2D^{2}}{\epsilon\alpha}}\right\} }{n}\nonumber \\
 & \quad\geq\text{ }\frac{1}{n}\log\mathbb{E}\left[\det\left(\frac{2}{e}\epsilon\boldsymbol{I}+\frac{1}{m}\boldsymbol{M}\boldsymbol{M}^{*}\right)\right]-\frac{D}{\sqrt{\epsilon\alpha}}\frac{\beta}{n}\nonumber \\
 & \quad\quad\quad\quad-\frac{\log\left\{ 1+\sqrt{\frac{8\pi D^{2}}{\epsilon\alpha}}e^{8\pi+\frac{2D^{2}}{\epsilon\alpha}}\right\} }{n}\label{eq:BoundsFepsilonBounded}
\end{align}
with probability exceeding $1-4\exp\left(-\frac{\beta^{2}}{8}\right)$.

To complete the argument for Part (a), we observe that (\ref{eq:BoundsFepsilon})
also implies
\begin{align*}
\sum_{i=1}^{n}f_{\epsilon}\left(\lambda_{i}\left(\frac{1}{m}\boldsymbol{M}\boldsymbol{M}^{*}\right)\right) & \leq\log\det\left(\epsilon\boldsymbol{I}+\frac{1}{m}\boldsymbol{M}\boldsymbol{M}^{*}\right)\\
 & \leq\sum_{i=1}^{n}f_{\frac{e}{2}\epsilon}\left(\lambda_{i}\left(\frac{1}{m}\boldsymbol{M}\boldsymbol{M}^{*}\right)\right).
\end{align*}
Substituting this into (\ref{eq:BoundsFepsilonUB}) and (\ref{eq:BoundsFepsilonBounded})
and making use of the identity given in Lemma \ref{lemma:ExpDetIplustMM}
complete the proof for Part (a).

(b) If the measure of $\boldsymbol{M}_{ij}$ satisfies LSI with uniformly
bounded constant $c_{\mathrm{ls}}$, then $g(x):=\log\left(\epsilon+x^{2}\right)$
is Lipschitz with Lipschitz bound $\frac{1}{\sqrt{\epsilon}}$. Applying
Theorem \ref{thm:Deviation-Ef} yields
\begin{align}
\begin{cases}
\small\frac{1}{n}\sum_{i=1}^{n}f\left(\lambda_{i}\left(\frac{1}{m}\boldsymbol{M}\boldsymbol{M}^{*}\right)\right)\leq\small\frac{1}{n}\log\mathbb{E}\left[\det\left(\epsilon\boldsymbol{I}+\frac{1}{m}\boldsymbol{M}\boldsymbol{M}^{*}\right)\right]\\
\quad\quad\quad\quad\quad\quad\quad\quad\quad\quad\quad\quad\quad+\frac{\sqrt{c_{\mathrm{ls}}}}{\sqrt{\epsilon\alpha}}\frac{\beta}{n},\\
\small\frac{1}{n}\sum_{i=1}^{n}f\left(\lambda_{i}\left(\frac{1}{m}\boldsymbol{M}\boldsymbol{M}^{*}\right)\right)\geq\small\frac{1}{n}\log\mathbb{E}\left[\det\left(\epsilon\boldsymbol{I}+\frac{1}{m}\boldsymbol{M}\boldsymbol{M}^{*}\right)\right]\\
\quad\quad\quad\quad\quad\quad\quad\quad\quad\quad\quad-\frac{\sqrt{c_{\mathrm{ls}}}}{\sqrt{\epsilon\alpha}}\frac{\beta}{n}-\frac{\log\left(1+\sqrt{\frac{4\pi c_{\mathrm{ls}}}{\epsilon\alpha}}e^{\frac{c_{\mathrm{ls}}}{4\epsilon\alpha}}\right)}{n},
\end{cases}\label{eq:BoundsFepsilonSubgaussian}
\end{align}
with probability exceeding $1-2\exp\left(-\beta^{2}\right)$. 

(c) If the measure of $\boldsymbol{M}_{ij}$ is heavy-tailed, then
we again need to use the sandwich bound between $f_{\epsilon}(x)$
and $f(x)$. Theorem \ref{thm:Deviation-Ef} indicates that
\begin{equation}
\begin{cases}
\small\frac{1}{n}\sum_{i=1}^{n}f_{\epsilon}\left(\lambda_{i}\left(\frac{1}{m}\boldsymbol{M}\boldsymbol{M}^{*}\right)\right)\leq\small\frac{1}{n}\log\mathbb{E}\left[\det\left(\epsilon\boldsymbol{I}+\frac{1}{m}\tilde{\boldsymbol{M}}\tilde{\boldsymbol{M}}^{*}\right)\right]\\
\quad\quad\quad\quad\quad\quad\,+\frac{\tau_{c}\sigma_{c}\sqrt{8c(n)\log n}}{\sqrt{\epsilon\alpha}n},\\
\small\frac{1}{n}\sum_{i=1}^{n}f_{\epsilon}\left(\lambda_{i}\left(\frac{1}{m}\boldsymbol{M}\boldsymbol{M}^{*}\right)\right)\geq\small\frac{1}{n}\log\mathbb{E}\left[\det\left(\frac{2}{e}\epsilon\boldsymbol{I}+\frac{1}{m}\tilde{\boldsymbol{M}}\tilde{\boldsymbol{M}}^{*}\right)\right]\\
\quad\quad-\frac{\tau_{c}\sigma_{c}\sqrt{8c(n)\log n}}{\sqrt{\epsilon\alpha}n}-\frac{\log\left(1+\sqrt{\frac{8\pi\tau_{c}^{2}\sigma_{c}^{2}}{\epsilon\alpha}}e^{8\pi+\frac{2\tau_{c}^{2}\sigma_{c}^{2}}{\epsilon\alpha}}\right)}{n},
\end{cases}\label{eq:HeavyTailDeviation-Ef-2}
\end{equation}
with probability exceeding $1-\frac{5}{n^{c(n)}}$. The only other
difference with the proof of Part (a) is that the entries of $\tilde{\boldsymbol{M}}$
has variance $\sigma_{c}^{2}$, and hence
\begin{align*}
 & \frac{1}{n}\log\mathbb{E}\left[\det\left(\epsilon\boldsymbol{I}+\frac{1}{m}\tilde{\boldsymbol{M}}\tilde{\boldsymbol{M}}^{*}\right)\right]\\
 & \quad=2\log\sigma_{c}+\frac{1}{n}\log\mathbb{E}\left[\det\left(\frac{\epsilon}{\sigma_{c}^{2}}\boldsymbol{I}+\frac{1}{m}\boldsymbol{M}\boldsymbol{M}^{*}\right)\right].
\end{align*}

Finally, the proof is complete by applying Lemma \ref{lemma:ExpDetIplustMM}
on $\mathbb{E}\left[\det\left(\epsilon\boldsymbol{I}+\frac{1}{m}\boldsymbol{M}\boldsymbol{M}^{*}\right)\right]$.

\end{proof}

\section{Proof of Lemma \ref{Lemma:Bound-LogEAA}\label{sec:Proof-of-Lemma-Bound-LogEAA}}

Using the results derived in Lemma \ref{lemma:ExpDetIplustMM}, one
can bound 
\begin{align*}
 & \mathbb{E}\left[\det\left(\epsilon\boldsymbol{I}+\frac{1}{m}\boldsymbol{M}\boldsymbol{M}^{*}\right)\right]=m^{-n}\sum_{i=0}^{n}{n \choose i}\frac{m!}{\left(m-i\right)!}\epsilon^{n-i}m^{n-i}\\
 & \quad=m!m^{-n}\sum_{i=0}^{n}{n \choose i}\frac{\epsilon^{i}m^{i}}{\left(m-n+i\right)!}\\
 & \quad\leq m!m^{-n}\left(n+1\right)\max_{i:0\leq i\leq n}{n \choose i}\frac{\epsilon^{i}m^{i}}{\left(m-n+i\right)!}.
\end{align*}
Denote by $i_{\max}$ the index of the largest term as follows 
\[
i_{\max}:=\arg\max_{i:0\leq i\leq n}{n \choose i}\frac{\epsilon^{i}m^{i}}{\left(m-n+i\right)!}.
\]

Suppose for now that $i_{\max}=\delta n$ for some numerical value
$\delta$, then we can bound
\begin{align}
 & \frac{1}{n}\log\mathbb{E}\left[\det\left(\epsilon\boldsymbol{I}+\frac{1}{m}\boldsymbol{M}\boldsymbol{M}^{*}\right)\right]\nonumber \\
 & \quad\leq\text{ }\frac{\log m!}{n}+\frac{\log\left(n+1\right)}{n}+\frac{\log{n \choose \delta n}}{n}+\delta\log\epsilon\nonumber \\
 & \quad\quad\quad\quad-\left(1-\delta\right)\log m-\frac{\log\left(m-n+\delta n\right)!}{n}\nonumber \\
 & \quad=\text{ }\frac{\log\frac{m!}{\left(m-n+\delta n\right)!\left(n-\delta n\right)!}}{n}+\frac{\log\left(n+1\right)}{n}+\frac{\log{n \choose \delta n}}{n}\nonumber \\
 & \quad\quad\quad\quad+\frac{\log\left(n-\delta n\right)!}{n}-\left(1-\delta\right)\log m+\delta\log\epsilon\nonumber \\
 & \quad=\text{ }\frac{\alpha\log{m \choose n-\delta n}}{m}+\frac{\log\left(n+1\right)}{n}+\frac{\log{n \choose \delta n}}{n}+\frac{\log\left(n-\delta n\right)!}{n}\nonumber \\
 & \quad\quad\quad\quad-\left(1-\delta\right)\log m+\delta\log\epsilon.\label{eq:UBlogEdet}
\end{align}
It follows immediately from the well-known entropy formula (e.g. \cite[Example 11.1.3]{cover2012elements})
that 
\begin{equation}
\mathcal{H}\left(\frac{k}{n}\right)-\frac{\log\left(n+1\right)}{n}\leq\frac{1}{n}\log{n \choose k}\leq\mathcal{H}\left(\frac{k}{n}\right),\label{eq:EntropyBound}
\end{equation}
where $\mathcal{H}(x):=x\log\frac{1}{x}+\left(1-x\right)\log\frac{1}{1-x}$
denotes the binary entropy function. Also, the famous Stirling approximation
bounds asserts that
\[
\sqrt{2\pi}n^{n+\frac{1}{2}}e^{-n}\leq n!\leq en^{n+\frac{1}{2}}e^{-n},
\]
\begin{align}
\Longleftrightarrow\text{ }\frac{2+\log n}{2n}+\log n-1 & \geq\frac{1}{n}\log n!\nonumber \\
 & \geq\frac{\log\left(2\pi\right)+\log n}{2n}+\log n-1.\label{eq:StirlingBound}
\end{align}
Substituting (\ref{eq:EntropyBound}) and (\ref{eq:StirlingBound})
into (\ref{eq:UBlogEdet}) leads to
\begin{align*}
 & \frac{1}{n}\log\mathbb{E}\left[\det\left(\epsilon\boldsymbol{I}+\frac{1}{m}\boldsymbol{M}\boldsymbol{M}^{*}\right)\right]\\
 & \quad\leq\text{ }\alpha\mathcal{H}\left(\frac{(1-\delta)n}{m}\right)+\frac{\log\left(n+1\right)}{n}+\mathcal{H}\left(\delta\right)\\
 & \quad\quad\quad+\frac{\log e+\frac{1}{2}\log\left[\left(1-\delta\right)n\right]}{n}+\left(1-\delta\right)\log\left[\left(1-\delta\right)n\right]\\
 & \quad\quad\quad-\left(1-\delta\right)\log e-\left(1-\delta\right)\log m+\delta\log\epsilon\\
 & \quad=\text{ }\text{ }\alpha\mathcal{H}\left(\frac{1-\delta}{\alpha}\right)+\frac{\log\left(n+1\right)+1+\frac{1}{2}\log\left[\left(1-\delta\right)n\right]}{n}\\
 & \quad\quad\quad+\mathcal{H}\left(\delta\right)+\left(1-\delta\right)\log\frac{1}{\alpha}+\left(1-\delta\right)\log(1-\delta)\\
 & \quad\quad\quad-\left(1-\delta\right)+\delta\log\epsilon.
\end{align*}
Making use of the inequality 
\[
\log\left(n+1\right)+1+\frac{1}{2}\log n\leq1.5\log\left(en\right),
\]
one obtains 
\begin{align}
 & \frac{1}{n}\log\mathbb{E}\left[\det\left(\epsilon\boldsymbol{I}+\frac{1}{m}\boldsymbol{M}\boldsymbol{M}^{*}\right)\right]\nonumber \\
 & \quad\leq\text{ }\alpha\mathcal{H}\left(\frac{1-\delta}{\alpha}\right)+\frac{1.5\log\left(en\right)}{n}+\delta\log\frac{1}{\delta}\nonumber \\
 & \quad\quad\quad-\left(1-\delta\right)\log\left(e\alpha\right)+\delta\log\epsilon\nonumber \\
 & \quad=\text{ }\alpha\mathcal{H}\left(\frac{1}{\alpha}\right)-\log\left(\alpha e\right)+\frac{1.5\log\left(en\right)}{n}+\delta\log\frac{1}{\delta}\nonumber \\
 & \quad\quad\quad+\delta\log\left(e\alpha\right)+\delta\log\epsilon+\alpha\footnotesize\left[\footnotesize\mathcal{H}\left(\frac{1-\delta}{\alpha}\right)-\footnotesize\mathcal{H}\left(\frac{1}{\alpha}\right)\right]\nonumber \\
 & \quad\leq\text{ }\left(\alpha-1\right)\log\left(\frac{\alpha}{\alpha-1}\right)-1+\frac{1.5\log\left(en\right)}{n}+\delta\log\frac{1}{\delta}\nonumber \\
 & \quad\quad\quad+\delta\log\left(e\alpha\right)+\delta\log\epsilon+\mathcal{H}\left(\delta\right),\label{eq:UBLogDetAA}
\end{align}
where the last inequality follows from the following identity on binary
entropy function \cite{verdu2014IT}: for any $p,q$ obeying $0\leq p,q\leq1$
and $0<pq<1$, one has 
\[
\mathcal{H}\left(p\right)-\mathcal{H}\left(pq\right)=\left(1-pq\right)\mathcal{H}\left(\frac{p-pq}{1-pq}\right)-p\mathcal{H}(q),
\]
\[
\Rightarrow\quad\mathcal{H}\left(\frac{1-\delta}{\alpha}\right)-\mathcal{H}\left(\frac{1}{\alpha}\right)\leq\frac{1}{\alpha}\mathcal{H}(1-\delta)=\frac{1}{\alpha}\mathcal{H}(\delta).
\]

It remains to estimate the index $i_{\max}$ or, equivalently, the
value $\delta$. If we define
\[
s_{\epsilon}\left(i\right):={n \choose i}\frac{\epsilon^{i}m^{i}}{\left(m-n+i\right)!}\quad\text{and}\quad r_{\epsilon}\left(i\right):=\frac{s_{\epsilon}\left(i+1\right)}{s_{\epsilon}\left(i\right)},
\]
then one can compute
\[
r_{\epsilon}\left(i\right)=\frac{\frac{n!}{\left(i+1\right)!\left(n-i-1\right)!}\frac{\epsilon^{i+1}m^{i+1}}{\left(m-n+i+1\right)!}}{\frac{n!}{i!\left(n-i\right)!}\frac{\epsilon^{i}m^{i}}{\left(m-n+i\right)!}}=\frac{\epsilon m\left(n-i\right)}{\left(i+1\right)\left(m-n+i+1\right)},
\]
which is a decreasing function in $i$. Suppose that $n\epsilon>\max\left\{ 4,2\left(1-\frac{1}{\alpha}+\frac{1}{\alpha n}\right)\right\} $.
By setting $r_{\epsilon}\left(x\right)=1$, we can derive a valid
positive solution $x_{0}$ as follows
\begin{align*}
x_{0} & =\small\frac{-\left(m+n+2+\epsilon m\right)}{2}\\
 & \quad\quad\small+\frac{\sqrt{\left(m+n+2+\epsilon m\right)^{2}-4\left(m-n+1\right)+4\epsilon mn}}{2}.
\end{align*}
Simple algebraic manipulation yields
\begin{align*}
x_{0} & <\small\frac{-\left(m+n+2+\epsilon m\right)+\left(m+n+2+\epsilon m\right)+2\sqrt{\epsilon mn}}{2}\\
 & \leq\text{ }\sqrt{\epsilon\alpha}n,
\end{align*}
and
\begin{align*}
x_{0} & >\small\frac{-\left(m+n+2+\epsilon m\right)+\sqrt{\left(m+n+2+\epsilon m\right)^{2}+2\epsilon mn}}{2}\\
 & =\small\frac{\epsilon mn}{\left(m+n+2+\epsilon m\right)+\sqrt{\left(m+n+2+\epsilon m\right)^{2}+2\epsilon mn}}\\
 & =\small\frac{\alpha\epsilon n}{\alpha+1+\frac{2}{n}+\alpha\epsilon+\sqrt{\left(\alpha+1+\frac{2}{n}+\alpha\epsilon\right)^{2}+2\epsilon\alpha}}\\
 & \geq\small\frac{\alpha}{2\left(\alpha+1+\frac{2}{n}+\alpha\right)+\sqrt{2\alpha}}\epsilon n\\
 & >\frac{\alpha}{5\left(\alpha+1\right)}\epsilon n.
\end{align*}
Therefore, we can conclude that $\delta\in\left[\frac{\alpha}{5\left(1+\alpha\right)}\epsilon,\sqrt{\alpha\epsilon}\right]$.
Assume that $\epsilon<\frac{1}{e^{2}\alpha}$. Substituting this into
(\ref{eq:UBLogDetAA}) then leads to
\[
\small\frac{1}{n}\log\mathbb{E}\left[\det\left(\epsilon\boldsymbol{I}+\frac{1}{m}\boldsymbol{M}\boldsymbol{M}^{*}\right)\right]\leq\left(\alpha-1\right)\log\left(\frac{\alpha}{\alpha-1}\right)-1+\mathcal{R}_{n,\epsilon},
\]
where $\mathcal{R}_{n,\epsilon}$ denotes the residual term
\[
\mathcal{R}_{n,\epsilon}:=\frac{1.5\log\left(en\right)}{n}+\sqrt{\alpha\epsilon}\left(\log\frac{1}{\sqrt{\alpha\epsilon}}+\log\left(e\alpha\right)\right)+\mathcal{H}\left(\sqrt{\alpha\epsilon}\right).
\]
In particular, if we further have $\epsilon<\min\left\{ \frac{1}{e^{2}\alpha^{3}},\frac{1}{4\alpha}\right\} $,
then one has $\log\alpha e<\log\frac{1}{\sqrt{\alpha\epsilon}}$ and
\[
\mathcal{H}\left(\sqrt{\alpha\epsilon}\right)\leq2\sqrt{\alpha\epsilon}\log\frac{1}{\sqrt{\alpha\epsilon}},
\]
which in turn leads to
\begin{align*}
\mathcal{R}_{n,\epsilon} & \leq\frac{1.5\log\left(en\right)}{n}+4\sqrt{\alpha\epsilon}\log\frac{1}{\sqrt{\alpha\epsilon}}.
\end{align*}

On the other hand, the lower bound on $\frac{1}{n}\log\mathbb{E}\left[\det\left(\epsilon\boldsymbol{I}+\frac{1}{m}\boldsymbol{M}\boldsymbol{M}^{*}\right)\right]$
can be easily obtained through the following argument
\begin{align*}
\mathbb{E}\left[\det\left(\epsilon\boldsymbol{I}+\frac{1}{m}\boldsymbol{M}\boldsymbol{M}^{*}\right)\right] & \geq\mathbb{E}\left[\det\left(\frac{1}{m}\boldsymbol{M}\boldsymbol{M}^{*}\right)\right]\\
 & =\frac{m!}{\left(m-n\right)!m^{n}},
\end{align*}
and hence
\begin{align}
 & \frac{1}{n}\log\mathbb{E}\left[\det\left(\epsilon\boldsymbol{I}+\frac{1}{m}\boldsymbol{M}\boldsymbol{M}^{*}\right)\right]\nonumber \\
 & \quad\geq\text{ }\frac{1}{n}\log{m \choose n}+\frac{\log n!}{n}-\log m\nonumber \\
 & \quad\geq\alpha\mathcal{H}\left(\frac{1}{\alpha}\right)-\frac{\log\left(m+1\right)}{n}+\frac{\log\left(2\pi\right)+\frac{1}{2}\log n}{n}\nonumber \\
 & \quad\quad\quad\quad\quad\quad\quad+\log\left(\frac{n}{m}\right)-1\label{eq:IntermediateStep}\\
 & \quad\geq\text{ }\left(\alpha-1\right)\log\left(\frac{\alpha}{\alpha-1}\right)-1-\frac{\log\frac{1}{2\pi}\left(m+1\right)-\frac{1}{2}\log n}{n},
\end{align}
where (\ref{eq:IntermediateStep}) makes use of the bounds (\ref{eq:EntropyBound})
and (\ref{eq:StirlingBound}). This concludes the proof.

\section{Proof of Theorem \ref{thm:MMSE-iid} and Corollary \ref{cor:MMSE-iid-Gaussian}\label{sec:Proof-of-Theorem-MMSE-iid}}

\begin{proof}[{\bf Proof of Theorem \ref{thm:MMSE-iid}}]Suppose that
the singular value decomposition of $\boldsymbol{H}$ can be expressed
by $\boldsymbol{H}=\boldsymbol{U}\boldsymbol{\Lambda}\boldsymbol{V}^{*}$
where $\boldsymbol{\Lambda}$ is a diagonal matrix consist of all
singular values of $\boldsymbol{H}$. Then the MMSE per input component
can be expressed as 
\begin{align*}
 & \frac{\mbox{\ensuremath{\mathsf{NMMSE}}}\left(\boldsymbol{H},\small\mathsf{SNR}\right)}{p}\\
 & \quad=1-\frac{1}{p}\mathrm{tr}\left(\boldsymbol{V}\boldsymbol{\Lambda}\boldsymbol{U}^{*}\left(\frac{1}{\small\mathsf{SNR}}\boldsymbol{I}+\boldsymbol{U}\boldsymbol{\Lambda}^{2}\boldsymbol{U}^{*}\right)^{-1}\boldsymbol{U}\boldsymbol{\Lambda}\boldsymbol{V}^{*}\right)\\
 & \quad=1-\frac{1}{p}\mathrm{tr}\left(\boldsymbol{\Lambda}\left(\frac{1}{\small\mathsf{SNR}}\boldsymbol{I}+\boldsymbol{\Lambda}^{2}\right)^{-1}\boldsymbol{\Lambda}\right)\\
 & \quad=\frac{1}{p}\sum_{i=1}^{p}\frac{1}{\small\mathsf{SNR}^{-1}+\lambda_{i}\left(\boldsymbol{H}^{*}\boldsymbol{H}\right)}.
\end{align*}

(i) Note that $g_{\epsilon}(x):=\frac{1}{\epsilon^{2}+x^{2}}$ is
not a convex function. In order to apply Proposition \ref{thm:GeneralTemplate},
we define a convexified variant $\tilde{g}_{\epsilon}(x)$ of $g_{\epsilon}(x)$.
Specifically, we set
\[
\tilde{g}_{\epsilon}(x):=\begin{cases}
\frac{1}{\epsilon^{2}+x^{2}},\quad & \text{if }x>\frac{1}{\sqrt{3}}\epsilon,\\
-\frac{3\sqrt{3}}{8\epsilon^{3}}\left(x-\frac{\epsilon}{\sqrt{3}}\right)+\frac{3}{4\epsilon^{2}},\quad & \text{if }x\leq\frac{1}{\sqrt{3}}\epsilon,
\end{cases}
\]
which satisfies
\[
\left\Vert \tilde{g}_{\epsilon}(x)\right\Vert _{\mathcal{L}}\leq\frac{3\sqrt{3}}{8\epsilon^{3}}.
\]
One can easily check that
\[
g_{\frac{3}{2\sqrt{2}}\epsilon}(x)\leq\tilde{g}_{\frac{3}{2\sqrt{2}}\epsilon}(x)\leq g_{\epsilon}(x)\leq\tilde{g}_{\epsilon}(x)\leq g_{\frac{2\sqrt{2}}{3}\epsilon}(x).
\]
This implies that we can sandwich the target function as follows
\begin{align}
 & \frac{1}{p}\sum_{i=1}^{p}\frac{1}{\epsilon^{2}+\lambda_{i}\left(\boldsymbol{H}^{*}\boldsymbol{H}\right)}\nonumber \\
 & \text{ }\text{ }\text{ }\leq\frac{1}{p}\sum_{i=1}^{p}\left(\tilde{g}_{\epsilon}\left(\sqrt{\lambda_{i}\left(\boldsymbol{H}^{*}\boldsymbol{H}\right)}\right)-\mathbb{E}\left[\tilde{g}_{\epsilon}\left(\sqrt{\lambda_{i}\left(\boldsymbol{H}^{*}\boldsymbol{H}\right)}\right)\right]\right)\nonumber \\
 & \quad\quad\quad\quad+\mathbb{E}\left[g_{\frac{2\sqrt{2}}{3}\epsilon}\left(\sqrt{\lambda_{i}\left(\boldsymbol{H}^{*}\boldsymbol{H}\right)}\right)\right],\label{eq:MMSE_ub}
\end{align}
and
\begin{align}
 & \small\frac{1}{p}\sum_{i=1}^{p}\frac{1}{\epsilon^{2}+\lambda_{i}\left(\boldsymbol{H}^{*}\boldsymbol{H}\right)}\label{eq:MMSE_lb}\\
 & \small\text{ }\text{ }\geq\text{ }\frac{1}{p}\sum_{i=1}^{p}\left(\tilde{g}_{\frac{3}{2\sqrt{2}}\epsilon}\left(\sqrt{\lambda_{i}\left(\boldsymbol{H}^{*}\boldsymbol{H}\right)}\right)-\mathbb{E}\left[\tilde{g}_{\frac{3}{2\sqrt{2}}\epsilon}\left(\sqrt{\lambda_{i}\left(\boldsymbol{H}^{*}\boldsymbol{H}\right)}\right)\right]\right)\nonumber \\
 & \small\quad\quad\quad\quad+\mathbb{E}g_{\frac{3}{2\sqrt{2}}\epsilon}\left(\sqrt{\lambda_{i}\left(\boldsymbol{H}^{*}\boldsymbol{H}\right)}\right).\nonumber 
\end{align}
Recall that
\begin{equation}
\boldsymbol{H}^{*}\boldsymbol{H}=\boldsymbol{M}^{*}\boldsymbol{A}^{*}\boldsymbol{A}\boldsymbol{M},
\end{equation}
where the entries of $\boldsymbol{M}$ are independently generated
with variance $\frac{1}{n}$. Since $\tilde{g}_{\epsilon}(x)$ is
a convex function in $x$, applying Proposition \ref{thm:GeneralTemplate}
yields the following results.
\begin{enumerate}
\item If $\sqrt{p}\boldsymbol{M}_{ij}$'s are bounded by $D$, then for
any $\beta>8\sqrt{\pi}$, we have
\begin{align*}
 & \frac{1}{p}\sum_{i=1}^{p}\frac{1}{\epsilon^{2}+\lambda_{i}\left(\boldsymbol{H}^{*}\boldsymbol{H}\right)}\\
 & \quad\leq\frac{3\sqrt{3}D\left\Vert \boldsymbol{A}\right\Vert }{8\epsilon^{3}}\frac{\beta}{p}+\frac{1}{p}\sum_{i=1}^{p}\mathbb{E}\left[g_{\frac{2\sqrt{2}}{3}\epsilon}\left(\sqrt{\lambda_{i}\left(\boldsymbol{H}^{*}\boldsymbol{H}\right)}\right)\right],
\end{align*}
and
\begin{align*}
 & \frac{1}{p}\sum_{i=1}^{p}\frac{1}{\epsilon^{2}+\lambda_{i}\left(\boldsymbol{H}^{*}\boldsymbol{H}\right)}\\
 & \quad\geq\left(\frac{2\sqrt{2}D\left\Vert \boldsymbol{A}\right\Vert }{3\sqrt{3}\epsilon^{3}}\right)\frac{\beta}{p}-\frac{1}{p}\sum_{i=1}^{p}\mathbb{E}\left[g_{\frac{3}{2\sqrt{2}}\epsilon}\left(\sqrt{\lambda_{i}\left(\boldsymbol{H}^{*}\boldsymbol{H}\right)}\right)\right]
\end{align*}
with probability exceeding $1-8\exp\left(-\frac{\beta^{2}}{16}\right)$.
\item If the measure of $\sqrt{p}\boldsymbol{M}_{ij}$ satisfies the LSI
with a uniformly bounded constant $c_{\mathrm{ls}}$, then for any
$\beta>0,$ one has
\begin{align}
 & \left|\frac{1}{p}\sum_{i=1}^{p}\frac{1}{\epsilon^{2}+\lambda_{i}\left(\boldsymbol{H}^{*}\boldsymbol{H}\right)}-\mathbb{E}\left[g_{\epsilon}\left(\sqrt{\lambda_{i}\left(\boldsymbol{H}^{*}\boldsymbol{H}\right)}\right)\right]\right|\nonumber \\
 & \quad\leq\left(\frac{3\sqrt{3}\sqrt{c_{\mathrm{ls}}}\left\Vert \boldsymbol{A}\right\Vert }{8\epsilon^{3}}\right)\frac{\beta}{p}\label{eq:LogSobolevDeviation-Corollary-1}
\end{align}
with probability exceeding $1-4\exp\left(-\frac{\beta^{2}}{2}\right)$.
Here, we have made use of the fact that $\left\Vert g_{\epsilon}\right\Vert _{\mathcal{L}}\leq\frac{3\sqrt{3}}{8\epsilon^{3}}$.
\item Suppose that $\sqrt{p}\boldsymbol{M}_{ij}$'s are independently drawn
from symmetric heavy-tailed distributions. Proposition \ref{thm:GeneralTemplate}
suggests that
\begin{align*}
 & \frac{1}{p}\sum_{i=1}^{p}\frac{1}{\epsilon^{2}+\lambda_{i}\left(\boldsymbol{H}^{*}\boldsymbol{H}\right)}\leq\frac{3\sqrt{3}\tau_{c}\sigma_{c}\left\Vert \boldsymbol{A}\right\Vert \sqrt{c(n)\log n}}{2\epsilon^{3}p}\\
 & \quad\quad\quad+\frac{1}{p}\sum_{i=1}^{p}\mathbb{E}\left[g_{\frac{2\sqrt{2}}{3}\epsilon}\left(\sqrt{\lambda_{i}\left(\tilde{\boldsymbol{M}}^{*}\boldsymbol{A}^{*}\boldsymbol{A}\tilde{\boldsymbol{M}}\right)}\right)\right],
\end{align*}
and
\begin{align*}
 & \frac{1}{p}\sum_{i=1}^{p}\frac{1}{\epsilon^{2}+\lambda_{i}\left(\boldsymbol{H}^{*}\boldsymbol{H}\right)}\geq\frac{8\sqrt{2}\tau_{c}\sigma_{c}\left\Vert \boldsymbol{A}\right\Vert \sqrt{c(n)\log n}}{3\sqrt{3}\epsilon^{3}p}\\
 & \quad\quad\quad-\frac{1}{p}\sum_{i=1}^{p}\mathbb{E}\left[g_{\frac{3}{2\sqrt{2}}\epsilon}\left(\sqrt{\lambda_{i}\left(\tilde{\boldsymbol{M}}^{*}\boldsymbol{A}^{*}\boldsymbol{A}\tilde{\boldsymbol{M}}\right)}\right)\right]
\end{align*}
with probability exceeding $1-\frac{10}{n^{c(n)}}$. Here, $\tilde{\boldsymbol{M}}$
is a truncated copy of $\boldsymbol{M}$ such that $\tilde{\boldsymbol{M}}_{ij}=\boldsymbol{M}_{ij}{\bf 1}_{\left\{ \left|\boldsymbol{H}_{ij}\right|\leq\tau_{\mathrm{c}}\right\} }$.
\end{enumerate}
This completes the proof. \end{proof}

\begin{proof}[{\bf Proof of Corollary \ref{cor:MMSE-iid-Gaussian}}]In
the moderate-to-high SNR regime (i.e. when all of $\lambda_{i}\left(\boldsymbol{H}^{*}\boldsymbol{H}\right)$
are large), one can apply the following simple bound
\begin{align*}
\small\frac{1}{p}\sum_{i=1}^{p}\left(\frac{1}{\lambda_{i}\left(\boldsymbol{H}^{*}\boldsymbol{H}\right)}-\frac{\epsilon^{2}}{\lambda_{i}^{2}\left(\boldsymbol{H}^{*}\boldsymbol{H}\right)}\right) & \small\leq\frac{1}{p}\sum_{i=1}^{p}\mathbb{E}\left[g_{\epsilon}\left(\sqrt{\lambda_{i}\left(\boldsymbol{H}^{*}\boldsymbol{H}\right)}\right)\right]\\
 & \small\leq\frac{1}{p}\sum_{i=1}^{p}\frac{1}{\lambda_{i}\left(\boldsymbol{H}^{*}\boldsymbol{H}\right)}.
\end{align*}
This immediately leads to
\begin{align}
\small\frac{1}{p}\mathrm{tr}\left(\left(\boldsymbol{H}^{*}\boldsymbol{H}\right)^{-1}\right)-\frac{\epsilon^{2}}{p}\mathrm{tr}\left(\left(\boldsymbol{H}^{*}\boldsymbol{H}\right)^{-2}\right) & \small\leq\frac{1}{p}\mathrm{tr}\left(\left(\epsilon^{2}\boldsymbol{I}_{p}+\boldsymbol{H}^{*}\boldsymbol{H}\right)^{-1}\right)\nonumber \\
 & \small\leq\frac{1}{p}\mathrm{tr}\left(\left(\boldsymbol{H}^{*}\boldsymbol{H}\right)^{-1}\right).\label{eq:TrIplusMMinv}
\end{align}

When $\boldsymbol{H}_{ij}\sim\mathcal{N}\left(0,\frac{1}{p}\right)$,
applying \cite[Theorem 2.2.8]{Fujikoshi2010} gives
\begin{align}
\mathrm{tr}\left(\left(\boldsymbol{H}^{*}\boldsymbol{H}\right)^{-1}\right) & =\frac{p^{2}}{\left(n-p-1\right)}=\frac{\alpha p}{\left(1-\alpha-\frac{1}{n}\right)}.\label{eq:TrMMinv}
\end{align}
Similarly, making use of \cite[Theorem 2.2.8]{Fujikoshi2010} leads
to
\begin{align}
\mathrm{tr}\left(\left(\boldsymbol{M}^{*}\boldsymbol{M}\right)^{-2}\right) & =p^{2}\frac{p\left(n-p-2\right)+p+p^{2}}{\left(n-p\right)\left(n-p-1\right)\left(n-p-3\right)}\nonumber \\
 & =\frac{\left(1-\frac{1}{n}\right)\alpha^{2}p}{\left(1-\alpha\right)\left(1-\alpha-\frac{1}{n}\right)\left(1-\alpha-\frac{3}{n}\right)}.\label{eq:TrMMinv2}
\end{align}
Substituting (\ref{eq:TrMMinv}) and (\ref{eq:TrMMinv2}) into (\ref{eq:TrIplusMMinv})
yields
\begin{align}
\frac{\alpha}{1-\alpha}-\frac{1}{\mathrm{SNR}}\frac{\alpha^{2}}{\left(1-\alpha-\frac{3}{n}\right)^{3}} & \leq\frac{1}{p}\mathrm{tr}\left(\left(\epsilon^{2}\boldsymbol{I}_{p}+\boldsymbol{H}^{*}\boldsymbol{H}\right)^{-1}\right)\nonumber \\
 & \leq\frac{\alpha}{1-\alpha-\frac{1}{n}},
\end{align}
thus concluding the proof.\end{proof}

\bibliographystyle{IEEEtran} \bibliographystyle{IEEEtran} \bibliographystyle{IEEEtran}
\bibliographystyle{IEEEtran}
\bibliography{bibfileMinimax}

\begin{IEEEbiographynophoto}{Yuxin Chen} (S'09) received the B.S. in Microelectronics with High Distinction from Tsinghua University in 2008, the M.S. in Electrical and Computer Engineering from the University of Texas at Austin in 2010, and the M.S. in Statistics from Stanford University in 2013. He is currently a Ph.D. candidate in the Department of Electrical Engineering at Stanford University. His research interests include information theory, compressed sensing, network science and high-dimensional statistics. \end{IEEEbiographynophoto} 

\begin{IEEEbiographynophoto}{Andrea J. Goldsmith} (S'90-M'93-SM'99-F'05)
is the Stephen Harris professor in the School of Engineering and a professor of Electrical Engineering at Stanford University. She was previously on the faculty of Electrical Engineering at Caltech. Her research interests are in information theory and communication theory, and their application to wireless communications and related fields. She co-founded and serves as Chief Scientist of Accelera, Inc., and previously co-founded and served as CTO of Quantenna Communications, Inc. She has also held industry positions at Maxim Technologies, Memorylink Corporation, and AT\&T Bell Laboratories. Dr. Goldsmith is a Fellow of the IEEE and of Stanford, and she has received several awards for her work, including the IEEE Communications Society and Information Theory Society joint paper award, the IEEE Communications Society Best Tutorial Paper Award, the National Academy of Engineering Gilbreth Lecture Award, the IEEE ComSoc Communications Theory Technical Achievement Award, the IEEE ComSoc Wireless Communications Technical Achievement Award, the Alfred P. Sloan Fellowship, and the Silicon Valley/San Jose Business Journal's Women of Influence Award. She is author of the book ``Wireless Communications'' and co-author of the books ``MIMO Wireless Communications'' and ``Principles of Cognitive Radio,'' all published by Cambridge University Press. She received the B.S., M.S. and Ph.D. degrees in Electrical Engineering from U.C. Berkeley.  

Dr. Goldsmith has served on the Steering Committee for the IEEE Transactions on Wireless Communications and as editor for the IEEE Transactions on Information Theory, the Journal on Foundations and Trends in Communications and Information Theory and in Networks, the IEEE Transactions on Communications, and the IEEE Wireless Communications Magazine. She participates actively in committees and conference organization for the IEEE Information Theory and Communications Societies and has served on the Board of Governors for both societies. She has also been a Distinguished Lecturer for both societies, served as President of the IEEE Information Theory Society in 2009, founded and chaired the student committee of the IEEE Information Theory society, and chaired the Emerging Technology Committee of the IEEE Communications Society. At Stanford she received the inaugural University Postdoc Mentoring Award, served as Chair of Stanfords Faculty Senate in 2009 and currently serves on its Faculty Senate and on its Budget Group.
\end{IEEEbiographynophoto} \begin{IEEEbiographynophoto}{Yonina C. Eldar} (S'98-M'02-SM'07-F'12) received the B.Sc. degree in physics and  the B.Sc. degree in electrical engineering both from Tel-Aviv University (TAU), Tel-Aviv,  Israel, in 1995 and 1996, respectively, and the Ph.D. degree in electrical engineering and computer science from the  Massachusetts Institute of Technology (MIT), Cambridge, in 2002.

 From January 2002 to July 2002, she was a Postdoctoral Fellow at the  Digital Signal Processing Group at MIT. She is currently a Professor in the Department of Electrical  Engineering at the Technion-Israel Institute of Technology, Haifa and holds the  The Edwards Chair in Engineering. She is  also a Research Affiliate with the Research Laboratory of Electronics at MIT  and a Visiting Professor at Stanford University, Stanford, CA. Her research interests are in the broad areas of  statistical signal processing, sampling theory and compressed sensing,  optimization methods, and their applications to biology and optics.

Dr. Eldar was in the program for outstanding students at TAU from 1992 to 1996. In 1998, she held the Rosenblith Fellowship for study in electrical engineering at MIT, and in 2000, she held an IBM Research Fellowship. From  2002 to 2005, she was a Horev Fellow of the Leaders in Science and  Technology program at the Technion and an Alon Fellow. In 2004, she was  awarded the Wolf Foundation Krill Prize for Excellence in Scientific  Research, in 2005 the Andre and Bella Meyer Lectureship, in 2007 the Henry  Taub Prize for Excellence in Research, in 2008 the Hershel Rich Innovation  Award, the Award for Women with Distinguished Contributions, the Muriel \& David Jacknow Award for Excellence in Teaching, and the Technion Outstanding  Lecture Award, in 2009 the Technion's Award for Excellence in Teaching, in 2010 the Michael  Bruno Memorial Award from the Rothschild Foundation, and in 2011 the Weizmann Prize for Exact Sciences.   In 2012 she was elected to the Young Israel Academy of Science and to the Israel Committee for Higher Education, and elected an IEEE Fellow.  In 2013 she received the Technion's Award for Excellence in Teaching, the Hershel Rich Innovation Award, and the IEEE Signal Processing Technical Achievement Award. She received several best paper awards together with her research students and colleagues.  She received several best paper awards together with her research students and colleagues. She is the Editor in Chief of Foundations and Trends in Signal Processing. In the past, she was a Signal Processing Society Distinguished Lecturer, member of the IEEE Signal Processing Theory and Methods and Bio Imaging Signal Processing technical committees, and served as an associate editor for the IEEE Transactions On Signal Processing, the EURASIP Journal of Signal Processing, the SIAM Journal on Matrix Analysis and Applications, and the SIAM Journal on Imaging Sciences.
\end{IEEEbiographynophoto} 
\end{document}